\documentclass[11pt,a4paper]{report}
\usepackage{amsfonts}
\usepackage{mathrsfs}
\usepackage{cite}
\usepackage{graphicx}
\usepackage{amsmath}
\usepackage{amssymb}
\usepackage{algorithm}
\usepackage{algorithmic}
\usepackage{acronym}

\begin{document}

\begin{titlepage}
\vspace*{10mm}
\begin{center}
{\Huge Resource Management for Device-\\ [2mm]to-Device Underlay Communication}\\[30mm]
{\Large Chen Xu$^1$, Lingyang Song$^1$, and Zhu Han$^2$}\\[5mm]
{$^1$School of Electrical Engineering and Computer
Science\\ Peking University, Beijing, China.\\
$^2$Electrical and Computer Engineering Department\\
University of Houston, Houston, USA.}
\end{center}
\end{titlepage}

\pagebreak \tableofcontents \pagebreak
\chapter{Introduction}\label{chap:intro}

\section{Overview of Device-to-Device Communication}
As one of next-generation wireless communication systems, Third
Generation Partnership Project (3GPP) Long Term Evolution (LTE) is
committed to provide technologies for high data rates and system
capacity. Further, LTE-Advanced (LTE-A) was defined to support new
components for LTE to meet higher communication demands
\cite{Doppler2009}. Local area services are considered as popular
issues to be improved, and by reusing spectrum resources local data
rates have been increased dramatically. However, the unlicensed
spectrum reuse may bring inconvenience for local service providers
to guarantee a stable controlled environment, e.g., ad hoc network
\cite{Basagni2004}, which is not in the control of the base station
(BS) or other central nodes. Hence, accessing to the licensed
spectrum has attracted much attention.

Device-to-Device (D2D) communication is a technology component for
LTE-A. The existing researches allow D2D as an underlay to the
cellular network to increase the spectral efficiency
\cite{Doppler2009,Yu_VTC}. In D2D communication, user equipments
(UEs) transmit data signals to each other over a direct link using
the cellular resources instead of through the BS, which differs from
femtocell \cite{Zhang2009} where users communicate with the help of
small low-power cellular base stations. D2D users communicate
directly while remaining controlled under the BS. Therefore, the
potential of improving spectral utilization has promoted much work
in recent years
\cite{Koskela2010,Doppler_ICC,Doppler2010,Min2011,Hakola2010,Yu_PIMRC},
which shows that D2D can improve system performances by reusing
cellular resources. As a result, D2D is expected to be a key feature
supported by next generation cellular networks.


Although D2D communication brings improvement in spectral efficiency
and makes large benefits on system capacity, it also causes
interference to the cellular network as a result of spectrum
sharing. Thus, an efficient interference coordination must be
formulated to guarantee a target performance level of the cellular
communication. There exists several work about the power control of
D2D UEs for restricting co-channel interference
\cite{Doppler2009,Yu_VTC,Yu_ICC,Xing2010}. The authors in
\cite{Janis_PIMRC} utilized MIMO transmission schemes to avoid
interference from cellular downlink to D2D receivers sharing the
same resources, which aims at guaranteeing D2D performances.
Interference management both from cellular to D2D communication and
from D2D to cellular networks are considered in \cite{Peng2009}. In
order to further improve the gain from intra-cell spectrum reuse,
properly pairing the cellular and D2D users for sharing the same
resources has been studied \cite{Janis_VTC,Zulhasnine2010}. The
authors in \cite{Zulhasnine2010} proposed an alternative greedy
heuristic algorithm to lessen interference to the primary cellular
networks using channel state information (CSI). The scheme is
easy-operated but cannot prevent signaling overhead. In
\cite{Xu2010}, the resource allocation scheme avoids the harmful
interference by tracking the near-far interference, identifies the
interfering cellular users, and makes the uplink (UL) frequency
bands efficiently used. Also, the target is to prevent interference
from cellular to D2D communication. In \cite{Yu_TWC}, the authors
provided analysis on optimum resource allocation and power control
between the cellular and D2D connections that share the same
resources for different resource sharing modes, and evaluated the
performance of the D2D underlay system in both a single cell
scenario and the Manhattan grid environment. Then, the schemes are
to further optimize the resource usage among users sharing the same
resources. Based on the aforementioned work, it indicates that by
proper resource management, D2D communication can effectively
improve the system throughput with the interference between cellular
networks and D2D transmissions being restricted. However, the
problem of allocating cellular resources to D2D transmissions is of
great complexity. Our works differ from all mentioned above in that
we consider some schemes to maximize the system sum rate by allowing
multiple pairs share one cellular user's spectrum resource.

The organization of this book is as follows. The rest of Chapter \ref{chap:intro}
gives the basic signal and interference model for D2D communication underlaying cellular
networks, and discusses performances of different communication modes for user equipments (UEs). Finally,
game theory is proposed to solve some resource management problems of D2D underlay cellular systems.
In Chapter \ref{chap:tech}, some physical-layer techniques in D2D underlay communications are investigated,
including power control for D2D users and beamforming for interference avoidance. Radio resource
management for maximizing system throughput is studied in Chapter \ref{chap:allocation} which includes resource allocation algorithm, analysis and performance results of the proposed algorithm. In Chapter \ref{chap:cross}, cross-layer optimization is the key issue. A joint scheduling and resource allocation scheme is proposed. An example of D2D communication improving energy efficiency is given, where an auction-based resource allocation scheme is applied and battery lifetime of UEs is explicitly considered as the optimization goal.


\section{Signal and Interference Model}\label{sec:signalModel}

We consider a single cell scenario as illustrated in Fig.
\ref{fig:SystemModel_Intro}. For simplicity, just one cellular user (UE1)
and one D2D pair (UE2 and UE3) which is in D2D mode are located in
the cell. Three users share the same radio resources at the same
time, thus co-channel interference should be considered. The
position of UE2 is fixed as long as the distance from BS to it is D.
The position of the other D2D user UE3 is described by a uniform
distribution inside a region at most L from UE2. As most traditional
cellular system, UE1 is free to be anywhere inside the cell,
following a uniform distribution. In the simulation, we update the
locations of three users in each loop.
\begin{figure}[h!]
\begin{center}
\includegraphics[width=3.5in]{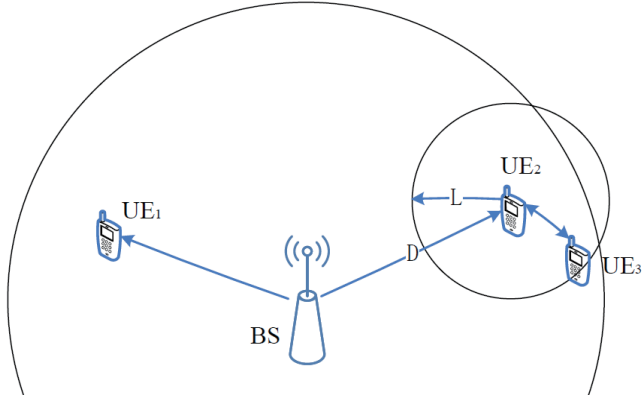}
\caption{Scenario of D2D underlay communication.}
\label{fig:SystemModel_Intro}
\end{center}
\end{figure}

According to Fig. \ref{fig:SystemModel_Intro}, three communicating users are in the system.
UE2 and UE3 are in D2D mode, and UE1 is cellular user. We set the
maximum distance between UE2 and UE3 is 25 meters. Actually, no more
than 100 meters distance between them can be effective. The existing
results only give a possible situation. Table \ref{tab:1} gives the main
simulation parameters.

The wireless propagation is modeled according to WINNER II channel models, and
D2D channel is based on office/indoor scenario while cellular
channel is based on urban macrocell scenario. Table \ref{tab:2} gives path
loss models. $d$ is the link distance in meters, and $n_{walls}$ is the amount of
walls penetrated in the link. $d_{BP}'  = {{4h_{BS}' h_{MS}' f_c }
\mathord{\left/ {\vphantom {{4h_{BS}' h_{MS}' f_c } c}} \right.
 \kern-\nulldelimiterspace} c}$, where $f_c$ is the centre frequency
in Hz, $c = 3.0 \times 10^8$m/s is the propagation velocity in free
space, and $h_{BS}'$ and $h_{MS}'$ are the effective antenna heights
at the BS and the MS, respectively. The effective antenna heights
$h_{BS}'$ and $h_{MS}'$ are computed as follows:
$h_{BS}'=h_{BS}-1.0m$, $h_{MS}'=h_{MS}-1.0$m where $h_{BS}$ and
$h_{MS}$ are the actual antenna heights, and the effective
environment height in urban environments is assumed to be equal to
1.0 m. The LOS probability is given in Table \ref{tab:3}.

\begin{table*}
\begin{center}
\caption{Main Simulation Parameters.}\label{tab:1}
\begin{tabular}{|l|l|}
\hline \bf{Parameter} & \bf{Value}\\ \hline Cellular & Isolated
cell,1-sector \\ \hline System area & User devices are distributed
within \\ & a range of 500m from the BS \\ \hline Noise spectral density
& -174dBm/Hz \\ \hline System bandwidth & 5MHz \\ \hline Noise
figure & 5dB at BS/9dB at device \\ \hline Antenna gains and
patterns & BS:14dBi; Device: Omnidirectional 0dBi \\ \hline Cluster
radius & 5m,10m,15m,20m,25m \\ \hline Transmit power & BS:46dBm;\\
& Device: 24dBm (without power control) \\ \hline
\end{tabular}
\end{center}
\end{table*}

\begin{table*}
\begin{center}
\caption{Path-loss Models \cite{WINNER}.}\label{tab:2}
\begin{tabular}{|l|l|l|}
\hline \bf{Scenario} & \bf{Path loss [dB]} & \bf {Shadow fading [dB] }\\
\hline D2D (LOS) & $18.7log_{10}(d)+46.8$ & 3 \\ \hline D2D  &
$36.8log_{10}(d)+43.8$ & 4\\(NLOS) & $+5(n_{walls}-1)$ & \\ \hline Cellular &
$26log_{10}(d)+39$ & 4 for $10m<d<d_{BP}$ \\ (LOS) &
$40.0log_{10}(d)+13.47$ & 6 for $d_{BP}'<d<5km$, \\& $-14.0log_{10}(h_{BS}')$ & $h_{BS}=25m$,
\\& $-14.0log_{10}(h_{MS}')$ &   $h_{MS}=1.5m$ \\
\hline Cellular  &
$(44.9-6.55log_{10}(h_{BS}))$ & 8 for $50m<d<5km$,\\(NLOS) & $\cdot log_{10}(d)+34.46$ & $h_{BS}=25m$,\\ & $+5.83log_{10}(h_{BS})$ & $h_{MS}=1.5m$ \\ \hline
D2D user- & $PL=PL_{b}+PL_{tw}+PL_{in}$ & 7 \\
Cellular user &
${PL_b  = PL_{B1} \left( {d_{out}  + d_{in} } \right)}$ &  \\
& ${PL_{tw}  = 14 + 15\left( {1 - \cos \left( \theta  \right)} \right)^2 }$ & \\
& ${PL_{in}  = 0.5d_{in} }$  &\\ \hline
\end{tabular}
\end{center}
\end{table*}

\begin{table*}
\begin{center}
\caption{LOS Probability \cite{WINNER}.}\label{tab:3}
\begin{tabular}{|l|l|}
\hline \bf{Scenario} & \bf{LOS probability}\\ \hline D2D & $P_{LOS}
= \left\{ {\begin{array}{*{20}c}
   1 & \hspace{-0.4cm}{,d \le 2.5}  \\
   {1 - 0.9\left( {1 - \left( {1.24 - 0.61\log _{10} \left( d \right)} \right)^3 } \right)^{{1 \mathord{\left/
 {\vphantom {1 3}} \right.
 \kern-\nulldelimiterspace} 3}} } & \hspace{-0.4cm}{,d > 2.5}  \\
\end{array}} \right.$ \\ \hline Cellular & $P_{LOS}  = \min \left( {{{18} \mathord{\left/
 {\vphantom {{18} d}} \right.
 \kern-\nulldelimiterspace} d},1} \right) \cdot \left( {1 - \exp \left( {{{ - d} \mathord{\left/
 {\vphantom {{ - d} {63}}} \right.
 \kern-\nulldelimiterspace} {63}}} \right)} \right) + \exp \left( {{{ - d} \mathord{\left/
 {\vphantom {{ - d} {63}}} \right.
 \kern-\nulldelimiterspace} {63}}} \right)$ \\ \hline
\end{tabular}
\end{center}
\end{table*}

Next, we investigate D2D and cellular SINR distribution with
power control. The LTE uplink open loop fraction power control
scheme (OFPC) is given as \cite{Ref:3GPP}
\begin{equation}
P = \min \left\{ {P_{\max } ,P_0  + 10 \cdot \log _{10} M + \alpha
\cdot L} \right\}
\end{equation}
The parameters for the power control scheme are given in Table \ref{fig:OFPC}.

\begin{table*}
\begin{center}
\caption{OFPC Parameters.}\label{fig:OFPC}
\begin{tabular}{|l|l|}
\hline \bf{Parameter} & \bf{Value}\\ \hline $P_{max}$ & 24dBm \\
\hline $P_0$ & -78dBm \\ \hline $\alpha$ & 0.8 \\ \hline L & Path
loss between two UEs in a pair. \\ \hline M & 1 \\ \hline
\end{tabular}
\end{center}
\end{table*}

In this scenario, the interference between D2D and cellular users
has been taken into account due to UL resource sharing. When the
distance between D2D and co-channel cellular users is not larger
than the maximum distance of D2D communication, the interference
channel can be based on indoor/office scenario. But when co-channel
interference comes from a farther location, D2D channel model is not
suitable. According to WINNER II channel models, we choose
indoor-to-outdoor/outdoor-to-indoor scenario to simulate
long-distance interference channel.

Table \ref{tab:2} also gives the interference channel model. $PL_{B1}$ is the
path loss of urban microcell scenario (see pp.44 in \cite{WINNER}
for parameter detail), $d_{out}$ is the distance between the outdoor
terminal and the point on the wall that is nearest to the indoor
terminal, $d_{in}$ is the distance from the wall to the indoor
terminal, $\theta$ is the angle between the outdoor path and the
normal of the wall. For simplicity, we can assume $\theta=0$ so that
$d_{out}+d_{in}=d$ in the simulation.


\begin{figure}[!ht]
\begin{center}
\includegraphics[height=3.1in]{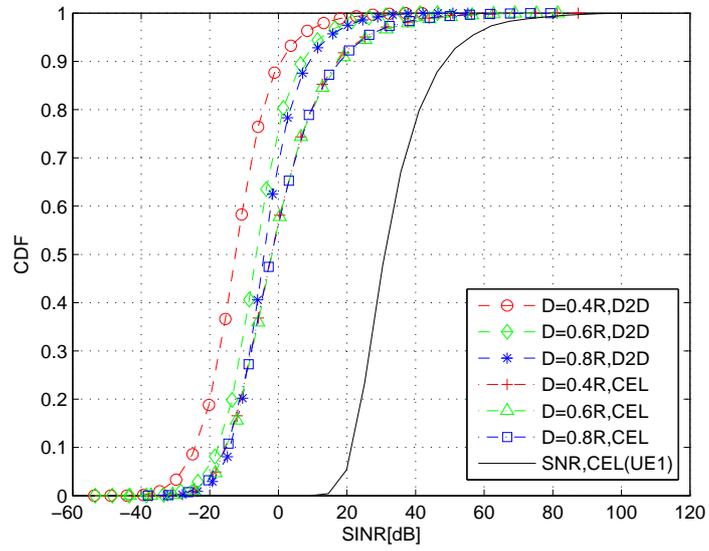}
\caption{SINR distribution of D2D underlay communication with L=25m
(DL).} \label{fig:SINR_DL_nonPC}
\end{center}
\end{figure}

\begin{figure}[!ht]
\begin{center}
\includegraphics[height=3.1in]{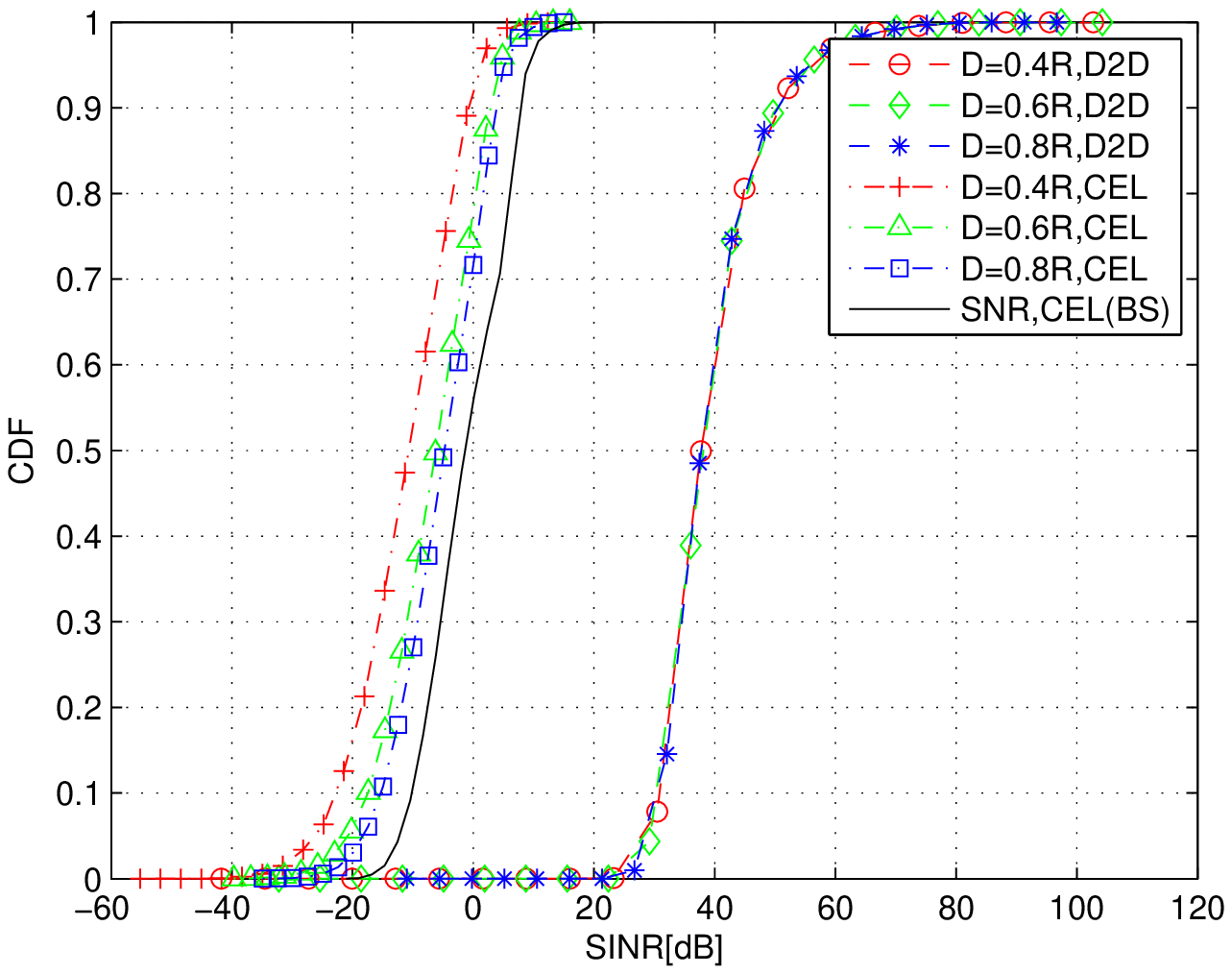}
\caption{SINR distribution of D2D underlay communication with L=25m
(UL).} \label{fig:SINR_UL_nonPC}
\end{center}
\end{figure}

\begin{figure}[!ht]
\begin{center}
\includegraphics[height=3.1in]{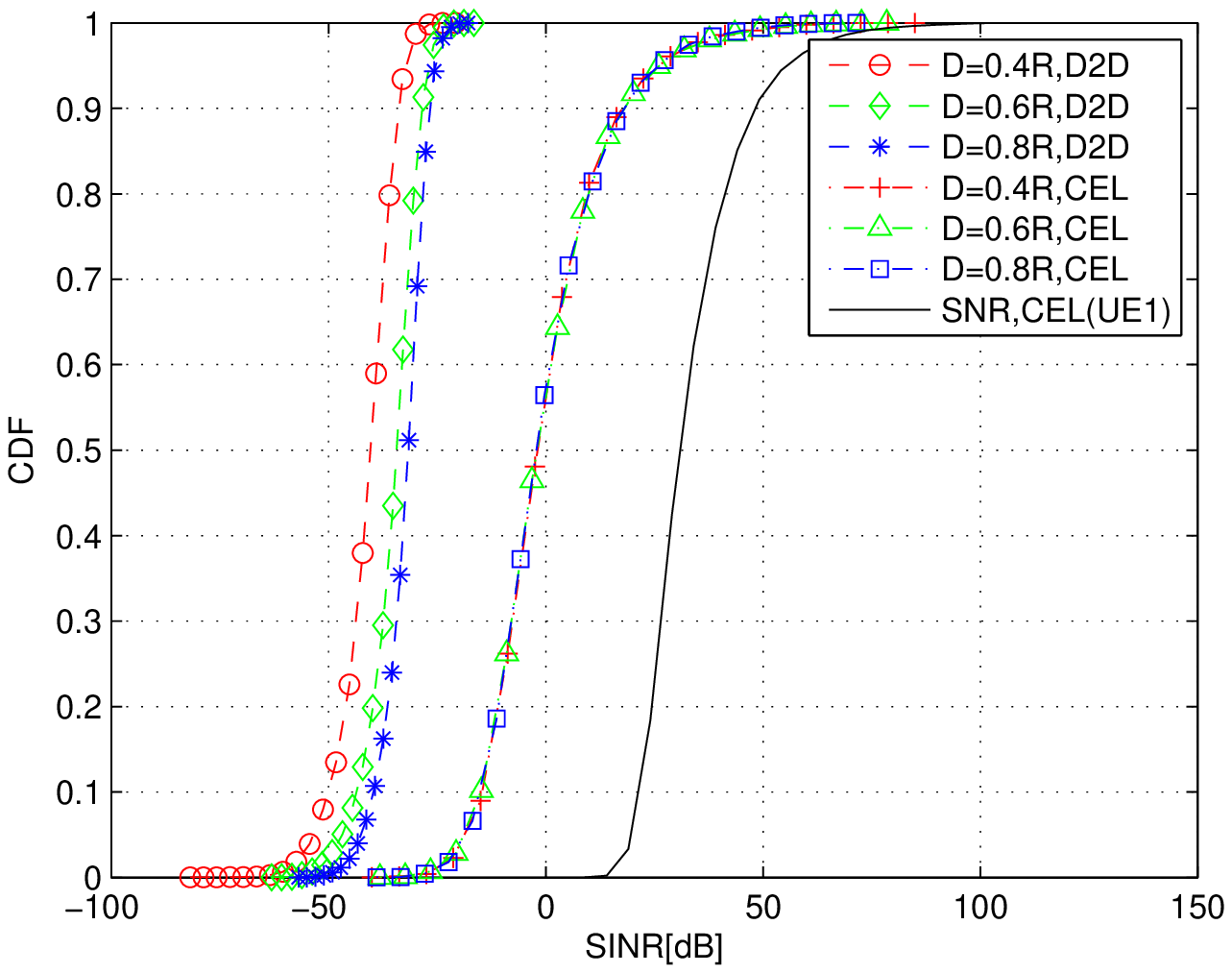}
\caption{SINR distribution of D2D underlay communication under PC
with L=25m (DL). } \label{fig:SINR_DL_PC}
\end{center}
\end{figure}

\begin{figure}[!ht]
\begin{center}
\includegraphics[height=3.1in]{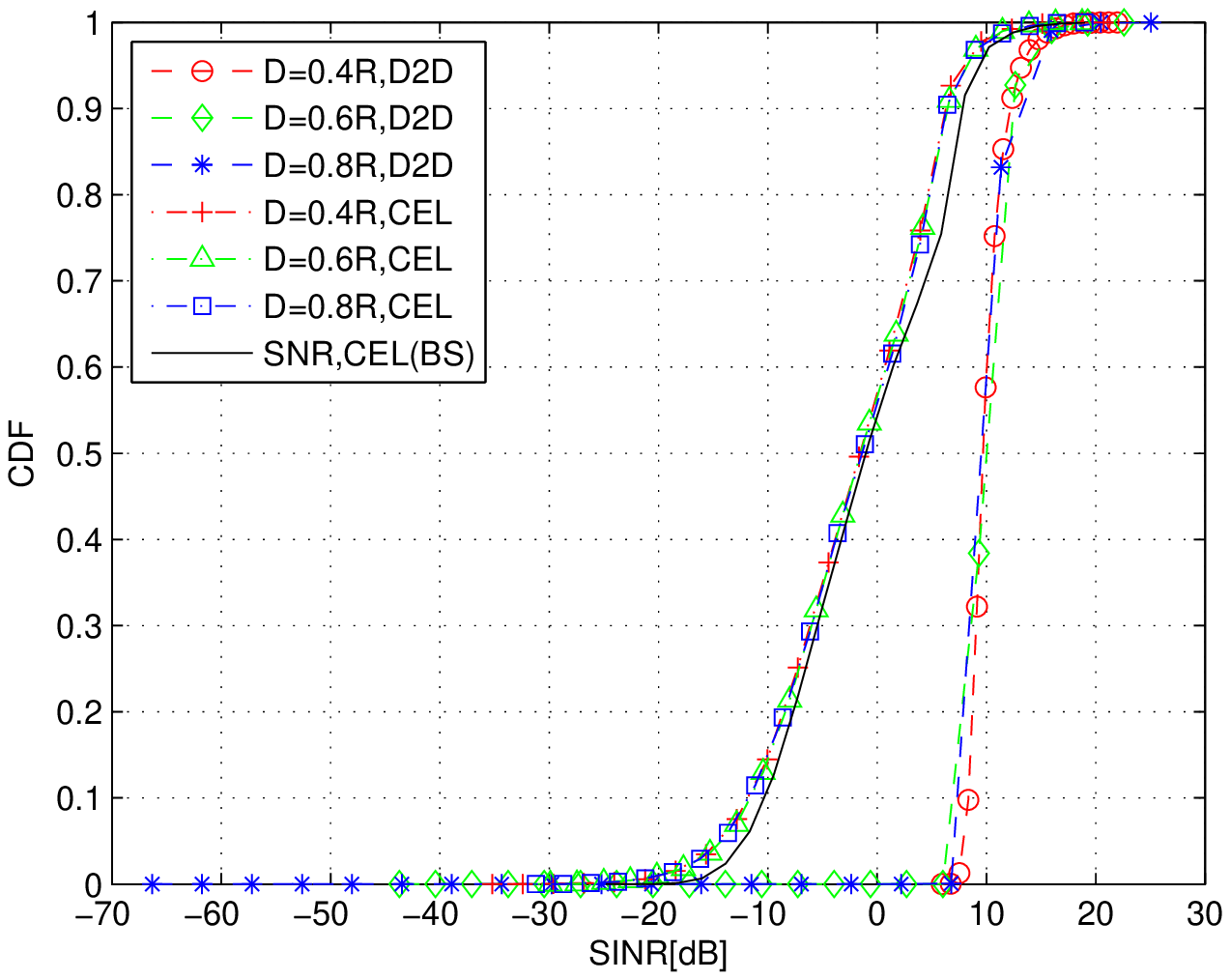}
\caption{SINR distribution of D2D underlay communication under PC
with L=25m (UL). } \label{fig:SINR_UL_PC}
\end{center}
\end{figure}

\begin{figure}[!ht]
\begin{center}
\includegraphics[height=3.1in]{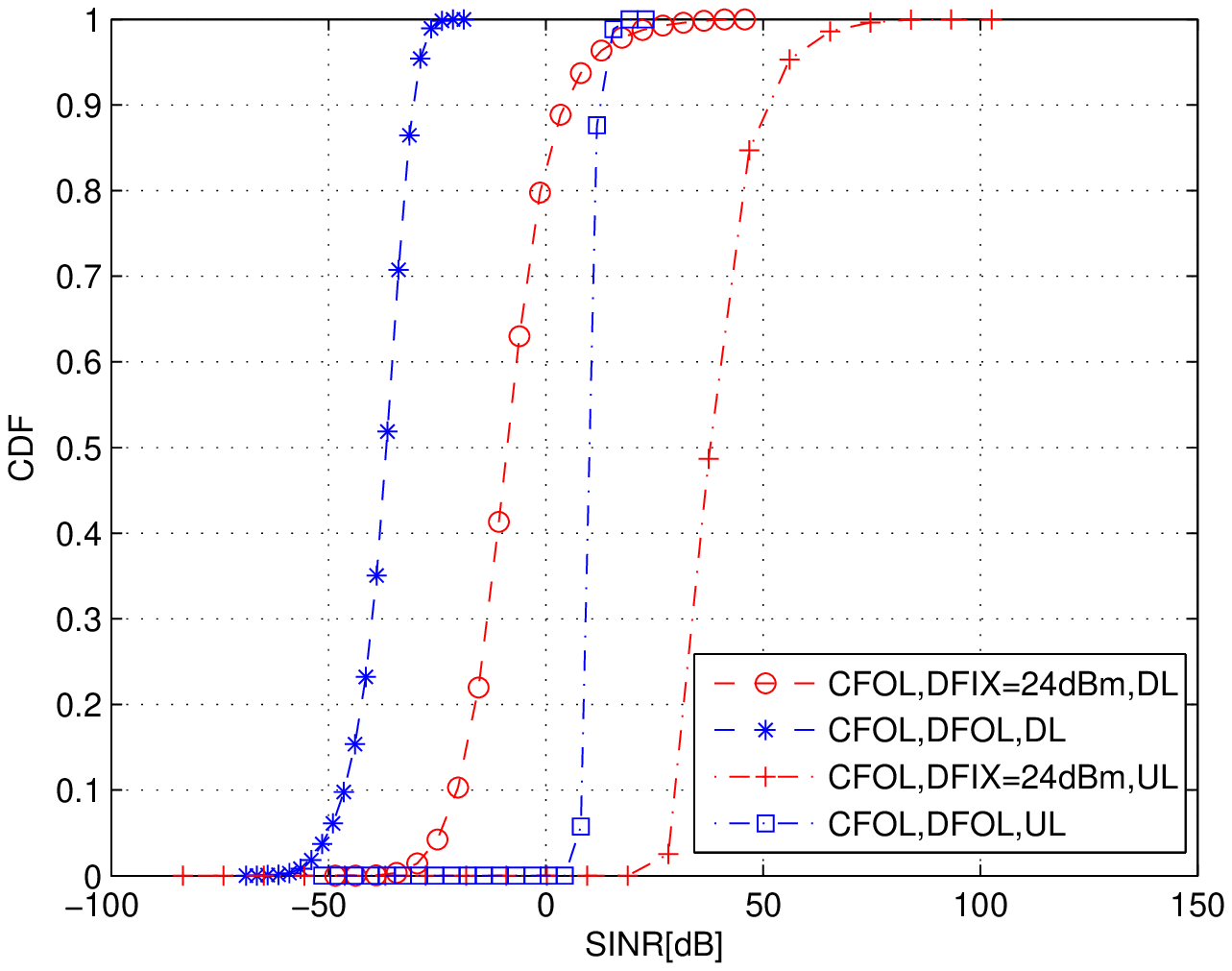}
\caption{SINR distribution of D2D underlay communication under PC
with L=25m, D=0.5R. CFOL: Cellular fraction open loop PC; DFIX: D2D
fixed power; DFOL: D2D fraction open loop PC.}
\label{fig:SINR_PC_nonPC}
\end{center}
\end{figure}

Consider a scenario of 19 cells, each of which is shown as Fig.
\ref{fig:SystemModel_Intro}. For simplicity, a model with one cellular
user and one D2D pair is considered. Update the locations of three
users in each simulation loop. Moreover, co-channel interference is
taken into account. Neighbor cell interference is also from D2D,
cellular users(UL), and BS(DL). Fig. \ref{fig:SINR_DL_nonPC} and
Fig. \ref{fig:SINR_UL_nonPC} give the SINR distribution of D2D
communication without power control (PC) mechanism in downlink (DL)
and uplink (UL) period, respectively.

When D2D users share cellular DL resources without PC: D2D SINR is
better when the pair is farther away from BS. Cellular (UE1) SINR is
less sensitive to the location of D2D users. UE1 SINR is better than
D2D SINR. When sharing DL resources, the interference to D2D is from
BS. The position of the pair has direct influence on the strength of
the interference. For cellular user, the strength of interference
depends not only on the position of D2D users, but also on the
position of cellular user. Since both are randomly distributed, the
position of D2D pair does not significantly affect the results. UE
transmit power is smaller than BS, thus interference from BS can be
higher than that from D2D. As a result, D2D SINR is clearly worse
than UE1.

When D2D users share cellular UL resources without PC: D2D SINR is
almost unchanged as the distance from BS to the pair changes. BS
SINR is better when D2D pair is farther away from BS. D2D SINR is
better than BS SINR. In UL resource sharing, the strength of D2D
interference depends not only on the position of D2D users, but also
on the position of cellular user. Since both are randomly
distributed, the position of D2D pair does not significantly affect
the results. For BS, the interference is from D2D users. The
position of the pair has direct influence on the strength of the
interference. UE3 is 0 $\sim$ 25m from UE2 and UE1 is 0 $\sim$ 500m away from BS,
which is likely to make D2D receive power larger than BS receive
power.

Fig. \ref{fig:SINR_DL_PC} and Fig. \ref{fig:SINR_UL_PC} are the SINR
distribution of D2D communication under PC in DL and UL period,
respectively. When D2D users share cellular DL resources with PC:
D2D SINR has decreased. UE1 SINR has increased. When D2D users share
cellular UL resources with PC: D2D SINR has decreased. BS SINR has
increased. Fig. \ref{fig:SINR_PC_nonPC} is the comparison of SINR
distribution between D2D communication with and without PC. By PC,
D2D SINR has decreased about 30dB. D2D SINR with PC gives smaller
dynamic range.

OFPC scheme limits the transmit power of D2D users, which leads to
D2D SINR degradation. Because D2D transmit power drops down, the
interference to cellular user and BS decreases. Due to the path loss
compensation in OFPC scheme, D2D SINR obtains a more concentrated
distribution.

\section{Mode Selection}\label{sec:modeSelection}
In Device-to-Device (D2D) underlay communication system, one of the
most challenging problems is to decide whether communicating devices
should use cellular or direct communication mode. In D2D mode the
data is directly transmitted to the receiver while cellular
communication mode requires the source device transmit to the base
station (BS) and then the destination device receives from the BS on
downlink (DL). Here we consider three different mode selection
criteria.
\begin{enumerate}
\item \emph{Cellular:} All devices are in cellular mode.
\item \emph{Force D2D:} D2D mode is selected always for all the
communicating devices.
\item \emph{PL D2D:} D2D mode is selected if any of the path losses
between source device and its serving BS, destination device and its
serving BS, is greater than the path loss between a source and a
destination in a pair.
\end{enumerate}

Here, we give the results in the multi-cell scenario of D2D
communication. Because it is closer to the practical applications.
The single-cell scenario is similar to this case, and will not be
repeated.
\begin{figure}[!ht]
\begin{center}
\includegraphics[height=3.2in]{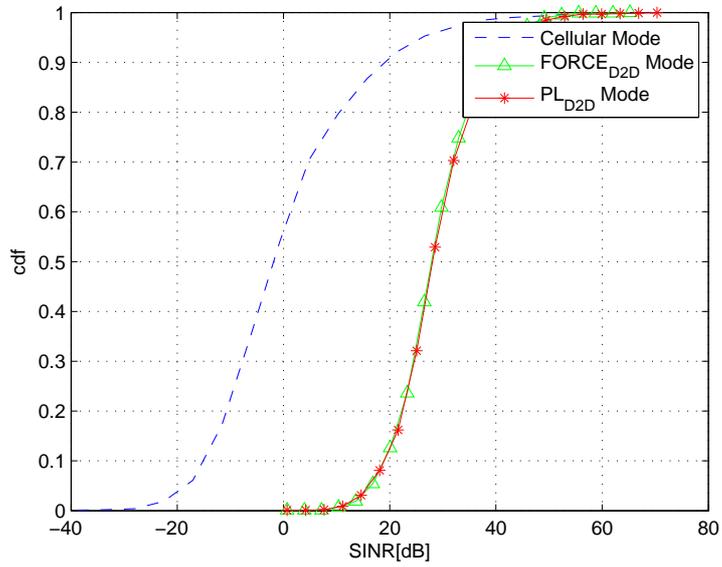}
\caption{SINR distribution of D2D underlay communication under
different mode with L=5m.} \label{fig:SINR_L=5}
\end{center}
\end{figure}
\begin{figure}[!ht]
\begin{center}
\includegraphics[height=3.2in]{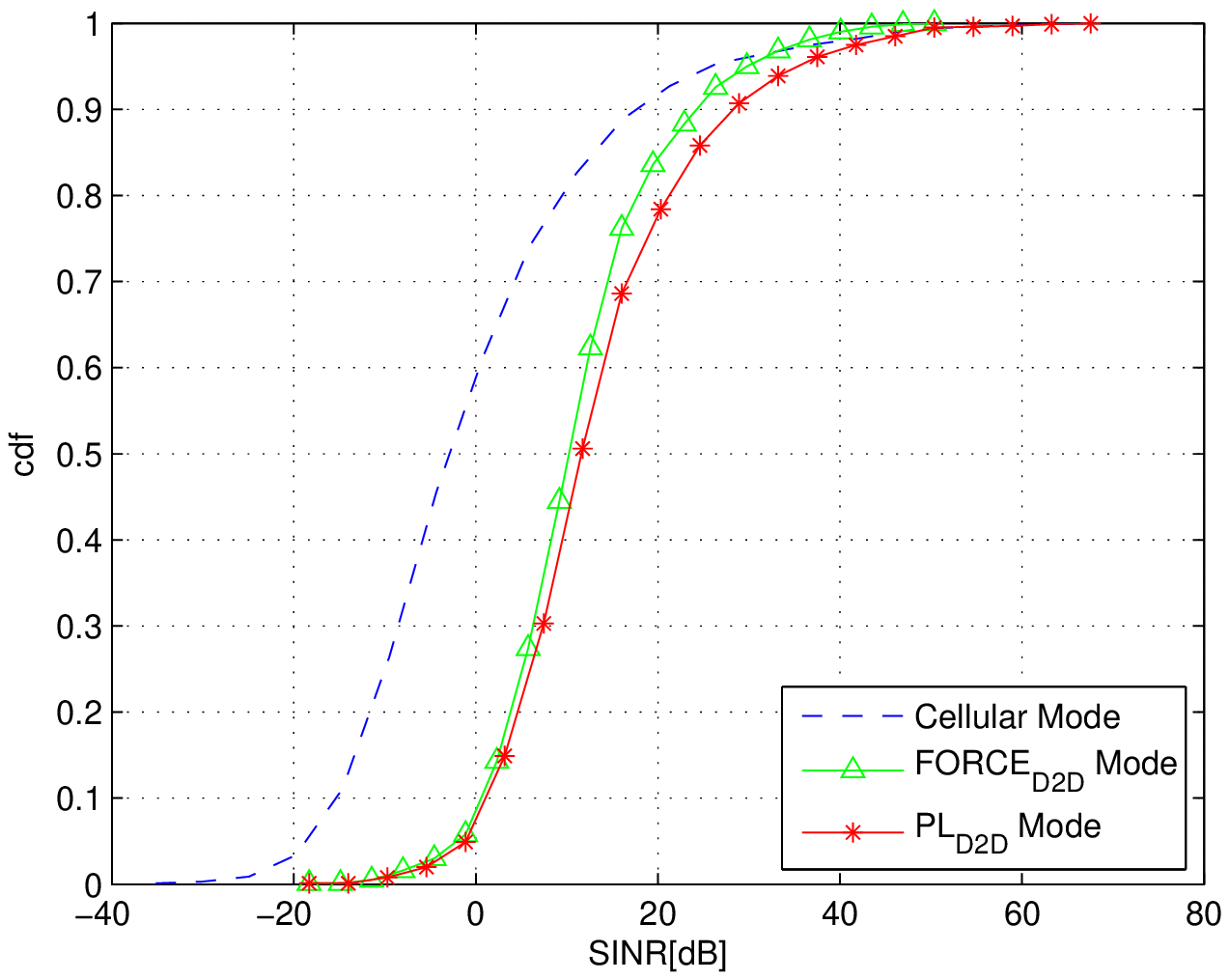}
\caption{SINR distribution of D2D underlay communication under
different mode with L=15m.}\label{fig:SINR_L=15}
\end{center}
\end{figure}
\begin{figure}[!ht]
\begin{center}
\includegraphics[height=3.2in]{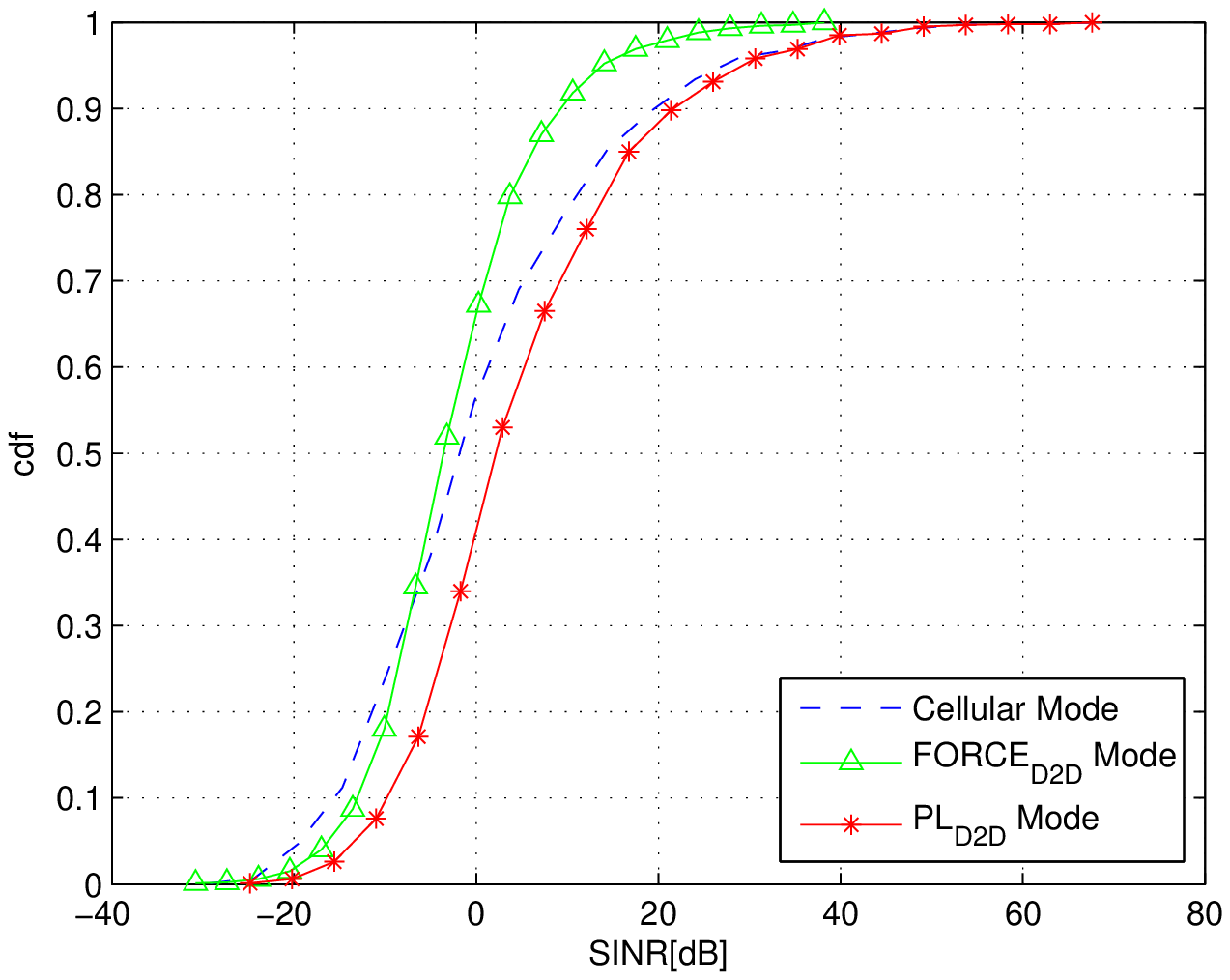}
\caption{SINR distribution of D2D underlay communication under
different mode with L=35m.} \label{fig:SINR_L=35}
\end{center}
\end{figure}
\begin{figure}[!ht]
\begin{center}
\includegraphics[height=3.2in]{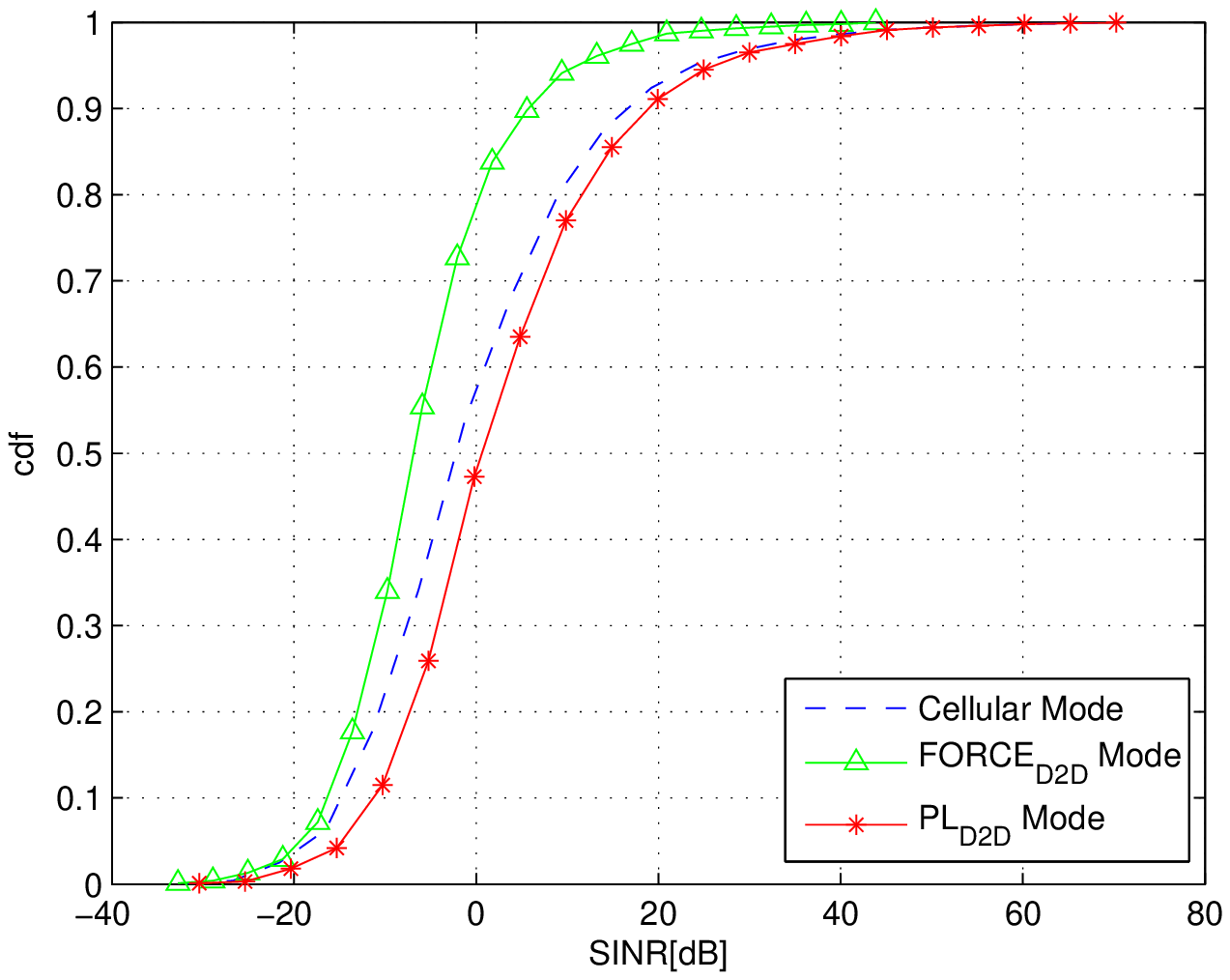}
\caption{SINR distribution of D2D underlay communication under
different mode with L=45m.} \label{fig:SINR_L=45}
\end{center}
\end{figure}

Consider a scenario of 19 cells with multiple users. The total
number of users in the center cell is $10^4$, in which 2000 users
are communicating.
At first, a communicating user is distributed uniformly in the cell.
Another one follows a uniform distribution inside a region at most L
from the first user, thus a link pair is formed. Accordingly, all
the users should be located in the cell.
Fig. \ref{fig:SINR_L=5}, \ref{fig:SINR_L=15},
\ref{fig:SINR_L=35}, \ref{fig:SINR_L=45} show the SINR distribution
of D2D communication under different mode with $L=5,15,35,45$m,
respectively. When the distance between D2D devices is small (L=5m):
FORCE, PL criterion give the same capacity distribution, and CEL
mode falls behind. As the maximum distance between D2D devices
increases to 15m: The capacity is still better under FORCE, PL
criterion, but it has decreased. PL mode gives higher capacity than
FORCE mode does. The maximum distance between D2D devices increases
to a larger value (L=35m): FORCE mode gives a large SINR dynamic
range and more than $50\%$ users SINR is lower than that under CEL
mode. PL mode is the best one. D2D devices are away from each other
at most 45m: FORCE mode gives the worst capacity distribution. CEL
and PL mode are similar.

As the maximum distance between users in a pair increases, the
performance of direct communication  degrades. There exists a
threshold value of L for deciding whether to use D2D communication.
PL mode is a method to solve the problem. Choose the better channel
condition from cellular and D2D communication can obtain an optimal
performance results.

\section{Introduction to Game Theory in D2D Communication}\label{sec:gameTheory}

Due to the interference caused by spectrum resource sharing between D2D and cellular
users, resource management becomes a key issue to be settled. Game theory offers a set of mathematical tools to study the complex
interactions among interdependent rational players and to predict
their choices of strategies \cite{Fudenberg1993}. In recent years, game theory
has emerged as a tool for the design of wireless communication networks. Therefore,
we employ game theory on resource scheduling and interference avoidance in D2D communication. In this section,
some necessary definitions are briefly introduced, and current researches
on wireless communications based on game theory are also mentioned.
Moreover, we give an overview of game theory method solving problems of D2D communication in this book.

The key elements of a game are the players, the actions, the payoffs
(utilities) and the information, together known as the rule of the game.
Players are the individuals making decisions, which can be denoted
as a set $\mathcal{M}=\{1,2,\ldots,M\}$. An action $a_i$ is a choice
that player $i$ makes. An action profile ${\bf a}=\{a_i | i \in \mathcal{M}\}$
is a set of all players' actions. In an auction, players are the bidders and actions
are the bids submitted by the bidders. Player $i$'s utility $u_i({\bf a})$ is a function of the action profile ${\bf a}$, and the utility describes how much gains the
player gets from the game for each possible action profile. In the games,
a player's utility equals his value for the action profile $v_i({\bf a})$
minus his payment $c_i({\bf a})$, i.e., $u_i({\bf a})=v_i({\bf a})-c_i({\bf a})$.
An important assumption of game theory is that all players are rational, i.e., they intend to choose actions to maximize their utilities. A player's information can be characterized by an information set, which tells what kind
of knowledge the player has at the decision instances. In order to maximize their utilities, the players would design their strategies, which are mappings from
one player's information sets to his actions.

A reasonable prediction of a game's outcome is an equilibrium, where
each player chooses a best strategy to maximize his utility. Among several
available equilibrium concepts, we mainly put emphasis on the Nash Equilibrium (NE).
In a static game, an NE is a strategy profile where no player can increase his utility by deviating unilaterally.

In the present
researches, game theory including a large number of different game
methods are used to analyze resource allocation problems, such as
power and wireless spectrum allocations in communication networks
\cite{Bae2008}, resource management in grids \cite{Das2005}, and
distributed resource coordination in mega-scale container terminal
\cite{Lau2007}. In \cite{Bae2008}, the authors proposed a sequential
auction for sharing the wireless resource, which is managed by a
spectrum broker that collects bids and allocates discrete resource
units using a sequential second-price auction.


A combinatorial auction model for resource management was introduced
in \cite{Das2005,Lau2007}. The combinatorial auction-based resource
allocation mechanism allows an agent (bidder) to place bids on
combinations of resources, called ``packages", rather than just
individual resource unit. Actually, the combinatorial auctions (CAs)
have been employed in a variety of industries for, e.g., truckload
transportation, airport arrival and departure slots, as well as
wireless communication services. The benchmark environment of
auction theory is the private value model, introduced by Vickrey
(1961), in which one bidder has a value for each package of items
and the value is not related to the private information of other
bidders \cite{Cramton2005}. Much of work has not recognized that
bidders care in complex ways about the items they compete. The CAs
motivate bidders to fully express their preferences, which is an
advantage in improving system efficiency and auction revenues. Up to
that point, our interest is to apply the CA game in solving
arbitrary D2D links reusing the same cellular frequency bands with
the purpose of optimizing the system capacity.

However, it exists a series of problems and challenges in CAs, such
as pricing and bidding rules, the winner determination problem (WDP)
which, as mentioned in the literature, leads to the NP-hard
allocation problem. Therefore, we focus on the evolution mechanisms
named iterative combinatorial auctions (I-CAs)
\cite{Pikovsky2008,Bichler2009}. In I-CAs, the bidders submit
multiple bids iteratively, and the auctioneer computes provisional
allocations and ask prices in each auction round.

In Chapter \ref{chap:allocation}, we study an effective spectrum resource allocation
for D2D communication as an underlay to further improve system
efficiency based on the I-CA. The whole system consists of the BS,
multiple cellular users that receive signals from the BS, and
multiple D2D pairs that communicate with respective receivers using
licensed spectrum resources. Considering that interference
minimization is a key point and multiple D2D pairs sharing the same
resources can bring large benefits on system capacity, we formulate
the problem as a reverse I-CA game. That means, the resources as the
bidders compete to obtain business, while D2D links as the goods or
services wait to be sold. By this way, the packages of D2D pairs are
auctioned off in each auction round. Furthermore, we investigate
some important properties of the proposed resource allocation
mechanism such as cheat-proof, convergence and price-monotonicity.
We reduce the computational complexity and apply the scheme to WINNER II channel
models \cite{WINNER} which contain a well-known indoor scenario. The
simulation results show that the auction algorithm leads to a good
performance on the system sum rate, and provides high system
efficiency while has lower complexity than the exhaustive search
allocation.

Prior works have considered little about time-domain scheduling of
D2D communication. In Section \ref{sec:scheduling}, we study joint scheduling, power
control and channel allocation for D2D communication using a game
theoretic approach. Note that if cellular and D2D UEs are simply
modeled as selfish players, the outcome is usually inefficient,
since cellular UEs do not want to share the channels with D2D UEs.
As D2D communication is an underlay to the primary cellular
networks, the concept of the Stackelberg game is well suited for the
system. The Stackelberg game is hierarchical, and has a leader and a
follower. The leader acts first, and then the follower observes the
leader's behavior and decides his strategy. The Stackelberg game has
been employed in cooperative networks \cite{Wang2007} and cognitive
radio networks \cite{Zhang_ACM,Bloem2007}. In our proposed
Stackelberg game, the cellular UEs are viewed as leaders and the D2D
UEs are followers. We group cellular and D2D UEs into
leader-follower pairs. The leader charges some fees for the follower
using the channel. We analyze the optimal price for the leader and
the optimal power for the follower. The strategies lead to a
Stackelberg equilibrium. Then, we propose a joint scheduling and
resource allocation algorithm. The leader-follower pairs form a
priority queue based on their utilities, and the system schedules
the D2D UEs according to their orders in the queue.

The booming wireless services are drawing more energy from UE
batteries. However, UEs are typically handheld equipments with
limited battery lifetime. Users have to constantly charge their
batteries. One major advantage of D2D communication is to decrease
UE transmit power consumption, and thus, extend the battery
lifetime. In Section \ref{sec:energy}, we further explore this issue. We consider
the energy consumption of UEs includes transmission energy and
circuit energy, and model the battery lifetime using the Peukert¡¯s
law \cite{Rao2003}. We formulate the problem as maximizing the
battery lifetime of D2D UEs subject to a rate constraint. The
problem is complicated to solve directly. Thus, we consider a
game-theoretic approach, where D2D UEs are viewed as players. The
players are self-interested, and they complete to maximize their own
battery lifetime. We construct the resource allocation game and
analyze the best response, Nash equilibrium and Pareto efficiency of
the game. The players may create externality when they selfishly
occupy radio resources, causing a decrease in the quality of
cellular communications. Thus, we modify the game by adding pricing
as penalty and propose the resource auction. Simulation results show
that the proposed algorithm has good performances in battery
lifetime.

\newpage

\chapter{Physical-layer Techniques}\label{chap:tech}
The demand of high-speed data services to wireless bandwidths grows
rapidly, which has promoted various technology development. D2D has
been proposed to be an underlay to the cellular network aiming to
improve spectrum efficiency and system sum rate. For its potential
to resource reuse and system capacity improvement, D2D communication
is considered to be a key feature of the next generation wireless
network, and attracts much attention. But it is worth noting that
D2D may cause undesirable interference to the primary cellular users
due to the spectrum sharing. This chapter focuses on some physical-layer
techniques for resource management and interference avoidance of D2D
communication underlaying cellular networks.

A simple power control scheme for D2D communication
is proposed in Section \ref{sec:power}, and in Section \ref{sec:beam}, we investigate a joint
beamforming and power control mechanism to avoid interference between cellular and D2D links as well as
maximize system throughput.

\section{Power Control}\label{sec:power}
In this section, a threshold-based
power control scheme for D2D links is proposed to improve the
performance of D2D underlay systems. Both interference management
and power saving are considered in this scheme.

\subsection{Power Control Scheme}

Fig. \ref{fig:thre_model} gives the radio resource sharing of D2D
and cellular links. We can see that the co-channel interference
cannot be ignored. Uplink (UL) resource reuse has better performance
on D2D channel rate and operability compared to downlink resource
reuse. However, to efficiently sharing UL spectrum resource, it is
necessary to mitigate the interference from D2D transmitter to the
BS; moreover, saving power as much as possible while satisfying a
reliable level of performance for D2D users will promote the system
efficiency.
\begin{figure}[!ht]
\centering
\includegraphics[width=4.8in]{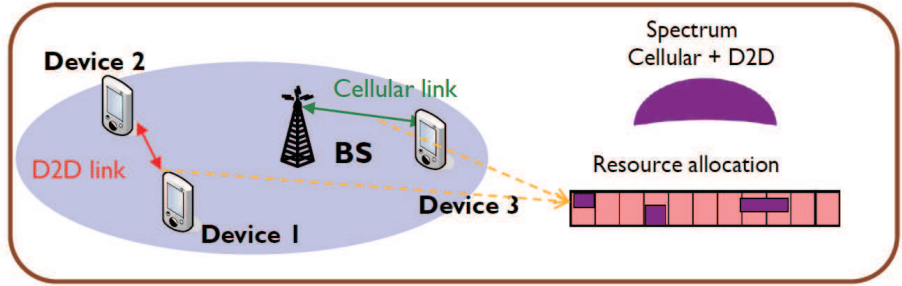}
\caption{A scenario of resource reuse between D2D and cellular
links.} \label{fig:thre_model}
\end{figure}

Fig. \ref{fig:thre_interf} shows the interference scenario of UL
resource sharing. We can see that D2D transmitter UE1 causes
interference to the BS, while cellular user UE3 causes interference
to D2D receiver UE2. Existing work \cite{Doppler2009,Yu_VTC,Yu_ICC}
propose some power control schemes for D2D transmissions.
\cite{Doppler2009} employs the eNB to control the maximum transmit
power of D2D transmitters which achieves the purpose of limiting the
co-channel interference. In \cite{Yu_VTC}, D2D power is controlled
by the eNB according to statistical results at the eNB.
\cite{Yu_ICC} proposes a greedy sum rate maximization optimization
with full CSI assumption. However, these schemes do not consider the
practical communication constraints and detailed mechanism design.
The patent \cite{Patent_1} adjusts D2D transmit power according to
HARQ feedback from the eNB to cellular UE, which is unreliable in
judging the interference status by HARQ monitoring.

\begin{figure}[!ht]
\centering
\includegraphics[width=2.8in]{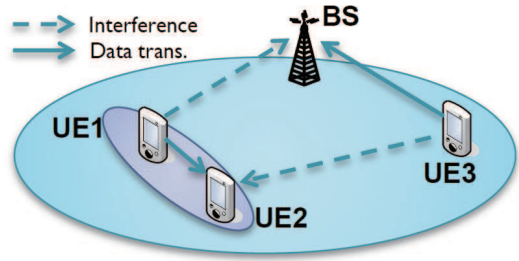}
\caption{The interference scenario of D2D and cellular links under
uplink resource sharing} \label{fig:thre_interf}
\end{figure}

The patent \cite{Patent_2} takes the following scheme: the eNB makes
measurements of interference from D2D transmitter to cellular link,
computes appropriate backoff or boost values, and sends a power
control command to D2D transmitter. Although this scheme control the
interference relatively effectively, the quality of D2D link is not
counted for the system, which may cause performance loss. Further,
it needs centralized scheduling as the eNB should control both D2D
and cellular communication for interference measurement, which leads
to large system overhead. Based on the above, considering both
cellular and D2D link performances, this section proposes a power
control method which can be utilized with distributed scheduling
under uplink resource sharing.

The key ideas of this scheme are as follows.
\begin{enumerate}
\item The BS has no direct control on the D2D link but notifies an
interference margin threshold to D2D transmitter.
\item The BS feeds back the CSI of D2D transmitter UL channel
(not needed for TDD system as the CSI can be obtained by channel
symmetry).
\item D2D transmit power is calculated by D2D transmitter itself
with knowledge of the CSI and the interference margin threshold.
\item D2D transmitter can freely decide whether to transmit according
to allowed transmit power and D2D link status.
\end{enumerate}

The key benefits of the above scheme: the system satisfies D2D link
quality while guarantees cellular link away from destructive
interference, which further improves system performance. In
addition, the scheme has better scalability as distributed
characteristics. It could both restrain interference and guarantee
the feasibility of D2D connection.

A simplified process of signaling interactions is shown in Fig.
\ref{fig:thre_signal}, and Fig. \ref{fig:thre_process} lists the
implementation process of the proposed D2D power control.

\begin{figure}[!ht]
\centering
\includegraphics[width=3.8in]{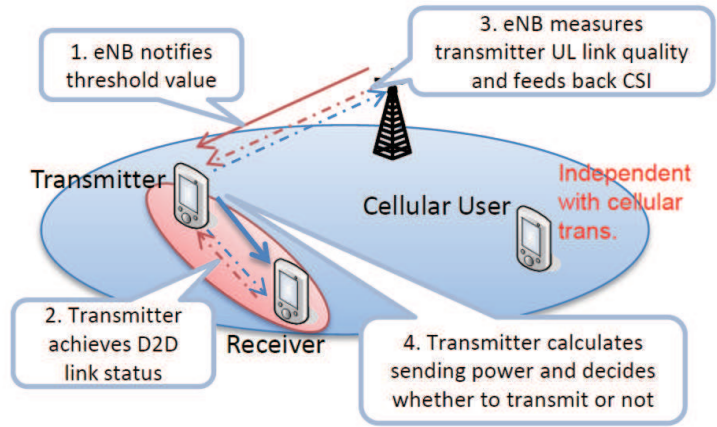}
\caption{Signaling interaction of the threshold based power control
scheme.} \label{fig:thre_signal}
\end{figure}

\begin{figure}[!ht]
\centering
\includegraphics[width=3.8in]{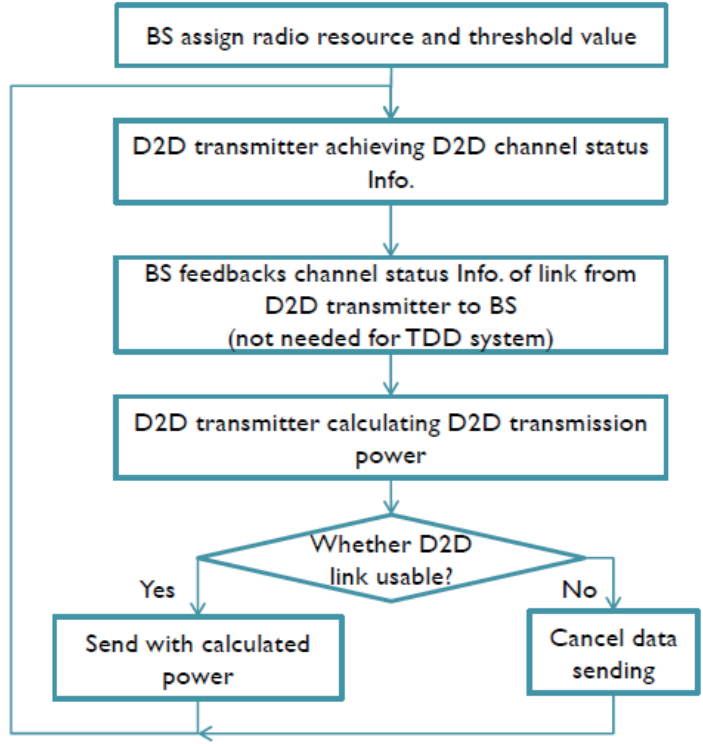}
\caption{Implementation process of the threshold based power control
scheme.} \label{fig:thre_process}
\end{figure}

\subsection{Threshold Based Power Calculation}

In this scheme, D2D transmitter can calculate the permitted transmit
power bound according to the obtained CSI and the interference
margin threshold. For the uplink resource sharing, we assume the
minimum SINR at the BS is required for $\beta_B$, and the
interference margin threshold is represented as $\kappa$. Then, the
corresponding SNR (without the co-channel interference) and SINR at
the BS are expressed as
\begin{equation}\label{eq:thre_SNR}
SNR_B  = \frac{{P_c L_{cB}^{ - 1} h_{cB} }}{{\sigma ^2 }} \ge \kappa
\beta _B,
\end{equation}
\begin{equation}\label{eq:thre_SINR}
SINR_B  = \frac{{P_c L_{cB}^{ - 1} h_{cB} }}{{P_d L_{dB}^{ - 1}
h_{dB}  + \sigma ^2 }} \ge \beta _B,
\end{equation}
respectively. Here, $P_c$ and $P_d$ denote the transmit signal power
of cellular and D2D user, respectively. $L_{cB}$ and $h_{cB}$
represent the path-loss and the channel gain between the BS and
cellular user. $L_{dB}$ and $h_{dB}$ are the path-loss and the
channel gain between D2D transmitter and the BS. $\sigma^2$ is the
AWGN noise power.

According to (\ref{eq:thre_SNR}), the transmit power of cellular
user satisfies
\begin{equation}
P_c  \ge \kappa \beta _B L_{cB} h_{cB}^{ - 1} \sigma ^2  = P_c^{\min
}.
\end{equation}
If the transmit power of cellular user takes the minimum value
$P_c^{min}$, (\ref{eq:thre_SINR}) can be transformed into
\begin{equation}
P_d^{\max }  = \left( {\kappa  - 1} \right)L_{dB} h_{dB}^{ - 1}
\sigma ^2  \ge P_d.
\end{equation}
Thus, on the premise of satisfying the interference margin $\kappa$,
the maximum transmit power of D2D is $P_d^{max}$.

We consider that the channel gain in the practical system may be
estimated, then we have
\begin{equation}
P_d^{\max ^* }  = \left( {\kappa  - 1} \right)L_{dB} \sigma ^2
\left( {\hat h_{dB} } \right)^{ - 1}.
\end{equation}
Moreover, we consider the received power of D2D user $
P_{d_r } = P_d^{\max ^* } L_{dd}^{ - 1} h_{dd}$, where $L_{dd}$ and
$h_{dd}$ are the path-loss and the channel gain between D2D
transmitter and receiver. Thus, we fine tune the transmit power as
$P_d^*  = P_d^{\max ^* } \left( {\hat h_{dd} } \right)^{ - 1}$, and
we obtain
\begin{equation}
P_d^*  = \left( {\kappa  - 1} \right)L_{dB} \sigma ^2 \left( {\hat
h_{dB} } \right)^{ - 1} \left( {\hat h_{dd} } \right)^{ - 1},
\end{equation}
where $\hat h_{dB}$ and $\hat h_{dd}$ are the channel gain
estimation values.

For the performance of D2D links, it requires to satisfy the minimum
SINR of receiver $\beta_d$, i.e.,
\begin{equation}
\frac{{P_d L_{dd}^{ - 1} h_{dd} }}{{P_c L_{cd}^{ - 1} h_{cd}  +
\sigma ^2 }} \ge \beta _d,
\end{equation}
where $L_{cd}$ and $h_cd$ represent the path-loss and the channel
gain between cellular user and D2D receiver. So we have
\begin{equation}
P_d  \ge L_{dd} h_{dd}^{ - 1} \left( {P_c L_{cd}^{ - 1} h_{cd}  +
\sigma ^2 } \right)\beta _d  = P_d^{\min }.
\end{equation}
Here, $P_d^{min}$ is the minimum transmit power to guarantee the
performance of D2D links.


Next, we give some simulation results and related
analysis. Power saving and energy efficiency are two sides that been
investigated. We define the power saving as
\begin{equation}
\alpha  = \frac{{\left( {P_c^{OL}  + P_d^{fixed} } \right) - \left(
{P_c^{OL}  + P_d^* } \right)}}{{P_c^{OL}  + P_d^{fixed} }},
\end{equation}
where $P_d^{fixed}$ is the fixed transmit power of D2D users,
$P_c^{OL}$ is the transmit power of cellular users by LTE uplink
open loop fraction power control scheme, and $P_d^*$ is the transmit
power of D2D users by the proposed power control scheme while
satisfying the constraints $P_d^{\min }  \le P_d^*  \le 23dBm$. That
is, power saving factor $\alpha$ represents the percentage of saving
power by the power control scheme compared to the fixed power
scheme. In the simulation, energy efficiency can be calculated by
\begin{equation}
\eta  = \frac{{\log _2 \left( {1 + SINR_c^{OL} } \right) + \log _2
\left( {1 + SINR_d^* } \right)}}{{P_c^{OL}  + P_d^* }}.
\end{equation}
Here, $SINR_c^{OL}$ and $SINR_d^*$ denote the SINR at the BS and D2D
receiver, respectively. The simulation parameters are listed in Table \ref{tab:thre_parameter}.

\begin{table*}
\begin{center}
\caption{Main Simulation Parameters}\label{tab:thre_parameter}
\begin{tabular}{|l|l|}
\hline \bf{Parameter} & \bf{Value}\\ \hline Cellular & Isolated cell, 1-sector  \\
\hline System area & User device are distributed in a \\ & hexagonal cell with
500m radius. \\ \hline Noise spectral density & -174dBm/Hz \\
\hline Sub-carrier bandwidth & 15kHz \\
\hline Noise figure & 5dB at BS/9dB at device \\
\hline Antenna gains and patterns & BS: 14dBi \\ & Device:Omnidirectional 0dBi \\
\hline Cluster radius & 150m \\
\hline Minimum SINR & BS: 10dB \\ & Device: 5dB \\
\hline Number of cellular users (channels) & 1/1:8 \\
\hline Interference margin & 3dB \\
\hline Transmit power bound & Max: 23dBm \\
\hline Detecting signal power & 5dBm \\
\hline Channel model & WINNER II \\ \hline
\end{tabular}
\end{center}
\end{table*}

\begin{figure}[!ht]
\centering
\includegraphics[height=3.3in]{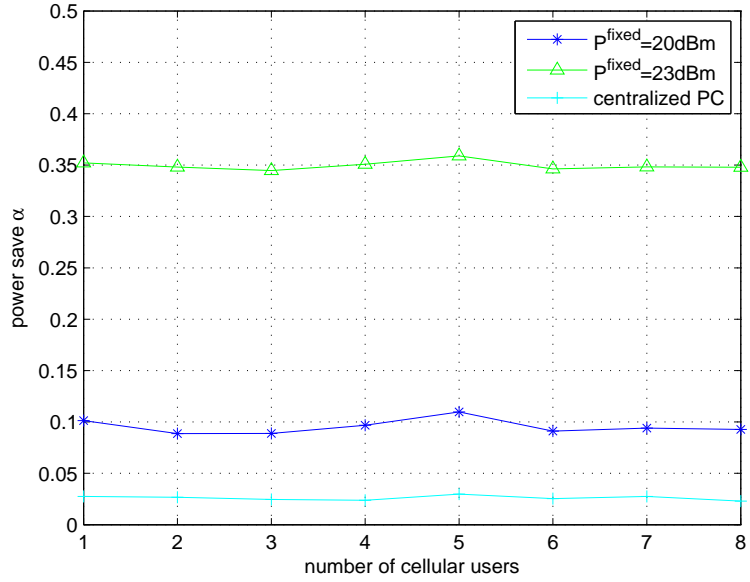}
\caption{Power savings under statistical channel estimation with
different fixed power (UL).} \label{fig:thre_1}
\end{figure}

\begin{figure}[!ht]
\centering
\includegraphics[height=3.3in]{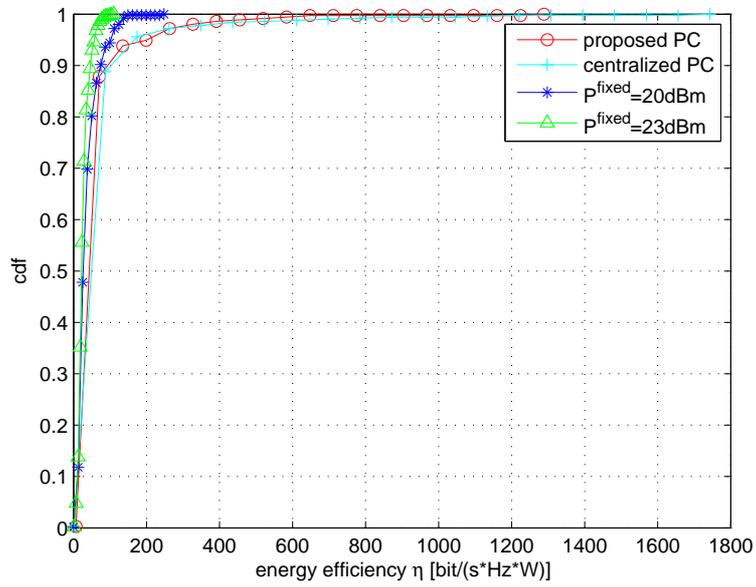}
\caption{Energy efficiency under statistical channel estimation with
different transmit power (UL).} \label{fig:thre_2}
\end{figure}

\begin{figure}[!ht]
\centering
\includegraphics[height=3.3in]{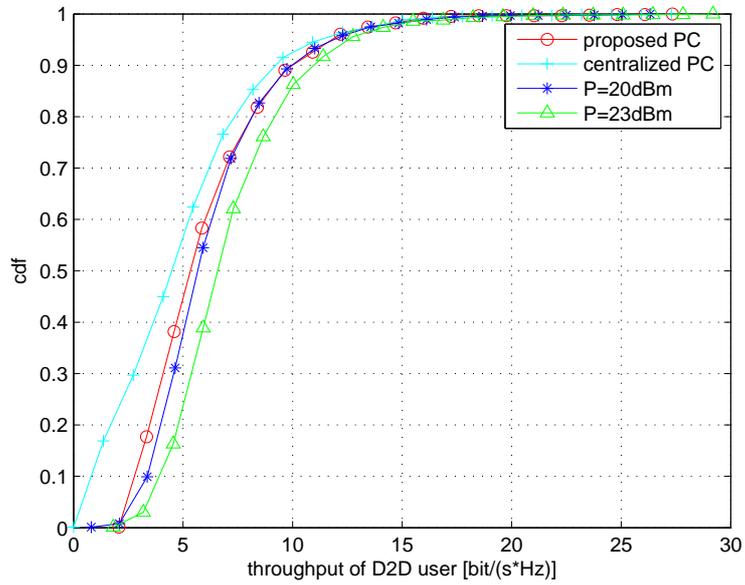}
\caption{Throughput distribution of D2D communication under
statistical channel estimation.} \label{fig:thre_3}
\end{figure}

\begin{figure}[!ht]
\centering
\includegraphics[height=3.3in]{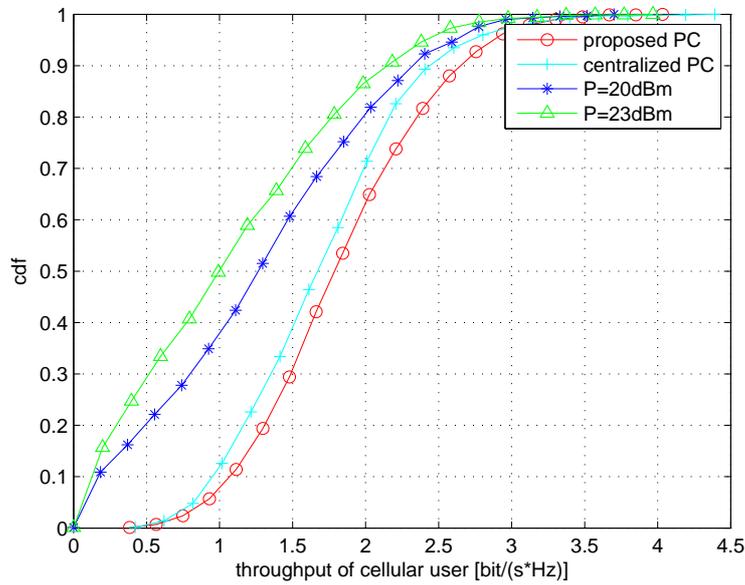}
\caption{Throughput distribution of cellular communication under
statistical channel estimation.} \label{fig:thre_4}
\end{figure}

\begin{figure}[!ht]
\centering
\includegraphics[height=3.3in]{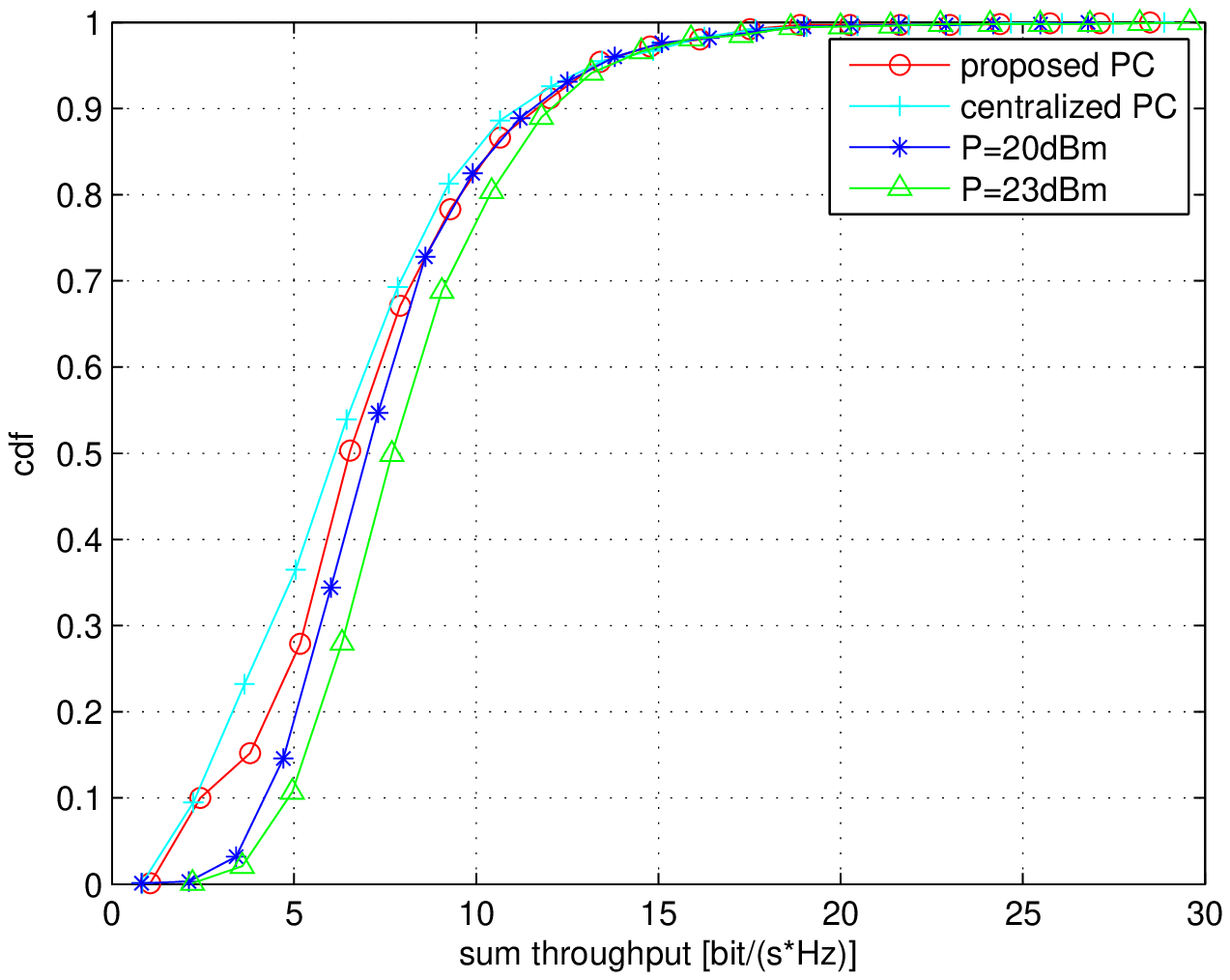}
\caption{Sum throughput distribution of cellular and D2D
communications under statistical channel estimation.}
\label{fig:thre_5}
\end{figure}

Fig. \ref{fig:thre_1} and Fig. \ref{fig:thre_2} give the percentage
of power saving and the distribution of energy efficiency under
statistical channel estimation, respectively. The power control
scheme saves about $35\%$ power than the fixed power 23dBm. Besides,
the power control scheme statistically performs better than fixed
power case in energy efficiency, and the performance is similar to
the centralized scheme in the patent \cite{Patent_2}.

Fig. \ref{fig:thre_3} $\sim$  Fig. \ref{fig:thre_5} show the
distribution of D2D, cellular and system throughput under
statistical channel estimation. Because of the D2D power control,
the throughput of D2D link is lower than that under the fixed
maximum power 23dBm, but higher than that under the centralized PC
scheme due to the D2D link performance guarantee. The sum throughput
is also appreciable as excessive interference to cellular links is
avoided.


This section has proposed a distributed threshold-based power
control scheme that guarantees the feasibility of D2D connection,
and at the same time limits cellular SINR degradation. The BS has no
direct control on the D2D link but notifies an interference margin
threshold to D2D transmitter, and the value can be tuned to meet
corresponding SINR requirements. Power is calculated by D2D
transmitter itself, which makes the operation flexible and
convenient, improving the system efficiency.

\section{Beamforming}\label{sec:beam}

Same frequency-time resources could be shared by cellular and D2D
links to enhance the system capacity, but co-channel interference
exists. During downlink (DL), D2D links receive more interference
from the BS. Interference management is necessary to optimize the
system performance. On one side, we hope to achieve a reliable level
of performance for both cellular and D2D users. On the other side,
to maximize the system throughput is our objective. In this section,
we investigate a joint beamforming and power control method to
reduce the interference and further improve the system performance.

\subsection{Joint Beamforming and Power Control Scheme}\label{subsec:beam}

In this scheme,
we consider a single cell scenario, and only one cellular user and
one D2D pair in the model. We assume that the BS is equipped with
multiple antennas while UEs are respectively with a single antenna.
Fig. \ref{fig:beam_model} gives the system model, where the solid
lines indicate the data transmissions and the dotted lines indicate
the interference links. Moreover, we assume the channel responses
are known by the BS, and the SINR minimum threshold of both cellular
and D2D users are set by the BS.
\begin{figure}[!ht]
\centering
\includegraphics[width=2.8in]{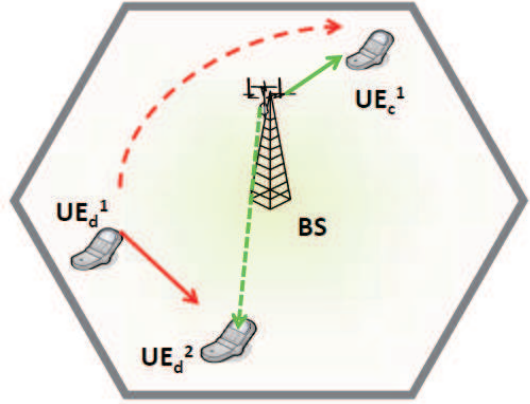}
\caption{System model of D2D communication underlaying cellular
networks with downlink resource sharing.} \label{fig:beam_model}
\end{figure}

Fig. \ref{fig:beam_model2} gives an example of beamforming with two
antennas equipped at the BS. The channel response matrix can be
expressed as
\begin{equation}
{\bf{H}} = \left( {\begin{array}{*{20}c}
   {h_{11} } & {h_{12} }  \\
   {h_{21} } & {h_{22} }  \\
\end{array}} \right),
\end{equation}
where $h_{11}$, $h_{12}$ are data channel responses of cellular
link, and $h_{21}$, $h_{22}$ are interference channel responses of
D2D link. The transmitted signal at the BS is obtained by
\begin{equation}
{\bf{x}} = {\bf{WAs}}.
\end{equation}
Here, $\bf{W}$ is beamforming matrix, $\bf{A}$ is power
normalization matrix, and $\bf{S}$ is data vector. Thus, the
received signal at cellular user $UE_c^1$ and D2D receiver $UE_d^2$
can be jointly written as
\begin{equation}
{\bf{y}} = {\bf{HWAs}} + {\bf{n}}.
\end{equation}

\begin{figure}[!ht]
\centering
\includegraphics[width=2.8in]{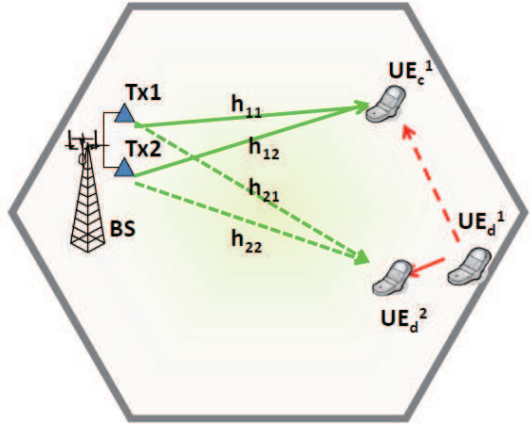}
\caption{An example of beamforming with two antennas equipped at the
BS.} \label{fig:beam_model2}
\end{figure}
In this model, the BS is the control center which conduct
beamforming and power control at the same time.

The key ideas of this scheme are as follows.
\begin{enumerate}
\item The BS carries out beamforming to avoid D2D receiving excessive
interference from the BS.
\item D2D receiver and cellular user feed back downlink CSI to the
BS.
\item The BS calculates transmit power to maximize system sum rate
subjected to SINR threshold of both cellular and D2D links.
\end{enumerate}

The key benefits of the above scheme: the system better adapts to
D2D link quality in downlink resource sharing. In general, the
scheme could guarantee the performance of both cellular and D2D
links, and could maximize system throughput as centralized
characteristics.

A simplified process of signaling interactions is shown in Fig.
\ref{fig:beam_signal}, and Fig. \ref{fig:beam_process} lists the
implementation process of the proposed joint beamforming and power
control scheme.

\begin{figure}[!ht]
\centering
\includegraphics[width=3.8in]{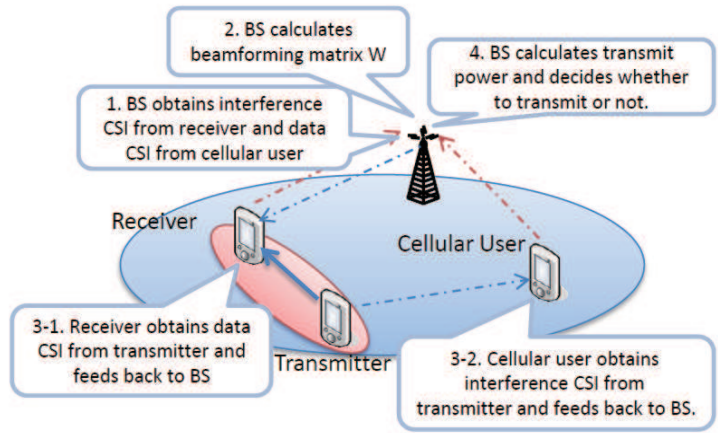}
\caption{Signaling interaction of the joint beamforming and power
control scheme.} \label{fig:beam_signal}
\end{figure}

\begin{figure}[!ht]
\centering
\includegraphics[width=3.8in]{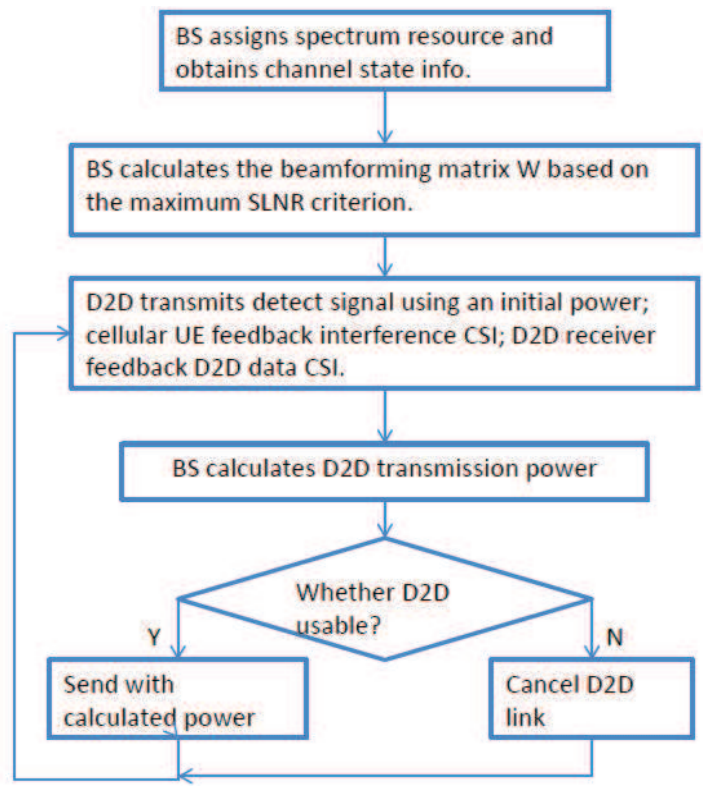}
\caption{Implementation process of the joint beamforming and power
control scheme.} \label{fig:beam_process}
\end{figure}

\subsection{Beamforming Matrix and Power Calculation}

The BS calculates the beamforming matrix according to the CSI that cellular
and D2D user feedback to it. Based on the system model in subsection
\ref{subsec:beam}, the received signal at cellular user and
D2D receiver can be written as
\begin{equation}
y_c  = {\bf{h}}_c^H {\bf{W}}\sqrt {P_B } s_c  + h_{dc} \sqrt {P_d }
s_d  + n,
\end{equation}
\begin{equation}
y_d  = h_{dd} \sqrt {P_d } s_d  + {\bf{h}}_d^H {\bf{W}}\sqrt {P_B }
s_c  + n,
\end{equation}
respectively. Here, ${\bf{h}}_c  = \left( {\begin{array}{*{20}c}
   {h_{11} } & {h_{21} }  \\
\end{array}} \right)^T$ is the signal channel response of cellular
user, and ${\bf{h}}_d  = \left( {\begin{array}{*{20}c}
   {h_{12} } & {h_{22} }  \\
\end{array}} \right)^T$ is the interference channel response of D2D
receiver. $ {\bf{W}} = \left( {\begin{array}{*{20}c}
   {w_1 } & {w_2 }  \\
\end{array}} \right)^T$ is the beamforming matrix satisfying ${\bf{W}}^H {\bf{W}} = 1 $.
$s_c$ and $s_d$ represent transmit signals from the BS and D2D
transmitter, respectively. $P_B$ and $P_d$ denote transmit power
from the BS and D2D transmitter, respectively. $n$ is thermal noise
with variance $N_0$. We employ maximizing SLNR as the beamforming
criterion. That is,
\begin{equation}
\max \begin{array}{*{20}c}
   {}  \\
\end{array}\frac{{{\bf{W}}^H {\bf{h}}_c {\bf{h}}_c^H {\bf{W}}}}
{{{\bf{W}}^H {\bf{h}}_d {\bf{h}}_d^H {\bf{W}} + \frac{{N_0 }}{{P_B
}}}}.
\end{equation}
Thus, the beamforming matrix can be obtained by
\begin{equation}
{\bf{W}} = \frac{1}{\rho }\left( {{\bf{HH}}^H  + \frac{{N_0 }}{{P_B
}}{\bf{I}}} \right)^{ - 1} {\bf{h}}_c,
\end{equation}
where $ {\bf{H}} = \left( {\begin{array}{*{20}c}
   {{\bf{h}}_c } & {{\bf{h}}_d }  \\
\end{array}} \right)$ is the channel response from the BS to users,
and $\rho  = \left\| {\left( {{\bf{HH}}^H  + \frac{{N_0 }}{{P_B
}}{\bf{I}}} \right)^{ - 1} {\bf{h}}_c } \right\|$ is a normalization
factor so that ${\bf{W}}^H {\bf{W}} = 1$.

In this scheme, our objective is to maximize the system sum rate
which is expressed as
\begin{equation}
R = \log _2 (1 + SINR_c ) + \log _2 (1 + SINR_d ).
\end{equation}
In addition, the D2D transmit power $P_d$ also need to satisfy the
SINR threshold of both cellular and D2D links, i.e.,
\begin{equation}
SINR_c  = \frac{{P_B \left\| {{\bf{h}}_c^H {\bf{W}}} \right\|^2
}}{{P_d h_{dc}^2  + N_0 }} \ge \beta _c,
\end{equation}
\begin{equation}
SINR_d  = \frac{{P_d h_{dd}^2 }}{{P_B \left\| {{\bf{h}}_d^H
{\bf{W}}} \right\|^2  + N_0 }} \ge \beta _d,
\end{equation}
where $\beta_c$ and $\beta_d$ are the SINR minimum threshold of
cellular user and D2D receiver, respectively. Thus, we can summarize
the objective function as
\begin{equation}
\max~R = \hspace{-0.1cm}\log _2 \hspace{-0.1cm} \left( {1 + \frac{{P_B \left\|
{{\bf{h}}_c^H {\bf{W}}} \right\|^2 }}{{P_d h_{dc}^2  + N_0 }}}
 \right) \hspace{-0.1cm}+ \log _2\hspace{-0.1cm} \left( {1 + \frac{{P_d h_{dd}^2 }}
 {{P_B \left\| {{\bf{h}}_d^H {\bf{W}}} \right\|^2  + N_0 }}}
 \right),
\end{equation}
\vspace{-0.5cm}
\begin{align}
&subject~to~\nonumber \\
&\left( {P_B \left\| {{\bf{h}}_d^H {\bf{W}}} \right\|^2 \hspace{-0.1cm} + \hspace{-0.1cm}N_0 } \right)\beta _d h_{dd}^{ - 2} \hspace{-0.1cm}\le \hspace{-0.1cm}P_d
\hspace{-0.1cm}\le \hspace{-0.1cm}\min \left( {\left( {P_B \left\| {{\bf{h}}_c^H {\bf{W}}} \right\|^2 \beta _c^{ - 1} \hspace{-0.1cm} -\hspace{-0.1cm} N_0 } \right)h_{dc}^{ - 2} ,P_{\max } }\hspace{-0.1cm} \right),
\end{align}
where $P_{max}$ is the maximum transmit power that UE supports.


Next, we give some
simulation results and related analysis. Throughputs of cellular,
D2D users and the whole system are our main performance metrics. And
we compare some different interference management schemes on the
system performance, including:
\begin{enumerate}
\item Proposed PC \& BF: joint beamforming and power control;
\item No PC \& BF: only beamforming is applied, and power is fixed;
\item PC \& no BF: only power control is applied;
\item No PC \& no BF: neither power control nor beamforming is
applied.
\end{enumerate}
The simulation parameters are listed in Table
\ref{tab:beam_parameter}.


\begin{table*}
\begin{center}
\caption{Main Simulation Parameters}\label{tab:beam_parameter}
\begin{tabular}{|l|l|}
\hline \bf{Parameter} & \bf{Value}\\ \hline Cellular & Isolated cell, 1-sector \\
\hline System area & User devices are distributed in a \\
& hexagonal cell with 600m radius.\\
\hline Noise spectral density & -174dBm/Hz \\
\hline System bandwidth & 20MHz \\
\hline Sub-carrier bandwidth & 15kHz \\
\hline Sub-carrier number of each user & 64 \\
\hline Cluster radius (D2D user scattering) & 50m \\
\hline Minimum SINR & 5dB (both cellular and D2D)\\
\hline Number of cellular users (channels) & 1 \\
\hline Number of D2D pairs & 1/1:4 \\
\hline Device transmit power upper bound & 23dBm\\
\hline BS total transmit power & 46dBm\\
\hline Channel Model & WINNER II \\ \hline
\end{tabular}
\end{center}
\end{table*}

\begin{figure}[!ht]
\centering
\includegraphics[height=3.3in]{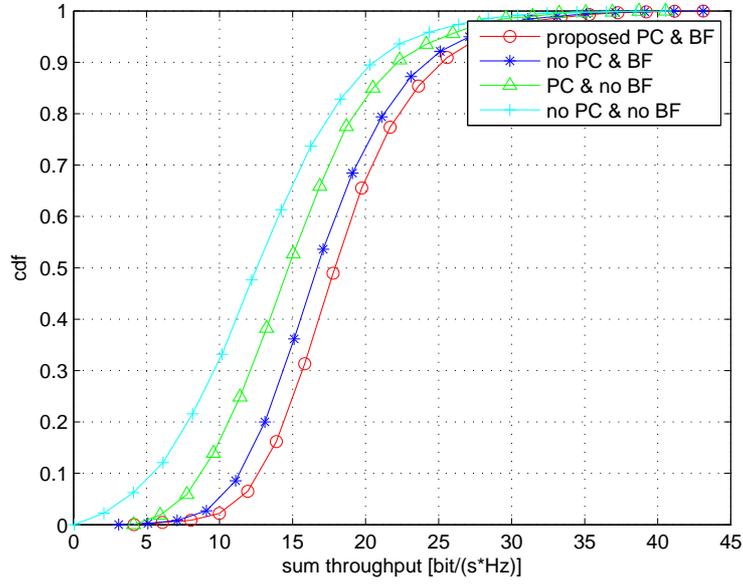}
\caption{System throughput distribution with different interference
management schemes.} \label{fig:beam_sumrate}
\end{figure}

\begin{figure}[!ht]
\centering
\includegraphics[height=3.3in]{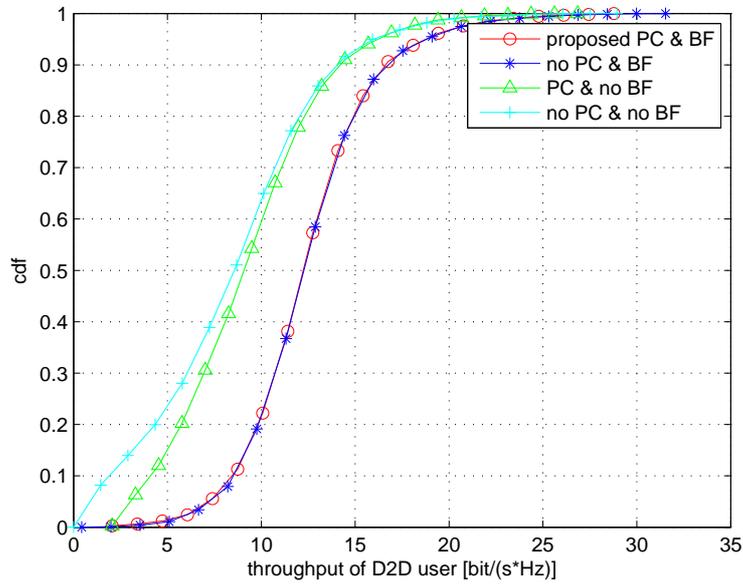}
\caption{Throughput distribution of D2D communication with different
interference management schemes.} \label{fig:beam_D2Drate}
\end{figure}

\begin{figure}[!ht]
\centering
\includegraphics[height=3.4in]{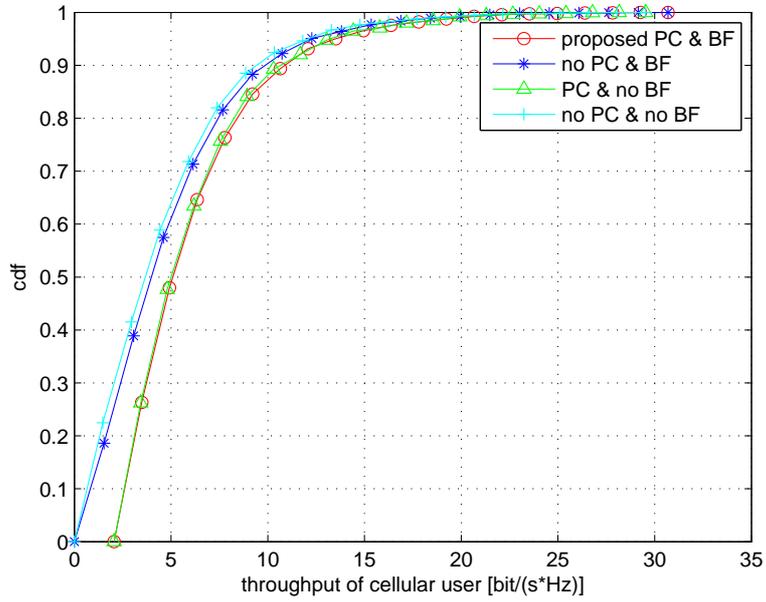}
\caption{Throughput distribution of cellular communication with
different interference management schemes.} \label{fig:beam_celrate}
\end{figure}

\begin{figure}[!ht]
\centering
\includegraphics[height=3.5in]{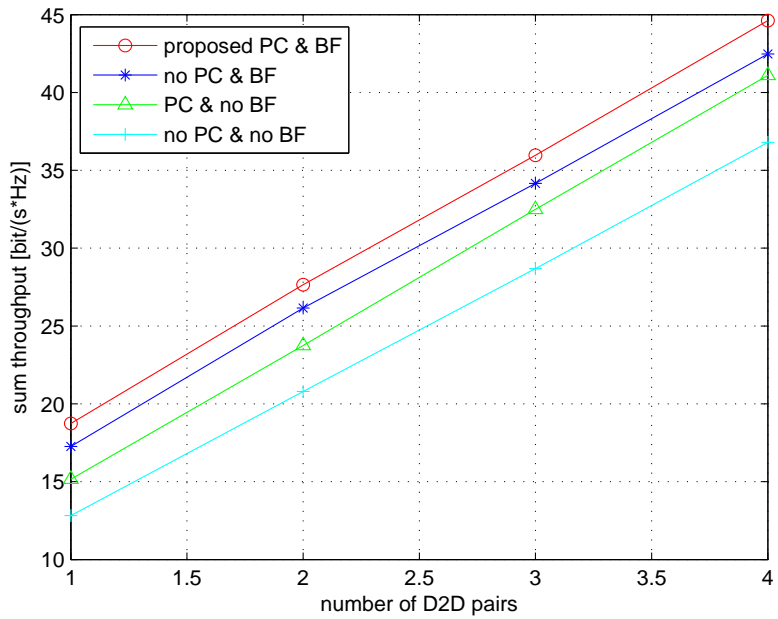}
\caption{System throughput with different number of D2D pairs.}
\label{fig:beam_sumrate_D}
\end{figure}

\begin{figure}[!ht]
\centering
\includegraphics[height=3.5in]{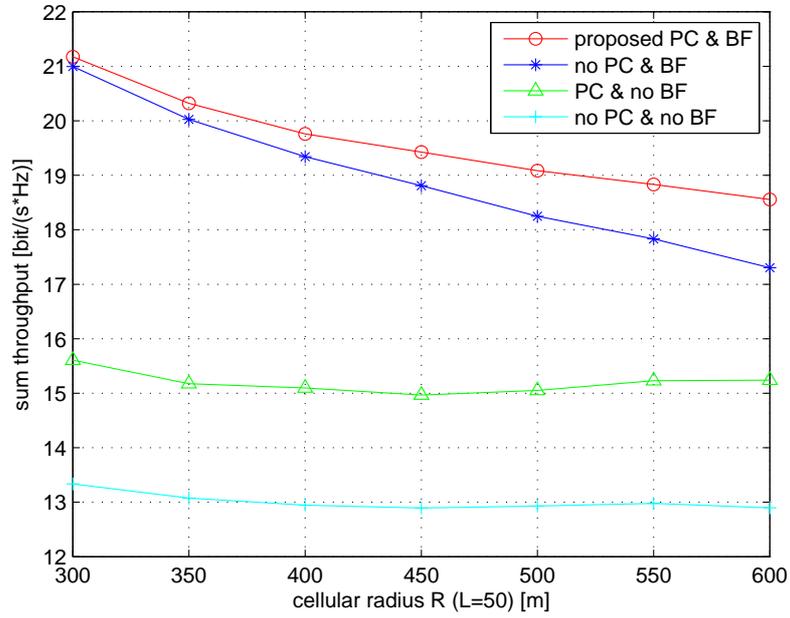}
\caption{System throughput with different cellular radius.}
\label{fig:beam_sumrate_R}
\end{figure}

\begin{figure}[!ht]
\centering
\includegraphics[height=3.4in]{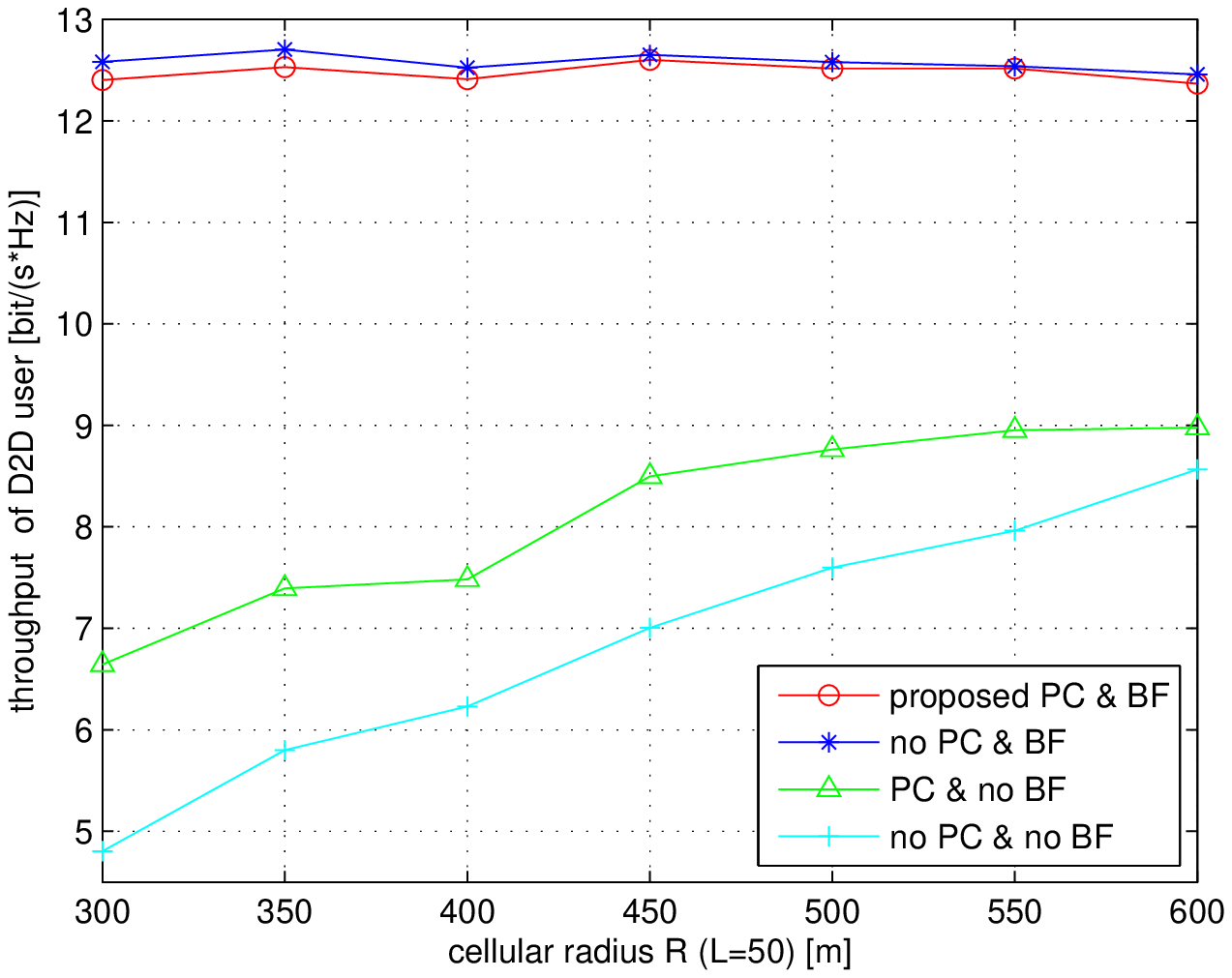}
\caption{Throughput of D2D communication with different cellular
radius.} \label{fig:beam_D2Drate_R}
\end{figure}

Fig. \ref{fig:beam_sumrate} shows the system throughput distribution
with different schemes. Obviously, joint beamforming and power
control scheme makes the whole system performance better. Fig.
\ref{fig:beam_D2Drate} and Fig. \ref{fig:beam_celrate} give the
throughput distribution of D2D and cellular user. On one side,
beamforming makes the performance of D2D communication better as
SLNR criteria constrains the interference from the BS to D2D
receiver. On the other side, power control makes the performance of
cellular communication better as it limits the interference from D2D
to cellular user.

For multiple D2D pairs, we consider the objective function the same
as that of single pair for simplicity. Fig. \ref{fig:beam_sumrate_D}
is the result of the system throughput with different number of D2D
pairs. We can see that the proposed scheme gives the optimal
performance.

Fig. \ref{fig:beam_sumrate_R} gives the system throughput with
different cellular radius. When the radius is small, beamforming
prevents interference from BS to D2D effectively which makes the
performance gain obvious. When the radius increases, the received
signal power of cellular user decreases. Thus, power control becomes
more important to prevent interference from D2D.

Fig. \ref{fig:beam_D2Drate_R} shows the throughput of D2D user with
different cellular radius. We find that power control does not bring
strong fading to D2D communication. In fact, power control includes
the mechanism that guarantees SINR of D2D being above the threshold,
and the objective function of power control is the system
throughput. We focus on the top two lines, and find the gap between
the proposed scheme and the other one is comparably obvious when the
cellular radius is small. Because cellular communication is dominant
when the radius is small, system throughput would get profit from
D2D power reducing. Thus, fading of D2D link becomes obvious.


This section has proposed a joint beamforming and power control
scheme that aims to maximize the system sum rate while guarantees
the performance of both cellular and D2D connections. The BS sets
SINR threshold for D2D and cellular links, and the value can be
tuned to meet corresponding performance requirements. The BS carries
out beamforming to avoid D2D from excessive interference. D2D
transmit power is calculated by the BS based on maximizing the
system sum rate. Also, the BS decides whether the calculated D2D
transmit power available according to SINR threshold of cellular and
D2D links.

\chapter{Radio Resource Management}\label{chap:allocation}
An innovative resource allocation
scheme is proposed to improve the performance of mobile
peer-to-peer, i.e., device-to-device (D2D), communications as an
underlay in the downlink (DL) cellular networks. To optimize the
system sum rate over the resource sharing of both D2D and cellular
modes, we introduce a reverse iterative combinatorial auction as the
allocation mechanism. In the auction, all the spectrum resources are
considered as a set of resource units, which as bidders compete to
obtain business while the packages of the D2D pairs are auctioned
off as goods in each auction round. We first formulate the valuation
of each resource unit, as a basis of the proposed auction. And then
a detailed non-monotonic descending price auction algorithm is
explained depending on the utility function that accounts for the
channel gain from D2D and the costs for the system. Further, we
prove that the proposed auction-based scheme is cheat-proof, and
converges in a finite number of iteration rounds. We explain
non-monotonicity in the price update process and show lower
complexity compared to a traditional combinatorial allocation. The
simulation results demonstrate that the algorithm efficiently leads
to a good performance on the system sum rate.

\begin{figure}[!ht]
\centering
\includegraphics[height=3.7in]{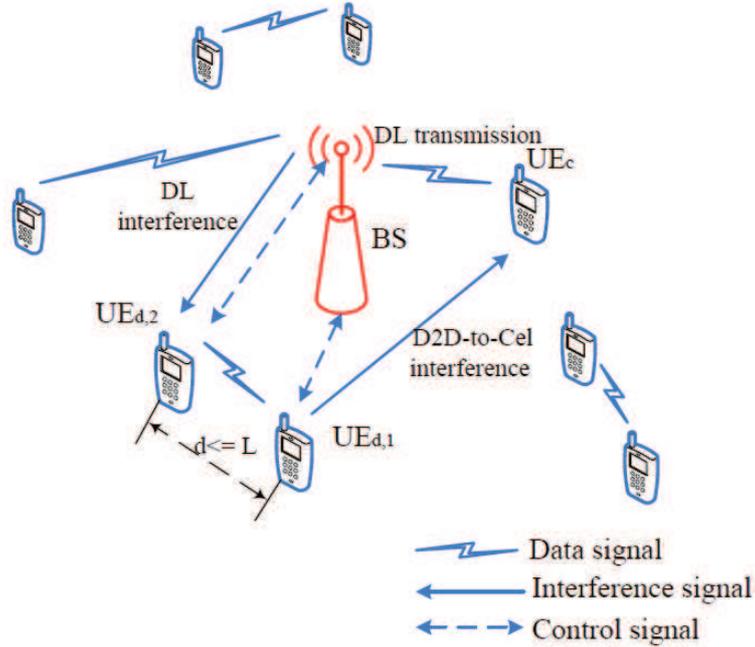}
\caption{System model of D2D communication underlaying cellular
networks with downlink resource sharing.} \label{fig:1}
\end{figure}

\section{Basics of Model and Auction Mechanism}\label{sec:system_model}

In this section, we introduce the system model for D2D underlay
communication. The scenario of multiple D2D and cellular users is
first described, and then, the expression of system sum rate is
given.

{\bf{1) Scenario Description}}

A model of a single cell with multiple users is considered. As shown
in Fig. \ref{fig:1}, UEs with data signals between each other are in
the D2D communication mode while UEs that transmit data signals with
the BS keep in the traditional cellular mode. Each user is equipped
with a single omnidirectional antenna. The locations of cellular
users and D2D pairs are randomly set and traversing the whole cell.
Without loss of generality, we employ the uniform distribution to
describe the user locations which is proposed for system simulation
in \cite{ITU}. Notice that from stochastic geometry with for Poisson
distributions, the users are uniformly located as well if the number
of users is known \cite{Haenggi2009}. For simplicity and clarity, we
illustrate co-channel interference scenario involving three UEs
(UE$_c$, UE$_{d,1}$ and UE$_{d,2}$), and omit the interference and
control signal signs among others. UE$_c$ is a traditional cellular
user that is distributed uniformly in the cell. UE$_{d,1}$ and
UE$_{d,2}$ are close enough to satisfy the distance constraints of
D2D communication, and at the same time they also have communicating
demands. One member of the D2D pair UE$_{d,1}$ is distributed
uniformly in the cell, and the position of the other member
UE$_{d,2}$ follows a uniform distribution inside a region at most
$L$ from UE$_{d,1}$.

The existing researches \cite{Xu2010,Yu_TWC} confirm that with power
control or resource scheduling mechanism, the inter-cell
interference can be managed efficiently. Therefore, we place an
emphasis on the intra-cell interference, which is due to resource
sharing of D2D and cellular communication. Generally speaking, the
session setup of D2D communication requires the following steps
\cite{Doppler2009}:
\begin{enumerate}
\item A request of communicating is initiated by one UE pair.
\item The system detects traffic originating from and destined to the
UE in the same subnet.
\item If the traffic fulfills a certain
criterion (e.g., data rate), the system considers the traffic as
potential D2D traffic.
\item The BS checks if D2D communication
offers higher throughput.
\item If both UEs are D2D capable and D2D mode offers higher
throughput, the BS may set up a D2D bearer.
\end{enumerate}
The cross-layer processes of resource control can be contained in
the above steps, and be generally summarized as: the transmitters
(both cellular and D2D users) send detecting signals. Then CSI would
be obtained by corresponding receivers and be feedback to the
control center (e.g. the BS). The power control and spectrum
allocation are conducted based on certain principles. Finally, the
BS sends control signals to users according to allocation results.

Even if the D2D connection setup is successful, the BS still
maintain detecting if UE should be back to the cellular
communication mode. Furthermore, the BS maintains the radio resource
control for both cellular and D2D communication. Based on these
communication features, we mainly focus on assigning
cellular resources to D2D communication.

Here, we consider a scenario of sharing downlink (DL)
resource of the cellular network as shown in Fig.~\ref{fig:1}. We
assume UE$_{d,1}$ is the transmitter of the D2D pair sharing the
same sub-channel with the BS, and thus, UE$_{d,2}$ as the D2D
receiver receives interference from the BS. Also, the cellular
receiver UE$_c$ is exposed to interference from UE$_{d,1}$. In
addition, the D2D users feed back the CSI to the BS, whereas the BS
transmits control signals to the D2D pair, in the way that the
system achieves D2D power control and resource allocation.

During the DL period of the cellular system, both cellular and D2D
users receive interference as they share the same sub-channels.
Here, we assume that any cellular user's resource blocks (RBs) can
be shared with multiple D2D pairs and each pair can use more than
one user's RBs for transmitting. We assume the numbers of cellular
users and D2D pairs in the model are $C$ and $D$, respectively.
During the DL period, the BS transmits signal $x_c$ to the $c$-th
($c=1,2,...,C$) cellular user, and the $d$-th ($d=1,2,...,D$) D2D
pair uses the same spectrum resources transmitting signal $x_d$. The
received signals at UE $c$ and D2D receiver $d$ are written as
\vspace{-0.2cm}
\begin{equation}
y_c  = \sqrt {P_B } h_{Bc} x_c  + \sum\limits_d {\beta _{cd} \sqrt
{P_d } h_{dc} x_d }  + n_c,
\end{equation}
\vspace{-0.2cm}
\begin{equation}
y_d  \hspace{-0.1cm}= \hspace{-0.1cm}\sqrt {P_d } h_{dd} x_d + \sqrt
{P_B } h_{Bd} x_c
 \hspace{-0.1cm} +  \hspace{-0.1cm} \sum\limits_{d'} {\beta _{dd'}
\sqrt {P_{d'} } h_{d'd} x_{d'} } \hspace{-0.1cm} +  n_d,
\end{equation}
where $P_B$, $P_d$ and $P_{d'}$ are the transmit power of BS, D2D
transmitter $d$, $d'$, respectively. $h_{ij}$ is the channel
response of the $i-j$ link that is from equipments $i$ to $j$. $n_c$
and $n_d$ are the additive white Gaussian noise (AWGN) at the
receivers with one-sided power spectral density (PSD) $N_0$.
$\beta_{cd}$ represents the presence of interference satisfying
$\beta_{cd}=1$ when RBs of UE $c$ are assigned to UE $d$, otherwise
$\beta_{cd}=0$. As a cellular user can share resources with multiple
D2D pairs, it also satisfies  $0 \le \sum\nolimits_d {\beta _{cd} }
\le D$. Similarly, $\beta_{dd'}$ represents the presence of
interference between D2D pairs $d$ and $d'$.

The channel is modeled as the Rayleigh fading
channel, and thus, the channel response follows the independent
identical complex Gaussian distribution. In addition, the free space
propagation path-loss model, $P = P_0  \cdot \left( {{d
\mathord{\left/ {\vphantom {d {d_0 }}} \right.
\kern-\nulldelimiterspace} {d_0 }}} \right)^{ - \alpha}$, is used
where $P_0$ and $P$ represent signal power measured at $d_0$ and $d$
away from the transmitter, respectively. $\alpha$ is the path-loss
exponent. Hence, the received power of each link can be expressed as
\vspace{-0.2cm}
\begin{equation}
P_{r,ij}  = P_i  \cdot h_{ij}^2  = P_i  \cdot \left( {d_{ij} }
\right)^{ - \alpha}  \cdot h_0^2,
\end{equation}
where $P_{r,ij}$ and $d_{ij}$ are the received power and the
distance of the $i-j$ link, respectively. $P_i$ represents the
transmit power of equipment $i$ and $h_0$ is the complex Gaussian
channel coefficient that obeys the distribution $\mathcal
{C}\mathcal {N}(0,1)$. And we simplify the received power at $d_0=1$
equals the transmit power.

{\bf{2) System Sum Rate}}

For the purpose of maximizing the network
capacity, the signal to interference plus noise ratio (SINR) should
be considered as an important indicator. The SINR of user $j$ is
\vspace{-0.1cm}
\begin{equation}\label{eq:SINR}
\gamma _j  = \frac{{P_i h_{ij}^2 }}{{P_{{\mathop{\rm int}} ,j}  +
N_0 }},
\end{equation}
\vspace{-0.1cm}where $P_{{\mathop{\rm int}} ,j}$ denotes the
interference signal power received by user $j$, and $N_0$ accounts
for the terminal noise at the receiver.

Determined by the Shannon capacity formula, we can calculate the
channel rate corresponding to the SINR of cellular and D2D users. As
cellular users suffer interference from D2D communicating that
sharing the same spectrum resource, the interference power of
cellular user $c$ is \vspace{-0.1cm}
\begin{equation}\label{eq:P_int,c}
P_{{\mathop{\rm int}} ,c}  = \sum\limits_d {\beta _{cd} P_d
h_{dc}^2}.
\end{equation}
While the interference of D2D receiver $d$ is from both BS and D2D
users that are assigned the same resources to, the interference
power of user $d$ can be expressed as \vspace{-0.1cm}
\begin{equation}\label{eq:P_int,d}
P_{{\mathop{\rm int}} ,d}  = P_B h_{Bd}^2  + \sum\limits_{d'} {\beta
_{dd'} P_{d'} h_{d'd}^2 }.
\end{equation}
\vspace{-0.2cm}

Based on (\ref{eq:SINR}), (\ref{eq:P_int,c}), and
(\ref{eq:P_int,d}), we can obtain the channel rate of cellular user
$c$ and D2D receiver $d$ as \vspace{-0.2cm}
\begin{equation}\label{eq:R_c}
R_c  = \log _2 \left( {1 + \frac{{P_B h_{Bc}^2 }}{{\sum\limits_d
{\beta _{cd} P_d h_{dc}^2 }  + N_0 }}} \right),
\end{equation}
\vspace{-0.2cm}
\begin{equation}\label{eq:R_d}
R_d  =  \log _2 \left( {1 + \frac{{P_d h_{dd}^2 }}{{P_B h_{Bd}^2 +
\sum\limits_{d'} {\beta _{dd'} P_{d'} h_{d'd}^2 }  + N_0 }}}
\right),
\end{equation}
respectively. Here, $d \ne d'$. So $\sum\limits_{d'} {{\beta
_{dd'}}{P_{d'}}{h_{d'd}^2}}$ represents the interference from the
other D2D pairs that share spectrum resources with pair $d$.

The DL system sum rate can be defined as \vspace{-0.2cm}
\begin{equation}\label{eq:Re}
\Re  = \sum\limits_{c = 1}^C {\left( {R_c  + \sum\limits_{d = 1}^D
{\beta _{cd} R_d } } \right)}.
\end{equation}
In the next subsection, we formulate the problem of designing
$\beta_{cd}$ for each D2D pair as an optimization issue of
maximizing $\Re$.

In this subsection, we introduce two concepts: valuation model and
utility function, which are bases of the auction mechanism. Also,
some definitions are given.

{\bf{3) Valuation Model}}

As D2D communication shares the same spectrum resources with
cellular communication at the same time slot, the co-channel
interference should be limited as much as possible to optimize the
system performance. The radio signals experience different degrees
of fading, and thus, the amount of interference depends on transmit
power and spatial distances. Accordingly, we focus on assigning
appropriate resource blocks (RBs) occupied by cellular users to D2D
pairs in order to minimize interference to achieve a higher system
sum rate. Next, we formulate the relation between the allocation
result and the rate of the shared channel. The relation can be
defined as a value function whose target value is the channel rate.

We define $\mathcal{D}$ as a package of variables representing the
index of D2D pairs that share the same resources. We assume the
total pairs can form $N$ such packages. Thus, if the members of the
$k$-th ($k=1,2,...,N$) D2D user package share resources with
cellular user $c$, the channel rates of UE $c$ and D2D pair $d$ ($d
\in \mathcal{D}_k$) can be written as \vspace{-0.2cm}
\begin{equation}\label{eq:R_ck}
{R_c^k} = {\log _2}\left( {1 +
\frac{{{P_B}h_{Bc}^2}}{{\sum\limits_{d \in \mathcal{D}_k}
{{P_d}h_{dc}^2}  + {N_0}}}} \right),
\end{equation}
\vspace{-0.2cm}
\begin{equation}\label{eq:R_dk}
{R_d^k}  =  {\log _2} \hspace{-0.1cm} \left( {1 +
\frac{{{P_d}h_{dd}^2}}{{{P_B}h_{Bd}^2 + \hspace{-0.1cm}
\sum\limits_{d' \in \mathcal{D}_k - \left\{ d \right\}}
\hspace{-0.1cm} {{P_{d'}}h_{d'd}^2}  + {N_0}}}} \right),
\end{equation}
respectively. The rate of the operating channel shared by UE $c$ and
D2D pairs $d \in \mathcal{D}_k$ is \vspace{-0.2cm}
\begin{equation}\label{eq:R}
{R_{ck}} = {R_c^k} + \sum\limits_{d \in \mathcal{D}_k} {{R_d^k}}.
\end{equation}

According to (\ref{eq:R_ck}) $\sim$ (\ref{eq:R}), when assigning
resources of UE $c$ to the $k$-th package of D2D pairs, the channel
rate is given by \vspace{-0.2cm}
\begin{align}\label{eq:V_ck}
V_c(k)\hspace{-0.1cm}= & \log _2 \left( {1 + \frac{{P_B h_{Bc}^2
}}{{\sum\limits_{d
\in \mathcal {D}_k } {P_d h_{dc}^2 }  + N_0 }}} \right) \nonumber \\
 & \hspace{-0.1cm} + \hspace{-0.15cm} \sum\limits_{d \in \mathcal {D}_k }
 \hspace{-0.1cm} {\log _2 \hspace{-0.1cm}
\left( {\hspace{-0.1cm} 1 \hspace{-0.1cm} + \hspace{-0.1cm}
\frac{{P_d h_{dd}^2 }}{{P_B h_{Bd}^2 \hspace{-0.1cm} +
\hspace{-0.15cm} \sum\limits_{d' \in \mathcal {D}_k \hspace{-0.05cm}
- \hspace{-0.05cm}\left\{ d \right\}} \hspace{-0.15cm} {P_{d'}
h_{d'd}^2 } \hspace{-0.1cm} + N_0 }}\hspace{-0.1cm}} \right)}.
\end{align}

In the proposed reverse I-CA mechanism, we consider spectrum
resources occupied by cellular user $c$ as one of the bidders who
submit bids to compete for the packages of D2D pairs, in order to
maximize the channel rate. It is obvious that there would be a gain
of channel rate owing to D2D communicating as long as the
contribution to data signals from D2D is larger than that to
interference signals. Considering the constraint of a positive
value, we define the performance gain as \vspace{-0.2cm}
\begin{equation}\label{eq:value}
v_{c}(k) =max\left(V_{c}(k) - V_{c},0\right),
\end{equation}
which is the private valuation of bidder $c$ for the package of D2D
pairs $\mathcal{D}_k$. Here, $V_{c}$ denotes the channel rate of UE
$c$ without co-channel interference and is obtained by
\vspace{-0.2cm}
\begin{equation}\label{eq:V_c}
V_c= {\log _2}\left( {1 + \frac{{{P_B}h_{Bc}^2}}{{{N_0}}}} \right).
\end{equation}
Thus, we have the following definition:
\newtheorem{definition}{Definition}
\begin{definition}
 A \textbf{\emph{valuation~model}} $\mathcal
{V} = \left\{ {v_c \left( k \right)} \right\}$ is a set of the
private valuations of all bidders $c \in \left\{ {1,2, \ldots ,C}
\right\}$ for all packages $\mathcal{D}_k \subseteq \left\{ {1,2,
\ldots ,D} \right\}$ ($k \in \left\{ 1,2, \ldots ,N \right\}$).
\end{definition}

{\bf{4) Utility Function}}

In the auction, the cellular resource denoted by $c$ obtains a gain
by getting a package of D2D communications. However, there exists
some cost such as control signals transmission and information
feedback during the access process. We define the cost as a pay
price.

\begin{definition}
The price to be payed by the bidder $c$ for
the package $\mathcal{D}_k$ is called \textbf{\emph{pay~price}}
denoted by $\mathcal {P}_c(k)$. The unit price of item $d$ ($\forall
k$, $d \in \mathcal{D}_k$) can be denoted by $p_c(d)$.
\end{definition}

Here, we consider linear anonymous prices \cite{Pikovsky2008}, which
means the prices are linear if the price of a package is equal to
the sum of the prices of its items, and the prices are anonymous if
the prices of the same package for different bidders are equal.
Thus, we have \vspace{-0.2cm}
\begin{equation}\label{eq:price}
\mathcal {P}_c(k)= \sum\limits_{d \in \mathcal {D}_k} {p_c \left( d
\right)}  = \sum\limits_{d \in \mathcal {D}_k} {p\left( d \right)}
,\forall c = 1,2, \ldots ,C.
\end{equation}
Therefore, the payment of a bidder is determined by the unit price
$p(d)$ and the size of bidding package $\mathcal {D}_k$.

\begin{definition}
 \textbf{\emph{Bidder~utility}}, or named
\textbf{\emph{bidder~payoff}} $\mathcal {U}_c(k)$ expresses
satisfaction of bidder $c$ getting package $\mathcal {D}_k$. The
bidder utility can be defined as \vspace{-0.2cm}
\begin{equation}\label{eq:U}
\mathcal {U}_c(k)= v_c \left( k \right) -\mathcal {P}_c(k).
\end{equation}
\end{definition}

Based on (\ref{eq:value}), (\ref{eq:price}), (\ref{eq:U}), $V_c(k)$
in (\ref{eq:V_ck}) and $V_c$ in (\ref{eq:V_c}), we can obtain the
utility of bidder $c$ as \vspace{-0.2cm}
\begin{align}\label{eq:utility_JSAC}
\mathcal {U}_c (k) = & \log _2 \left( {1 + \frac{{P_B h_{Bc}^2
}}{{\sum\limits_{d \in \mathcal {D}_k } {P_d h_{dc}^2 }  + N_0 }}}
\right) \nonumber \\
& + \hspace{-0.1cm} \sum\limits_{d \in \mathcal {D}_k
}\hspace{-0.1cm} {\log _2 \hspace{-0.1cm} \left( {1 + \frac{{P_d
h_{dd}^2 }}{{P_B h_{Bd}^2  + \hspace{-0.1cm}
 \sum\limits_{d' \in \mathcal {D}_k  - \left\{ d \right\}} \hspace{-0.1cm} {P_{d'}
h_{d'd}^2 }  + N_0 }}} \right)}
\nonumber \\
& - \log _2 \left( {1 + \frac{{P_B h_{Bc}^2 }}{{N_0 }}} \right) -
\sum\limits_{d \in \mathcal {D}_k } {p\left( d \right)}.
\end{align}

In order to describe the allocation outcome intuitively, we give the
definition below.

\begin{definition}
The result of the auction is a spectrum
allocation denoted by $\mathcal {X}=\left(X_1 ,X_2 , \ldots ,X_C
\right)$, which allocates a corresponding package to each bidder.
And the allocated packages may not intersect ($\forall i,j,~X_i \cap
X_j = \emptyset$).
\end{definition}

We consider a set of binary variables $\left\{ x_c(k)\right\}$ to
redefine the allocation as \vspace{-0.2cm}
\begin{equation}\label{eq:x_ck}
x_c \left( k \right) = \left\{ {\begin{array}{*{20}c}
   {1}, & {\mbox{if}~X_c  = \mathcal {D}_k},  \\
   {0}, & {\mbox{otherwise}}.  \\
\end{array}} \right.
\end{equation}
According to the literature, two most popular bidding languages are
exclusive-OR (XOR), which allows a bidder to submit multiple bids
but at most one of the bids can win, and additive-OR (OR), which
allows one to submit multiple bids and any non-intersecting
combination of the bids can win. We consider the XOR bidding
language. Thus, (\ref{eq:x_ck}) satisfies $
\sum\nolimits_{k = 1}^N {x_c \left( k \right)}  \le 1$ and
$\sum\nolimits_{k = 1}^N {x_c \left( k \right)}  = 0 \Rightarrow
X_c= \emptyset$ for $\forall c=1,2, \ldots ,C$. If given an
allocation $\mathcal {X}$, the total bidder utility of all bidders
can be denoted as $\mathcal {U}_{all}(\mathcal {X})=\sum\nolimits_{c
= 1}^C {\sum\nolimits_{k = 1}^N {x_c(k)\mathcal {U}_c(k)}}$.
Furthermore, the auctioneer revenue is denoted as $\mathcal
{A}(\mathcal {X})= \sum\nolimits_{c = 1}^C {\sum\nolimits_{k = 1}^N
{x_c(k)\mathcal {P}_c(k)}} $, which is usually considered to be the
auctioneer's gain.

\section{Resource Allocation Algorithm Based on Reverse Iterative
Combinatorial Auction}

In this section, we formulate the resource allocation for D2D
communication as a reverse I-CA game. First, we introduce some
concepts of the I-CA games. Then, we investigate details of the
allocation process.

{\bf{1) Reverse Iterative Combinatorial Auction Game}}

As mentioned before, we assume the total spectrum resources are
divided into $C$ units with each one already providing communication
service to one cellular user. By the auction game, the spectrum
units are assigned to $N$ user packages $\left\{
\mathcal{D}_1,\mathcal{D}_2,...,\mathcal{D}_N \right\}$, with each
package consisting of at least one D2D pair. In other words, the
spectrum units compete to obtain D2D communication for improving the
channel rate.

During an I-CA game, the auctioneer announces an initial price for
each item, and then, the bidders submit to the auctioneer their bids
at the current price. As long as the demand exceeds the supply, or
on the contrary that the supply exceeds the demand, the auctioneer
updates (raises or reduces) the corresponding price and the auction
goes to the next round.

Obviously, it can be shown that the overall gain, which includes the
total gain of the auctioneer and all bidders does not depend on the
pay price, but equals to the sum of the allocated packages'
valuations, i.e., \vspace{-0.2cm}
\begin{align}
~& \mathcal {A}\left( \mathcal {X} \right) + \mathcal {U}_{all}
\left( \mathcal {X} \right) \nonumber \\
\vspace{-0.1cm} = & \sum\limits_{c = 1}^C {\sum\limits_{k = 1}^N
{x_c \left( k \right)\mathcal {P}_c \left( k \right)} } +
\sum\limits_{c = 1}^C {\sum\limits_{k = 1}^N {x_c
\left( k \right)\mathcal {U}_c \left( k \right)} } \nonumber\\
\vspace{-0.2cm} = & \sum\limits_{c = 1}^C {\sum\limits_{k = 1}^N
{x_c \left( k \right)\mathcal {P}_c \left( k \right)} } +
\sum\limits_{c = 1}^C {\sum\limits_{k = 1}^N {x_c \left( k
\right)[\left( {v_c \left( k
\right) - \mathcal {P}_c \left( k \right)} \right)]} } \nonumber\\
\vspace{-0.2cm} = & \sum\limits_{c = 1}^C {\sum\limits_{k = 1}^N
{x_c \left( k \right)v_c \left( k \right)} }.
\end{align}
As our original intention, we employ the I-CA to obtain an efficient
allocation for spectrum resources.

\begin{definition}
 An \textbf{\emph{efficient allocation}}
denoted by $\tilde{\mathcal
{X}}=(\tilde{X}_1,\tilde{X}_2,\ldots,\tilde{X}_C)=\left\{
\tilde{x}_c(k) \right\}$ is an allocation that maximizes the overall
gain.
\end{definition}

Given the private bidder valuations for all possible packages in
(\ref{eq:value}), an efficient allocation can be obtained by solving
the combinatorial allocation problem (CAP).

\begin{definition}
The \textbf{\emph{Combinatorial Allocation
Problem (CAP)}}, also sometimes referred as \textbf{\emph{Winner
Determination Problem (WDP)}}, leads to an efficient allocation by
maximizing the overall gain: \hspace{-0.6cm}$\mathop {\max }\limits_{\mathcal {D}_k
= X_c  \in \mathcal {X} \in \mathscr{X}} \hspace{-0.1cm}\sum\nolimits_{c = 1}^C \hspace{-0.1cm}
{v_c \hspace{-0.1cm} \left( k \right)}$, where $\mathscr{X}$ denotes the set of all
possible allocations.
\end{definition}

An integer linear program using the binary decision variables
$\left\{ x_c(k)\right\}$ is formulated for the CAP as
\vspace{-0.2cm}
\begin{align}\label{eq:CAP}
&\max \sum\limits_{c = 1}^C {\sum\limits_{k = 1}^N {x_c \left( k
\right)v_c \left( k \right)} }, \\
\vspace{-0.2cm} s.t.~~~&\sum\limits_{k=1}^N {x_c(k)} \leq 1, \forall
c \in \left\{ 1,2, \ldots, C \right\}, \nonumber\\
\vspace{-0.2cm} &\sum\limits_{\mathcal {D}_k :d \in \mathcal {D}_k }
{\sum\limits_{c = 1}^C {x_c (k)}  \le 1} , \forall d \in \left\{
1,2, \ldots, D \right\}. \nonumber
\end{align}
The objective function maximizes the overall gain, and the
constraints guarantee: 1) at most one package can be allocated to
each bidder; 2) each item cannot be sold more than once.

In fact, there might be multiple optimal solutions of the CAP with
the same objective function. From the auctioneer's point of view,
tie-breaking rules are needed to determine which of the optimal
solutions is selected. In a real auction, the auctioneer does not
know the private valuations of the bidders, neither can it solve the
NP-hard problem. To solve the CAP, the auctioneer selects the
winners on the basis of the submitted bids in each round. Therefore,
in case of the XOR bidding language, the WDP formulation is similar
to the CAP and the only difference is the objective function
\vspace{-0.2cm}
\begin{equation}\label{eq:WDP}
\max \sum\limits_{c = 1}^C {\sum\limits_{k = 1}^N {x_c \left( k
\right)\mathcal {P}_c^t \left( k \right)} },
\end{equation}
where $\mathcal {P}_c^t(k)$ represents the pay price of bidder $c$
for package $\mathcal {D}_k$ in round $t$.

Based on Definition 5, the overcome of a CA is not always efficient.
Here, we employ allocating efficiency as a primary measure to
benchmark auctions.

\begin{definition}
\textbf{\emph{Allocating efficiency}} in CAs
can be expressed as the ratio of the overall gain of the final
allocation to that of an efficient allocation \cite{Pikovsky2008}
\vspace{-0.2cm}
\begin{equation}
\mathcal {E}(\mathcal X)=\frac{{\mathcal {A}\left( \mathcal{X}
\right) + \mathcal{U}_{all} \left( \mathcal{X} \right)}}{{\mathcal
{A} \left(\tilde{\mathcal {X}} \right) + \mathcal{U}_{all} \left(
\tilde {\mathcal {X}} \right)}},
\end{equation}
which has $\mathcal{E}(\mathcal X) \in [0,1]$.
\end{definition}

{\bf{2) Algorithm for Resource Allocation}}\label{subsec:algorithm}

In this subsection, the details of the allocation scheme based on
reverse I-CA are introduced. We has modeled the D2D resource
allocation problem as a reverse I-CA game and gave the valuation
model, utility function and other important concepts. Many I-CA
designs, especially for the centralized I-CA design, are based on
ask prices. The price-based I-CA designs differ by the pricing
scheme and price update rules. In the proposed algorithm, linear
prices are used as mentioned in Section \ref{sec:system_model} for
they are easy to understand for bidders and convenient to
communicate in each auction round. Because of the interference from
D2D links, cellular channels should guarantee the performance of
cellular system before allowing the D2D access. Hence, we consider a
descending price criterion in the algorithm. Prices update by a
greedy mode that once a bidder submits a bid for items or packages
the corresponding prices are fixed, otherwise the prices are
decreased.

At the beginning of the allocation, the BS collects the location
information of all the D2D pairs. In addition, the round index
$t=0$, the initial ask price $p^0(d)$ for each item (D2D pair) $d$,
and the fixed price reduction $\Delta>0$ are set up. When the
initial prices are announced to all the bidders (i.e. spectrum
resources occupied by cellular UEs), each bidder submits bids, which
consist of its desired packages and the corresponding pay prices.
Jump bidding where bidders are allowed to bid higher than the
prices, is not allowed in our scheme, thus bidders always bid at the
current prices. According to the CAP proposed in Definition 6 and
the analysis about the WDP, we simplify the problem of maximizing
the overall gain as a process of collecting the highest pay price.
As a result, bidder $c$ bids for package $\mathcal {D}_k$ as long as
$\mathcal{U}_c(k)\geq 0$. Combining (\ref{eq:price}) and
(\ref{eq:U}), we have \vspace{-0.2cm}
\begin{equation}\label{eq:bid_condition}
v_c \left( k \right) \ge \mathcal{P}_c^t \left( k \right) =
\sum\limits_{d \in \mathcal{D}_k } {p^t \left( d \right)},
\end{equation}
where the round index $t \ge 0$. In this case, let $b_c^t(k)=\left\{
\mathcal{D}_k,\mathcal{P}_c^t(k) \right\}$ denote the submitted bid
at the end of round $t$, and $\mathcal{B}^t=\left\{
b_c^t(k)\right\}$ denotes all the bids. When
(\ref{eq:bid_condition}) is not satisfied, bid $b_c^t(k)=\left\{
\emptyset, 0 \right\}$. If $\exists d \in \mathcal{D}_k$ satisfies
$\forall b_c^t(k) \in \mathcal{B}^t, \mathcal{D}_k \notin b_c^t(k)$,
it reveals that the supply exceeds the demand. Then, the BS sets
$t=t+1$, $p^{t+1}(d)=p^t(d)-\Delta$ where $d$ is the over-supplied
item, and the auction moves on to the next round.

In a normal case, as long as the price of a package decreases below
a bidder's valuation for that package, the bidder submits a bid for
it. The BS allocates the package to the bidder, and fixes the
corresponding prices of items. At the same time, constrained by the
XOR bidding language, the bidder is not allowed to participate the
following auction rounds. As the asking prices decrease discretely
every round, it may exist a situation that more than one bidders bid
for packages containing the same items simultaneously. The BS
detects the bids of all the bidders: 1) it exists
$b_{c_1}^t(k)=b_{c_2}^t(k)\neq\left\{ \emptyset,0\right\}$ ($c_1
\neq c_2,k\in \left\{1,2,\ldots,N \right\}$); 2) it exists
$b_{c_1}^t(k_1)=\left\{
\mathcal{D}_{k_1},\mathcal{P}_{c_1}^t(k_1)\right\}$,
$b_{c_2}^t(k_2)=\left\{
\mathcal{D}_{k_2},\mathcal{P}_{c_2}^t(k_2)\right\}$ ($k_1\neq
k_2,c_1,c_2\in \left\{ 1,2,\ldots,C\right\}$) satisfying
$\mathcal{D}_{k_1}\cap\mathcal{D}_{k_2}\neq \emptyset$. If either of
the above conditions is satisfied, the overall demand exceeds supply
for at least one item. Then, the BS sets a fine tuning
$p^t(d)=p^t(d)+\delta$ where $d$ is the temporary over-demanded
item, and $\delta$ can be set by $\delta=\Delta/i$ where $i$ is an
integer factor that affects the convergence rate. The allocation can
be determined by multiple iterations.

The auction continues until all the D2D links are auctioned off or
every channel wins a package. Our algorithm is detailed in
Table~\ref{table_1}.

\begin{table}[!t]
\renewcommand{\arraystretch}{1.3}
\caption{The resource allocation algorithm} \label{table_1}
\centering
\begin{tabular}{p{100mm}}
\hline

$\ast$ \textbf{Initial State:}

\quad The BS collects the location information of all D2D pairs. The
valuation of the $c$-th resource unit for package $k$ is $v_c(k), c
= 1,2,\ldots,C, k = 1,2,\ldots,N$, which is given by
(\ref{eq:value}). The round index
$t=0$, and the initial price $P^0(d)$, the fixed price reduction $\Delta >0$ are set up.\\

$\ast$ \textbf{Resource Allocation Algorithm:}

1. Bidder $c$ submits bids $\left\{
\mathcal{D}_k,\mathcal{P}_c(k)\right\}$ depending on its utility.

\quad $\star$ bidder $c$ bids for package $\mathcal{D}_k$ as long as
$\mathcal{U}_c(k)\geq0$, which is represented by
(\ref{eq:bid_condition}).

\quad $\star$ If $\mathcal{U}_c(k)<0$, bidder $c$ submits $\left\{
\emptyset,0 \right\}$.

2. If $\exists d \in \mathcal{D}_k$ satisfies $\forall b_c^t(k) \in
\mathcal{B}^t, \mathcal{D}_k \notin b_c^t(k)$, the BS sets $t=t+1$,
$p^{t+1}(d)=p^t(d)-\Delta$ where $d$ is the over-supplied item, and
the auction moves on to the next round. Return to step 1.

3. The BS detects the bids of all the bidders:

\quad \quad 1) it exists $b_{c_1}^t(k)=b_{c_2}^t(k)\neq\left\{
\emptyset,0\right\}$ ($c_1 \neq c_2,k\in \left\{1,2,\ldots,N
\right\}$);

\quad \quad 2) it exists $b_{c_1}^t(k_1)=\left\{
\mathcal{D}_{k_1},\mathcal{P}_{c_1}^t(k_1)\right\}$,
$b_{c_2}^t(k_2)=\left\{
\mathcal{D}_{k_2},\mathcal{P}_{c_2}^t(k_2)\right\}$ ($k_1\neq
k_2,c_1,c_2\in \left\{ 1,2,\ldots,C\right\}$) satisfying
$\mathcal{D}_{k_1}\cap\mathcal{D}_{k_2}\neq \emptyset$.

4. If neither of the conditions in step 3 is satisfied, go to step
5. Otherwise, the overall demand exceeds supply for at least one
item. The BS sets $p^t(d)=p^t(d)+\delta$, and $\delta$ can be set by
$\delta=\Delta/i$ where $i$ is an integer factor. Return to step 1.

5. The allocation can be determined by repeating the above steps.
The auction continues until all D2D links are auctioned off
or every cellular channel wins a package.\\

\hline
\end{tabular}
\end{table}

\section{Analysis of the Auction-Based Resource Allocation
Algorithm}\label{sec:analysis}

In this section, we investigate the important properties of the
proposed auction-based resource allocation mechanism.

{\bf{1) Cheat-Proof}}

As the general definition, cheat-proof means that reporting the true
demand in each auction round is a best response for each bidder.
\newtheorem{proposition}{Proposition}
\begin{proposition}
The resource allocation algorithm based on
the reverse I-CA is cheat-proof.
\end{proposition}
\newtheorem{proof}{Proof}
\begin{proof} From (\ref{eq:utility_JSAC}), we can get that the utility
of bidder $\mathcal{U}_c(k)$ depends on the valuation of the package
it bids and unit prices of the items. In details, it is the
interference (between cellular and D2D communications) that mainly
affects the utility. As the expression is extremely complex to
resolve, we consider the case that only one item constitutes the
package without loss of generality. The utility of bidder $c$ can be
rewritten as \vspace{-0.2cm}
\begin{align}
\mathcal{U}_c (d) = & \log _2 \hspace{-0.1cm} \left( {1 + \frac{{P_B
h_{Bc}^2 }}{{P_d h_{dc}^2 + N_0 }}} \right) + \log _2
\hspace{-0.1cm} \left( {1 + \frac{{P_d h_{dd}^2
}}{{P_B h_{Bd}^2  + N_0 }}} \right) \nonumber \\
& - \log _2 \left( {1 + \frac{{P_B h_{Bc}^2 }}{{N_0 }}} \right) -
p^t\left( d \right),
\end{align}
and the differential expressions of the utility with respect to
$h_{dc}$ and $h_{Bd}$ are \vspace{-0.2cm}
\begin{equation}\label{eq:differential_1}
\frac{{\partial \mathcal{U}_c \left( d \right)}}{{\partial h_{dc} }}
\hspace{-0.1cm} = \hspace{-0.1cm} \frac{{ - 2P_d h_{dc} P_B h_{Bc}^2
}}{{\ln 2\left( {P_d h_{dc}^2 \hspace{-0.1cm} + \hspace{-0.1cm} P_B
h_{Bc}^2 \hspace{-0.1cm} + \hspace{-0.1cm} N_0 } \right)\left( {P_d
h_{dc}^2 \hspace{-0.1cm} + \hspace{-0.1cm} N_0 } \right)}} < 0,
\end{equation}
\vspace{-0.2cm}
\begin{equation}\label{eq:differential_2}
\frac{{\partial \mathcal{U}_c \left( d \right)}}{{\partial h_{Bd} }}
\hspace{-0.1cm} = \hspace{-0.1cm} \frac{{ - 2P_B h_{Bd} P_d h_{dd}^2
}}{{\ln 2 \hspace{-0.05cm} \left( \hspace{-0.05cm} {P_B h_{Bd}^2
\hspace{-0.1cm} + \hspace{-0.1cm} P_d h_{dd}^2 \hspace{-0.1cm} +
\hspace{-0.1cm} N_0 } \right)\left( \hspace{-0.05cm} {P_B h_{Bd}^2
\hspace{-0.1cm} + \hspace{-0.1cm} N_0 } \right)}} < 0,
\end{equation}
respectively. Accordingly, utility $\mathcal{U}_c(d)$ is a
monotonically decreasing function with respect to both $h_{dc}$ and
$h_{Bd}$. Thus, the optimal strategy is to bid the D2D link that has
a lower channel gain with the cellular transmitter and receiver.

In a descending price auction, items are always too expensive to
afford at the beginning. With the number of iterations $t$
increasing, the prices of items drop off. Given a package
$\mathcal{D}_k$ in round $t$, bidder $c$ has the right to submit bid
$\left\{ \mathcal{D}_k,\mathcal{P}_c^t(k)\right\}$ or $\left\{
\emptyset,0 \right\}$. Given that all the other bidders submit their
true demands according to (\ref{eq:bid_condition}), we consider the
strategy of bidder $c$ in two cases: 1) if $c$ bids $\left\{
\emptyset,0 \right\}$ when its true valuation for $\mathcal{D}_k$
satisfies $\mathcal{U}_c(k) \geq 0$, it will quit this round and
lose the package which maximizes its channel rate; 2) if $c$ bids
$\left\{ \mathcal{D}_k,\mathcal{P}_c^t(k)\right\}$ when its true
valuation for $\mathcal{D}_k$ satisfies $\mathcal{U}_c(k) < 0$ and
finally wins the package, it will obviously get a negative surplus
that is unwanted.

From the above analysis, we can conclude that the optimal strategy
for cellular channel $c$ is to submit its true demand in each round,
or it will get a loss in its utility as a result of any deceiving.
That is, the proposed resource allocation algorithm is cheat-proof.
\end{proof}

{\bf{2) Convergence}}\label{subsec:conv}

In this subsection, we prove that the proposed algorithm has the
convergence property.

\begin{proposition}
The resource allocation algorithm based on
the reverse I-CA has the convergence property that the number of the
iterations is finite.
\end{proposition}
\begin{proof} According to Theorem 1, all the bidders submit their
true demands in each auction round, in order to obtain the utility
from winning. From (\ref{eq:utility_JSAC}), we can derive \vspace{-0.2cm}
\begin{equation}
\mathcal{U}_c^{t+1}-\mathcal{U}_c^{t}=\Delta>0,
\end{equation}
where $\mathcal{U}_c^{t}$ denotes the utility of bidder $c$ in round
$t$. According to the algorithm, we have that bidder $c$ will get
zero utility with no bid if $\mathcal{U}_c^t<0$, and have an
opportunity to win a positive utility with bid $\left\{
\mathcal{D}_k,\mathcal{P}_c^t(k)\right\}$ if $\mathcal{U}_c^t \geq
0$. Therefore, in the beginning, bidder $c$ plays a waiting game,
and once $\mathcal{U}_c^t(k) \geq 0$, it will bid for
$\mathcal{D}_k$. As long as it is the only one that submits a bid,
it will get the package. With a sufficiently large $t$ and
$\Delta>0$, we can finally get $x_c(k)=1$. Similarly, if more than
one bidders bid for the same item, we can have an allocation by
ascending price process with the step $\delta<\Delta$. Subjected to
$\sum\limits_{\mathcal {D}_k :d \in \mathcal {D}_k } {\sum\limits_{c
= 1}^C {x_c (k)}  \le 1}$ in (\ref{eq:CAP}), the package can not be
sold once more. Thus, for a finite number of packages $N$, the
number of iterations is finite. That is, the proposed scheme would
reach convergence.
\end{proof}

In addition, the value of the price step $\Delta$ has a direct
impact on the speed of convergence of the proposed scheme. The
scheme converges fast when $\Delta$ is large, while it converges
slowly when $\Delta$ is small. The fine tuning $\delta$ also has the
same nature, but less impact on convergence.

{\bf{3) Price Monotonicity}}

In an I-CA game, the price updates through several ways, i.e.,
monotonically increasing, monotonically decreasing and non-monotonic
modes. Here, we focus on the price non-monotonicity in the proposed
reverse I-CA algorithm.

\begin{proposition}
In the proposed descending price auction,
the raising item prices in a round may be necessary to reflect the
competitive situation. Moreover, it brings efficiency improvement.
\end{proposition}
\begin{proof}
From the algorithm proposed in Table \ref{table_1}, there exists a situation that more than one
bidders submit bids for the same package or different packages with
intersection when prices are reduced to some certain values. But
auctions do not allow one item being obtained by multiple bidders as
the second constraint in (\ref{eq:CAP}) shows. In this situation,
raising the corresponding prices by a fine tuning $\delta=\Delta/i$
makes bidders to reinspect their utility functions. Once a bidder
finds its utility less than zero, it quits from the competition. By
a finite number of iterations, the winner converges to one bidder.
Since the ascending price process maximizes the auctioneer revenue
as shown in (\ref{eq:WDP}), the allocation has higher efficiency
than a random allocation in that situation.
\end{proof}

{\bf{4) Complexity}}

As mentioned before, a traditional CAP in fact is an NP-hard
problem, the normal solution of which is the centralized exhaustive
search. We set that the number of items to be allocated is $m$, and
the number of bidders is $n$. For an exhaustive optimal algorithm,
an item can be allocated with $n$ possible results. Thus, all the
$m$ items are allocated with $n^m$ possible results. The complexity
of the algorithm can be denoted by $\mathcal {O}(n^m)$. In the
proposed reverse I-CA scheme, bidders reveal their entire utility
function, i.e., they calculate valuations for all possible packages,
the number of which is
$\mathcal{C}_m^1+\mathcal{C}_m^2+\cdots+\mathcal{C}_m^m=2^m - 1$. If
the total number of iterations is $t$, the complexity of the
auction-based scheme is $\mathcal{O}(n\left( {2^m - 1} \right)+t)$.
From the proposed algorithm, we have $p^t \left( d \right) = p^0
\left( d \right) - \Delta  \cdot t \ge 0 $ (The fine tuning has a
small impact on the result and can be omitted here). So the worst
case is $t = {{p^0 \left( d \right)} \mathord{\left/
 {\vphantom {{p^0 \left( d \right)} \Delta }} \right.
 \kern-\nulldelimiterspace} \Delta }$. It is obvious that for sufficient large values of $m$ and $n$,
general values of $p^0(d)$ and $\Delta$, a much lower complexity is
obtained by using the proposed reverse I-CA scheme. That is,
$\mathcal{O}(n^m)>\mathcal{O}(n\left(2^m-1\right)+{{p^0 \left( d
\right)} \mathord{\left/
 {\vphantom {{p^0 \left( d \right)} \Delta }} \right.
 \kern-\nulldelimiterspace} \Delta })$. If we constrain the
number of D2D pairs sharing the same channel to one, the complexity
would be further reduced to $\mathcal{O}(n \cdot m+{{p^0 \left( d
\right)} \mathord{\left/
 {\vphantom {{p^0 \left( d \right)} \Delta }} \right.
 \kern-\nulldelimiterspace} \Delta })$. And the performance of this
reduced scheme is included in the simulation in Section \ref{sec:simulations_JSAC}.

{\bf{5) Overhead}}

In D2D underlay system, the BS is still the control center of
resource allocation, and the global CSI should indeed be available
at the BS for the proposed scheme. In addition to the CSI detection,
feedback, and the control signaling transmission, the reverse I-CA
scheme does not need additional signaling overhead compared to
existing resource scheduling schemes such as maximum carrier to
interference (Max C/I) and proportional fair (PF), which also need
the global CSI to optimize the system performance. The difference is
that the reverse I-CA scheme requires more complicated CSI due to
the interference between D2D and cellular network.

At the beginning of the allocation, the transmitters need to send
some packets containing detection signals. Then, the obtained CSI at
each terminal (D2D or cellular receiver) would be feedback to the
BS. After that, iteration process would be conducted at the BS, and
no signaling needs to be exchanged among the network nodes until the
control signals forwarding.

Methods, such as CSI feedback compression and signal flooding, would
help reduce the overhead. In addition, the future work on D2D
communication could consider some mechanism that limit the number of
D2D pairs sharing the same channel by, e.g. distance constraint,
which would obviously help reduce the overhead. But for this chapter,
the target is to obtain the nearest-optimal solution, wherefore we
do not consider the simplification.

\section{Performance Results and Discussions}\label{sec:simulations_JSAC}

The main simulation
parameters are listed in Table \ref{table_2}. As shown in Fig.
\ref{fig:1}, simulations are carried out in a single cell. Both
path-loss model and shadow fading are considered for cellular and
D2D links. The wireless propagation is modeled according to WINNER
II channel models \cite{WINNER}, and D2D channel is based on
office/indoor scenario while cellular channel is based on the urban
microcell scenario.
\begin{table}
\begin{center}
\caption{Main Simulation Parameters}\label{table_2}
\begin{tabular}{|l|l|}
\hline \bf{Parameter} & \bf{Value}\\ \hline Cellular layout &
Isolated cell, 1-sector\\ \hline System area & The radius of the
cell is 500 m\\ \hline  Noise spectral density & -174 dBm/Hz\\
\hline Sub-carrier bandwidth & 15 kHz \\ \hline Noise figure & 9 dB at device\\
\hline Antenna gains & BS: 14 dBi; Device: 0 dBi\\ \hline
The maximum distance of D2D & 5 m\\
\hline Transmit power & BS: 46 dBm; Device: 23 dBm \\ \hline
\end{tabular}
\end{center}
\end{table}

{\bf{1) System Sum Rate}}

The system sum rate with different numbers of D2D pairs and
different numbers of resource units using the proposed auction
algorithm is illustrated in Fig. \ref{fig:2} $\sim$ Fig.
\ref{fig:4}. The sum rate can be obtained from (\ref{eq:Re}).

From Fig. \ref{fig:2} and Fig. \ref{fig:4}, we can see that the
system sum rate goes up with both the number of D2D pairs and the
number of resource units increasing. On one side, when the amount of
resources is fixed, more D2D users contribute to a higher system sum
rate. On the other side, as the amount of resource increases, the
probability of resources with less interference to D2D links being
assigned to them enhances, which can lead to the increased sum rate.
This phenomenon is similar to the effect of multiuser diversity.
Definitely, cellular users also contributes to the performance.

From another perspective, Fig. \ref{fig:2} $\sim$ Fig. \ref{fig:4}
shows the system sum rate for different allocation algorithms. The
curve marked exhaustive optimal is simulated by the exhaustive
search way, which guarantees a top bound of the system sum rate. The
curve marked reduced R-I-CA is the result of a reduced reverse I-CA
scheme, in which the number of D2D pairs sharing the same cellular
resources is constrained to one.
\begin{figure}[!ht]
\centering
\includegraphics[height=3.3in]{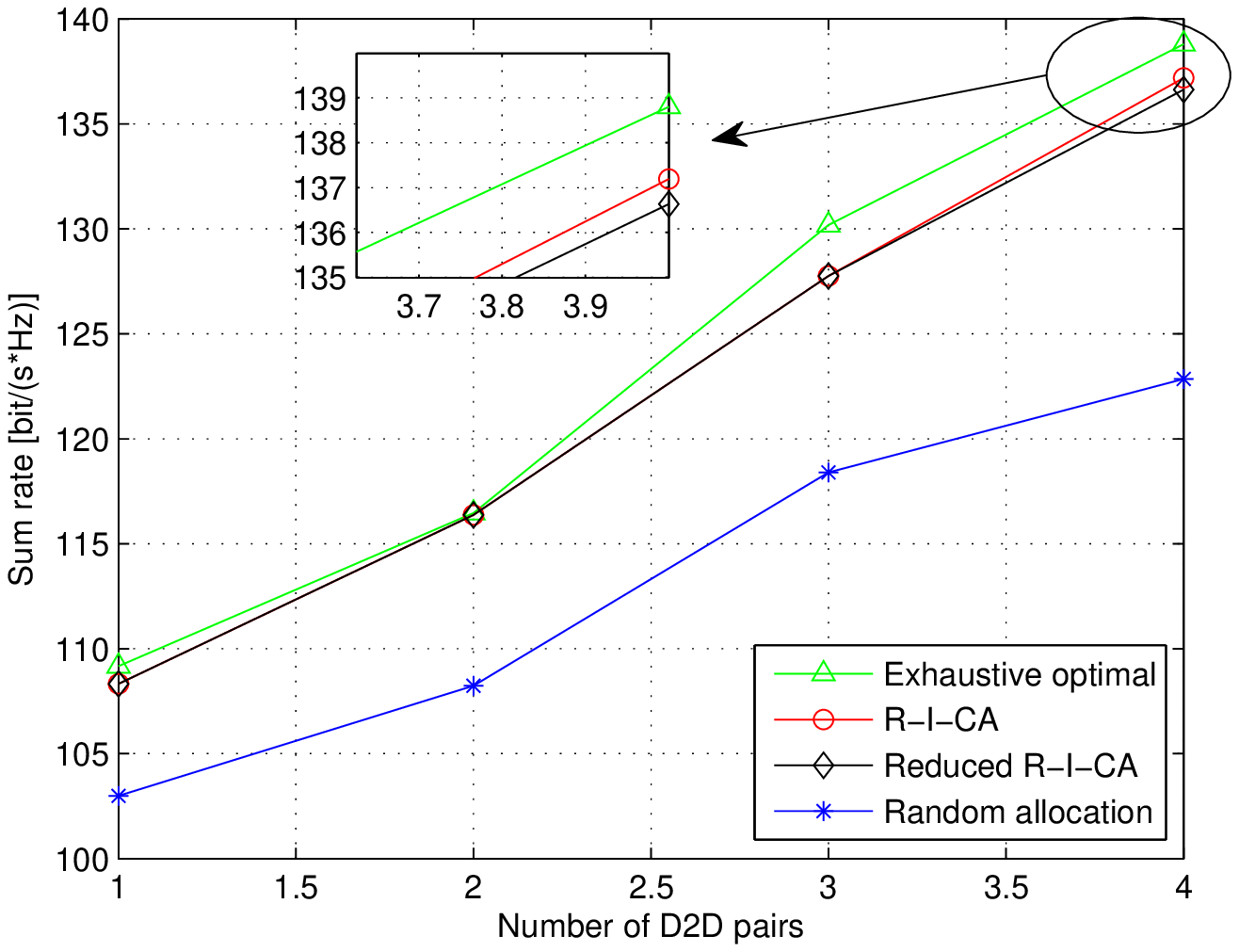}
\caption{System sum rate for different allocation algorithms in the
case of 8 resource units.} \label{fig:2}
\end{figure}
\begin{figure}[!ht]
\centering
\includegraphics[height=3.3in]{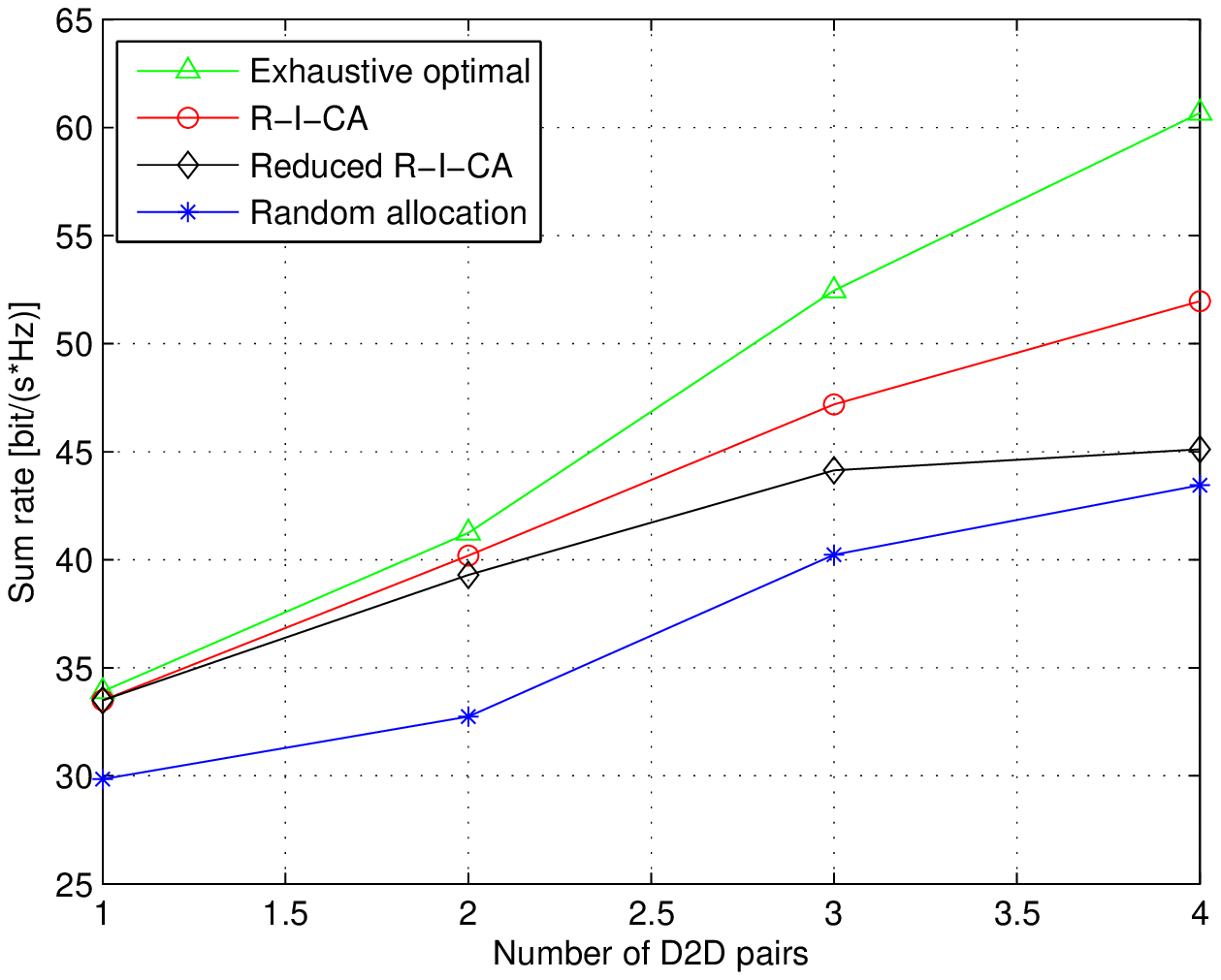}
\caption{System sum rate for different allocation algorithms in the
case of 2 resource units.} \label{fig:3}
\end{figure}
\begin{figure}[!ht]
\centering
\includegraphics[height=3.3in]{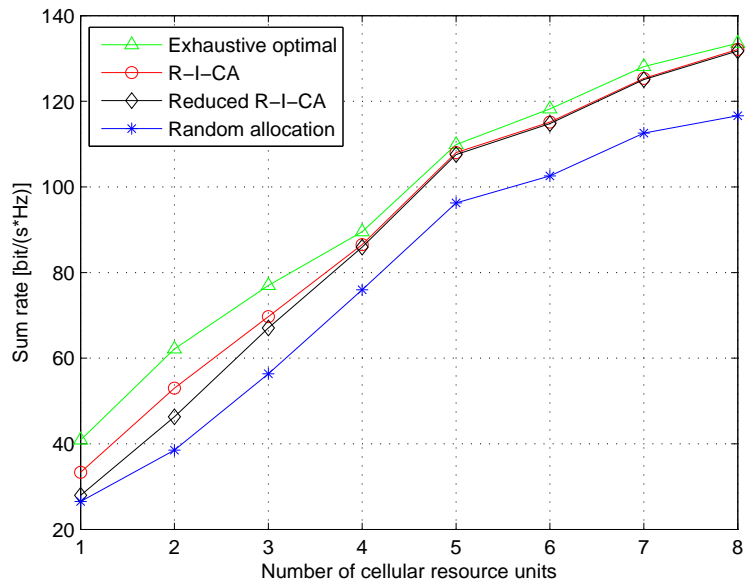}
\caption{System sum rate for different allocation algorithms in the
case of 4 D2D pairs.} \label{fig:4}
\end{figure}
The curve marked R-I-CA represents the performance of the proposed
reverse I-CA algorithm, and the last one is the simulation result
using random allocation of spectrum resources. Firstly, we can see
that the proposed auction algorithm is relatively much superior to
the random allocation. Secondly, the optimal allocation results in
the highest system sum rate, but the superiority compared to R-I-CA
is quite small, especially when the number of cellular resource
units increases as Fig. \ref{fig:4} shows. Moreover, we find that
the performance of reduced R-I-CA approximates to that of R-I-CA
scheme in case of 8 resource units, but differs obviously in case of
2 resource units shown in Fig. \ref{fig:3}. The reason for this
phenomenon is that the constraint of the reduced R-I-CA limits D2D
pairs accessing to the network when the number of resources units is
less than that of D2D pairs, thus a large capacity loss products.

\begin{figure}[!ht]
\centering
\includegraphics[height=3.3in]{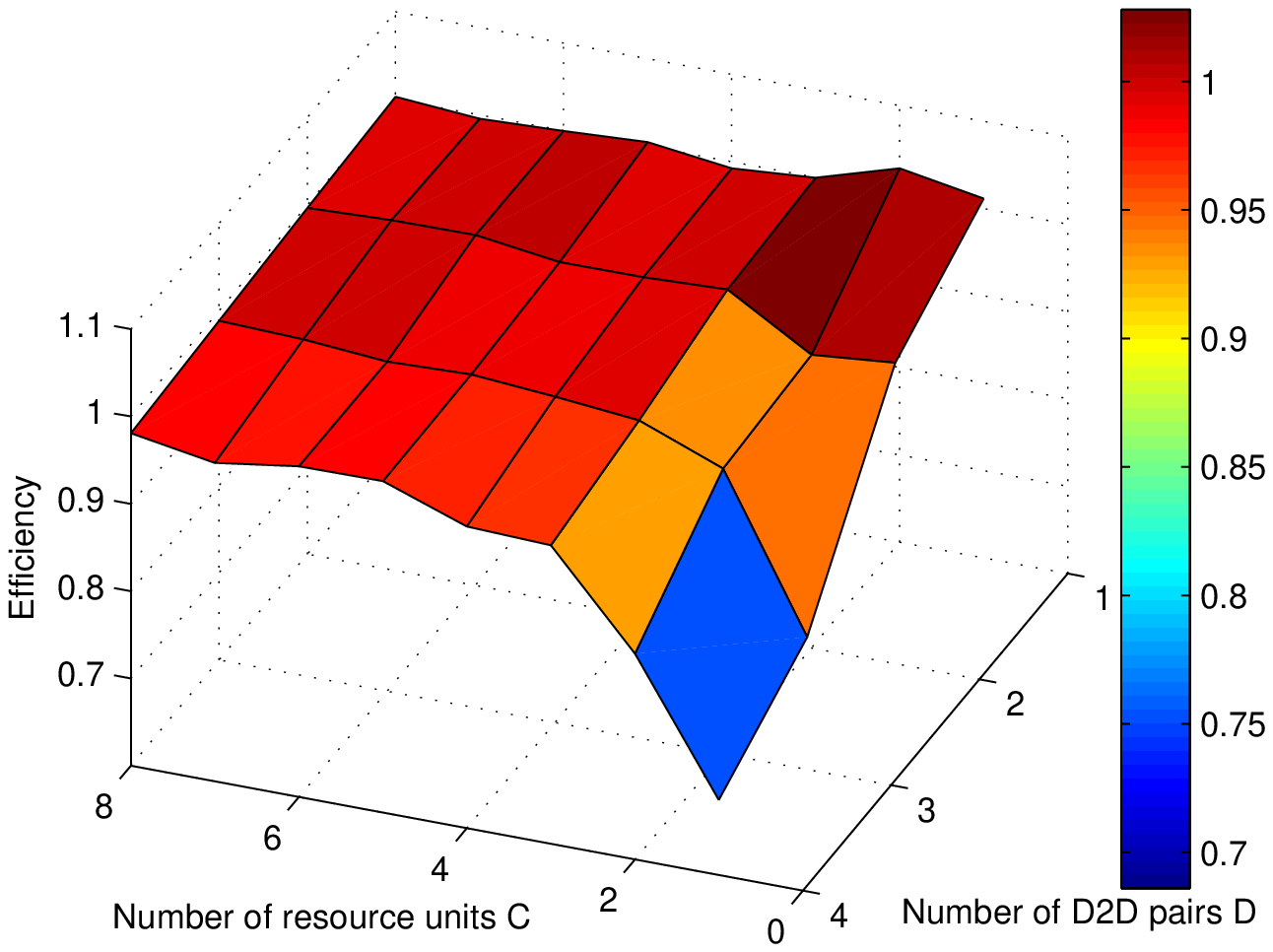}
\caption{System efficiency: $\eta$ with different numbers of D2D
pairs and different numbers of resource units.} \label{fig:5}
\end{figure}

{\bf{2) System Efficiency}}

We define the system efficiency as $\eta  = {\Re \mathord{\left/
 {\vphantom {\Re  {{\Re _{opt}}}}} \right.
 \kern-\nulldelimiterspace} {{\Re _{opt}}}}$, where $\Re_{opt}$
 represents the exhaustive optimal sum rate.
Fig. \ref{fig:5} shows the system efficiency with different numbers
of D2D pairs and different numbers of resource units. The simulation
result indicates that the proposed algorithm provides high (the
lowest value of $\eta$ is around 0.7) system efficiency. Moreover,
the efficiency is stable over different parameters of users and
resources.

As to the point of efficiency value being about 0.7, the number of
resource units and the number of D2D pairs are both small. The
linear price rule limits bidders to bid the maximal valuation
packages, but to bid the packages having maximal average unit
valuation. For this reason, the efficiency decreases slightly.

As to other points, the efficiency is stable above 0.9, which
reflects a small performance gap between the proposed algorithm and
the exhaustive search scheme. In fact, the descending price rule
determines the bidder that has the highest bid on current items
would win the corresponding package, which maximizes the current
overall gain. However, the gap cannot be avoided as the algorithm
essentially follows a local, or an approximate global optimum
principle.

{\bf{3) Price Monotonicity}}\label{subsec:monotonicity}

Fig. \ref{fig:6} shows an example of the price non-monotonicity in
the reverse I-CA scheme. The four curves represent unit price of
four D2D pairs. As the enlarged detail shows, the unit price of D2D
pair $2$ has an ascending process during the auction. As the step
$\delta$ is much less than descending step $\Delta$, the phenomenon
of ascending price is hard to pick out. When the items have been
sold out, their prices are fixed to the selling value. And from the
figure, we can find that the D2D pair $2$ is the last one to be
sold.
\begin{figure}[!ht]
\centering
\includegraphics[height=3.3in]{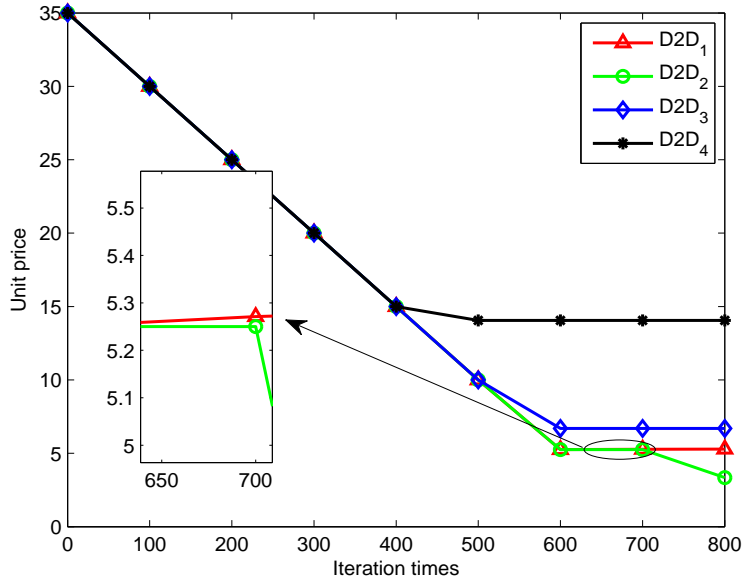}
\caption{Price monotonicity: an example of price non-monotonicity in
the reverse I-CA scheme.} \label{fig:6}
\end{figure}

In this chapter, we have investigated how to reduce the effects of
interference between D2D and cellular users, in order to improve the
system sum rate for a D2D underlay network. We have proposed the
reverse iterative combinatorial auction as the mechanism to allocate
the spectrum resources for D2D communications with multiple user
pairs. We have formulated the valuation of each D2D pair for each
resource unit, and then explained a detailed auction algorithm
depending on the utility function. A non-monotonic descending price
iteration process has been modeled and analyzed to be cheat-proof,
converge in a finite number of rounds, and has low complexity. The
simulation results show that the system sum rate goes up with both
the number of D2D pairs and the number of resource units increasing.
The proposed auction algorithm is much superior to the random
allocation, and provides high system efficiency, which is stable
over different parameters of users and resources.

\chapter{Cross-layer Optimization}\label{chap:cross}
D2D communication as an underlay to cellular
networks brings significant benefits to system throughput and
energy efficiency. However, as D2D UEs can cause interference to
cellular UEs, the scheduling and allocation of channel resources and
power to D2D communication need elaborate coordination. Thus, cross-layer
optimization is considered in this chapter. In Section \ref{sec:scheduling},
both system throughput and fairness are taken into account. A joint time-domain
scheduling and spectrum allocation scheme is studied. Based on the concept of
UE's battery lifetime, we propose an auction-based algorithm for power and channel
resource allocation of D2D communication in Section \ref{sec:energy}.

\section{Time-Domain Scheduling}\label{sec:scheduling}
In this section, we propose a joint scheduling and resource allocation scheme
to improve the performance of D2D communication. We take network
throughput and UEs' fairness into account by performing interference
management. Specifically, we develop a Stackelberg game framework in
which we group a cellular UE and a D2D UE to form a leader-follower
pair. The cellular user is the leader, and the D2D UE is the
follower who buys channel resources from the leader. We analyze the
equilibrium of the game, and propose an algorithm for joint
scheduling and resource allocation. Finally, we perform computer
simulations to study the performance of the proposed algorithm.

\subsection{Model Assumptions}\label{sec:model}

We still consider a single cell scenario with multiple UEs and one eNB
located at the center of the cell. Both the UEs and the eNB are
equipped with a single omni-directional antenna. The system includes
two types of UEs, D2D UEs and cellular UEs. The D2D UEs are in
pairs, each consisting of one transmitter and one receiver. We
consider a dense D2D environment, where the number of cellular UEs
and D2D UEs is $K$ and $D$ $(D>K)$, respectively. The set of
cellular UEs and D2D pairs is $\mathcal{K}$ and $\mathcal{D}$,
respectively. There are $K$ orthogonal channels, which is occupied
by the corresponding cellular UEs. The channels allocated to the
cellular UEs are fixed, and D2D communications share the channels
with cellular UEs. One channel is only allowed to be used by one
cellular UE and one D2D UE. In LTE, scheduling takes place every
transmission time interval (TTI) \cite{Song2010}, which consists of
two time slots. Channels are allocated among D2D UEs according to
their priority. During each TTI, $K$ D2D pairs are selected to share
the $K$ channels with the cellular UEs while other D2D UEs wait for
transmission.

\begin{figure}[!t]
\centering
\includegraphics[height=3.8in]{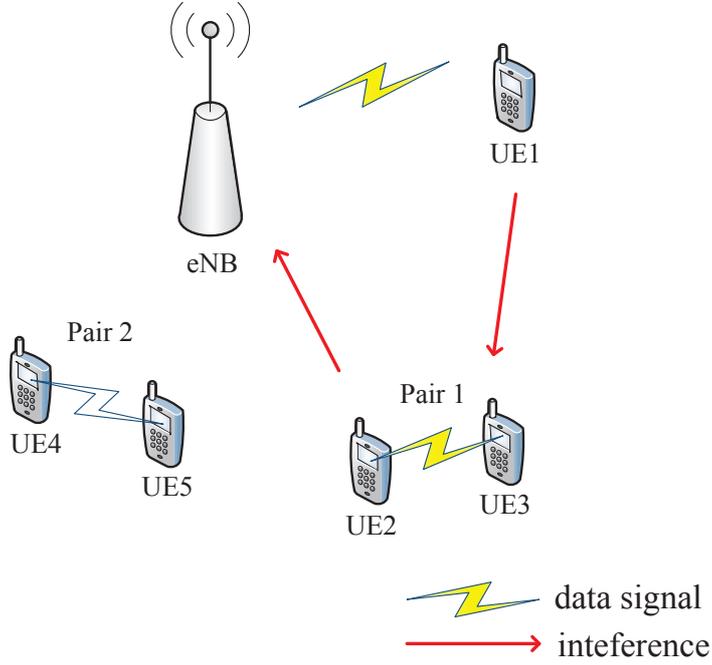}
\caption{System model of D2D underlay communication with uplink
resource sharing. UE$_1$ is a cellular user whereas UE$_2$ and
UE$_3$, UE$_4$ and UE$_5$ are D2D UEs.} \label{fig:system_model}
\end{figure}

A scenario of uplink resource sharing is illustrated in
Fig.~\ref{fig:system_model}, which includes one cellular UE (UE$_1$)
and two D2D pairs (UE$_2$ and UE$_3$, UE$_4$ and UE$_5$). UE$_2$ and
UE$_4$ are transmitters while UE$_3$ and UE$_5$ are receivers. The
two D2D UEs in a pair are close enough to satisfy the maximum
distance constraint of D2D communication, in order to guarantee the
quality of D2D services. D2D pair 1 is selected to share the channel
resource with UE$_1$ whereas D2D pair 2 cannot transmit. During the
uplink period of the cellular network, UE$_1$ transmits data to the
eNB, while the eNB suffers interference from UE$_2$. Also, the D2D
pair 1 is in communication while UE$_3$ are exposed to interference
from UE$_1$.

During the downlink period, the cellular UE receives data from the eNB and interference from D2D transmitter sharing the same channel. The D2D receivers suffer interference from the eNB. The transmit power of eNB is too large, causing serious interference to the D2D UEs, so it is hard to guarantee the quality of D2D services. Therefore, we focus on the uplink frame of the network.

We define a set of binary variables $\{x_{ik}\}(i\in\mathcal{D}, k\in\mathcal{K})$ to denote the current D2D pair in communication. $x_{ik}=1$ if the $i$-th D2D pair is selected to use channel $k$; and $x_{ik}=0$ otherwise. Based on the analysis above, during the uplink period, the $k$-th cellular UE transmits $s_k$ to the eNB, and the $i$-th D2D transmitter transmits $s_i$. The received signals at the eNB and D2D receiver $i$ are written as
\begin{equation}
  y_k^c = \sqrt{p_k g_{ke}}s_k + \sum_{i=1}^D x_{ik} \sqrt{p_i g_{ie}}s_i + n_k,
\end{equation}
\begin{equation}
  y_i^d = \sqrt{p_i g_{ii}}s_i+\sum_{k=1}^K x_{ik} \sqrt{p_k g_{ki}}s_k + n_i,
\end{equation}
where $p_k$ and $p_i$ is the transmit power of the $k$-th cellular UE and the $i$-th D2D transmitter, respectively. $g_{ki}$ denotes the channel gain between the $k$-th cellular UE and the $i$-th D2D receiver. $g_{ii}$ denotes the channel gain between the $i$-th D2D transmitter and the $i$-th D2D receiver, which are in a pair. $g_{ke}$ is the channel gain between cellular UE $k$ and the eNB, and $g_{ie}$ is the channel gain between D2D transmitter $i$ and the eNB. $n_k$ and $n_i$ is the additive white Gaussian noise (AWGN). Without loss of generality, we assume all UEs observe the same noise power $N_0$.

The received SINR at the $i$-th D2D receiver can be expressed as
\begin{equation}
  \gamma_i^d = \frac{p_i g_{ii}}{\sum_k x_{ik} p_k g_{ki}+N_0},
\end{equation}
The SINR at the eNB corresponding to cellular UE $k$ is
\begin{equation}
  \gamma_{k}^c = \frac{p_k g_{ke}}{\sum_{i} x_{ik} p_i g_{ie}+N_0}.
\end{equation}
The channel rate of UEs is obtained by
\begin{equation}
  r = \log_2(1+\gamma).
\end{equation}

\subsection{Stackelberg Game}\label{sec:problem}
As D2D communication takes place underlaying the primary cellular network, we focus on the power control and scheduling of the D2D UEs, while the transmit power and channel of the cellular UE are assumed to be fixed. D2D communication can utilize the proximity between UEs to improve the throughput performance of the system. In the meanwhile, the interference from D2D network to the cellular network should be limited. Thus, the transmit power of the D2D UEs should be properly controlled. Another goal is to guarantee the fairness among D2D UEs when scheduling.

Interactions among the selfish cellular and D2D UEs sharing a channel can be modeled as a non-cooperative game using the game theory framework. When the players choose their strategies independently without any coordination, it usually leads to an inefficient outcome. If we simply model the D2D scenario as a noncooperative game, D2D UEs will choose to use the maximum transmit power to maximize their own payoffs regardless of other players, whereas cellular UEs will choose not to share the channel resources with D2D UEs. This is an inefficient outcome, as the interference is too serious or D2D cannot get access into the network.

Therefore, we employ the Stackelberg game to coordinate the scheduling, in which the cellular UEs are the leaders and the D2D UEs are the followers. We focus on the behavior of a one-leader one-follower pair. The leader owns the channel resource and it can charge the D2D UE some fees for using the channels. The fees are fictitious money to coordinate the system. Thus, the cellular UE has an incentive to share the channel with the D2D UE if it is profitable, and the leader has the right to decide the price. For the D2D UE, under the charging price, it can choose the optimal power to maximize its payoff. In this way, an equilibrium can be reached.

{\bf{1) Utility Functions}}

We analyze the behavior of a one-leader one-follower pair, which includes cellular UE $k$ as the leader, and D2D pair $i$ as the follower. The utility of the leader can be defined as its own throughput performance plus the gain it earns from the follower. The fee should be decided according to the leader's own consideration. Thus, we set the fee proportional to the interference the leader observe. The utility function of the leader can be expressed as
\begin{equation}
  u_k(\alpha_k,p_i) = \log_2\left(1+\frac{p_k g_{ke}}{p_i g_{ie}+N_0}\right) + \alpha_k\beta p_i g_{ie},
\end{equation}
where $\alpha_k$ is the charging price ($\alpha_k>0$). $\beta$ is a scale factor to denote the ratio of the leader's gain and the follower's payment ($\beta>0$). $\beta$ is a key parameter to influence the outcome of the game, which we will discuss later. The optimization problem for the leader is to set a charging price that maximizes its utility, i.e.,
\begin{equation}
  \max\; u_k(\alpha_k,p_i), \quad\mathrm{s.t.}\; \alpha_k > 0.
\end{equation}

For the follower, the utility is its throughput performance minus the cost it pays for using the channel, which can be expressed as
\begin{equation}
  u_i(\alpha_k, p_i) = \log_2\left(1+\frac{p_i g_{ii}}{p_k g_{ki}+N_0}\right) - \alpha_k p_i g_{ie}.
\end{equation}
The optimization problem for the follower is to set proper transmit power to maximize its utility, i.e.,
\begin{equation}
  \max\; u_i(\alpha_k,p_i), \quad\mathrm{s.t.}\; p_{min} \leq p_i \leq p_{max}.
\end{equation}

In the Stackelberg game, the leader moves first and the follower moves sequentially, i.e., the leader sets the price first, and the follower selects its best transmit power based on the price. The leader knows ex ante that the follower observes its action. The game can be solved by backward induction.

{\bf{ 2) Follower Analysis}}


Given $\alpha_k$ decided by the leader, when $p_i$ approaches 0, the utility approaches 0 as well. As $p_i$ increases, $u_i$ also increases. If $p_i$ grows too large, $u_i$ will begin to decrease since the logarithmic function grows slower than the cost. The follower wants to maximize its utility by choosing proper transmit power. The best response is derived by solving
\begin{equation}
  \frac{\partial u_i}{\partial p_i}=\frac{1}{\ln 2} \frac{g_{ii}}{p_i g_{ii}+p_k g_{ki}+N_0} - \alpha_k g_{ie} = 0.
\end{equation}
The solution is
\begin{equation}
  \hat p_i = \frac{1}{\alpha_k g_{ie}\ln 2} - \frac{p_k g_{ki} + N_0}{g_{ii}}.
\label{eq:follower_power}
\end{equation}
The second order derivative is
\begin{equation}
  \frac{\partial^2 u_i}{\partial p_i^2} = -\dfrac{1}{\ln 2} \left(\frac{g_{ii}}{p_i g_{ii}+p_k g_{ki}+N_0}\right)^2 < 0.
\end{equation}
Thus, the solution in (\ref{eq:follower_power}) is a maximum point.

From (\ref{eq:follower_power}), we know the power is monotonically decreasing with $\alpha_k$, which means when the price is higher, the amount of power bought is smaller. Note that $p_{min} \leq p_i \leq p_{max}$, thus, the best response is searched in $\{p_{min},p_{max},\hat p_i\}$.

{\bf{3) Leader Analysis}}

The leader knows ex ante that the follower will react to his price
by searching in $\{p_{min},p_{max},\hat p_i\}$. If the leader sets
the price too low, the follower will only buy $p_{max}$, and the
leader will earn more if it raises the price. Besides, the price is
restricted to be set too high, to prevent inefficient outcome. Thus,
the leader will set a price such that $p_{min} \leq \hat p_i \leq
p_{max}$. Solving the inequations, we have
\begin{equation}
  \alpha_{kmin} = \frac{g_{ii}}{(g_{ii} p_{max}+p_k g_{ki}+N_0)g_{ie} \ln 2},
\end{equation}
\begin{equation}
  \alpha_{kmax} = \frac{g_{ii}}{(g_{ii} p_{min} + p_k g_{ki}+N_0)g_{ie}\ln 2},
\end{equation}
and $\alpha_{kmin}\leq \alpha \leq \alpha_{kmax}$.

Substituting the follower's strategy (\ref{eq:follower_power}) into the leader's utility function, we get
\begin{equation}
\begin{split}
  u_k&(\alpha_k)= \frac{\beta}{\ln 2} - \alpha_k \beta g_{ie}  \frac{p_k g_{ki}+N_0}{g_{ii}}\\
  +&\log_2\left[1+{p_k g_{ke}} \left({\frac{1}{\alpha_k\ln 2} - g_{ie}\frac{p_k g_{ki} + N_0}{g_{ii}} + N_0} \right)^{-1} \right].
\end{split}
\label{eq:leader_uitlity2}
\end{equation}
There is a tradeoff between the gain from the leader itself and the gain from the follower. When the leader raises the price, it will gain less from the follower according to (\ref{eq:leader_uitlity2}), but the follower will buy less power, which will lead to an increase in the leader's rate. Therefore, there is an optimal price for the leader to ask for.

Let $A=p_k g_{ke}, B={1}/{\ln 2}, C=-g_{ie} (p_k g_{ki}+N_0)/{g_{ii}}+N_0$, we have
\begin{equation}
  u_k(\alpha_k)=\log_2\left(1+\frac{A\alpha_k}{C\alpha_k + B}\right) + (C-N_0)\beta\alpha_k + B\beta.
\end{equation}
To get the optimal price, by taking the first order derivative, we obtain
\begin{equation}
  \frac{d u_k}{d \alpha_k} = \frac{AB^2}{(C\alpha_k+B)[(A+C)\alpha_k+B]} + (C-N_0)\beta.
\end{equation}
We consider the following cases:

1) $C=0$. The first order condition is
\begin{equation}
  \frac{d u_k}{d \alpha_k}=\frac{AB}{A\alpha_k+B}-N_0\beta=0.
\end{equation}
The solution is
\begin{equation}
  \hat\alpha_k = \frac{B}{N_0\beta}-\frac{B}{A}.
\end{equation}
The second order derivative is
\begin{equation}
  \frac{d^2 u_k}{d \alpha_k^2} = -B\left(\frac{A}{A\alpha_k+B}\right)^2 <0.
\end{equation}
Note that $\alpha_{kmin} = {B}/{(p_{max}g_{ie}+N_0-C)}$ and $\alpha_{kmax} \hspace{-0.1cm}= \hspace{-0.1cm}{B}/{(p_{min}g_{ie}\hspace{-0.1cm}+\hspace{-0.1cm}N_0\hspace{-0.1cm}-\hspace{-0.1cm}C)}$. Thus, the optimal $\alpha_k$ is searched in $\{ \hat\alpha_k, \alpha_{kmin}, \alpha_{kmax} \}$.

2) $A+C=0$.
The optimal price $\alpha_k$ can be solved similarly. The solution is searched in $\{\frac{B}{A}-\frac{B}{(A+N_0)\beta}, \alpha_{kmin}, \alpha_{kmax} \}$.

If $C\neq 0$ and $A+C \neq 0$, we denote $f(\alpha_k)=(C\alpha_k+B)[(A+C)\alpha_k+B]$, which is a quadratic function of $\alpha_k$. We notice that roots for the $f(\alpha_k)$ are $\alpha_{k1}=-\frac{B}{C}$ and $\alpha_{k2}=-\frac{B}{A+C}$. From $\alpha_k \leq \alpha_{kmax}$, we get $(C-N_0)\alpha_k+B \geq 0$. Thus, $C\alpha_k+B >0$ and $(A+C)\alpha_k+B>0$. We discuss the following three cases based on the sign of $C$ and $A+C$.

3) $C>0$.
We have $\alpha_k\geq \alpha_{kmin} >0>\alpha_{k2}>\alpha_{k1}$. $f(\alpha_k)$ is monotonically increasing with $\alpha_k$ and $f(\alpha_k)>0$ for $\alpha_k>\alpha_{kmin}$. Thus, the first order derivative of the utility $u_k'(\alpha_k)$ is monotonically decreasing with $\alpha_k$ and it follows $\lim\limits_{\alpha_k\rightarrow \infty}u_k'(\alpha_k)=(C-N_0)\beta<0$. If $u_k'(\alpha_{kmin})\leq 0$, it satisfies that $u_k'(\alpha_k)\leq 0, \alpha_k\in[\alpha_{kmin},\alpha_{kmax}]$. Thus, the optimal price is $\alpha_{kmin}$. Otherwise, if $u_k'(\alpha_{kmin})>0$, there exists a unique point such that $u_k'(\alpha_k)=0$. Solving $u_k'(\alpha_k)=0$ gives
\begin{equation}
  \alpha_k = \frac{-B(A+2C) \pm \sqrt\Delta}{2C(A+C)},
\end{equation}
where $\Delta = AB^2\left[A+4C(A+C)\frac{1}{(N_0-C)\beta}\right]$. The maximum point must be the larger root or on the boundary of the feasible region of $\alpha$.

4) $C<0$ and $A+C>0$. We have $\alpha_{k2}<0<\alpha_{kmin}\leq\alpha_k<\alpha_{k1}$, and $f(\alpha_k)>0$. If $\min\, u_k'(\alpha_k)\leq 0, \alpha_k\in[\alpha_{kmin},\alpha_{kmax}]$, with $\alpha_k$ increasing, $u_0'(\alpha_k)$ is increasing, decreasing, and increasing sequentially. Thus, the maximum point is either on the maximum boundary of the feasible region, or is the smaller root of $u_k'(\alpha_k)= 0$, i.e.,
\begin{equation}
   \alpha_k = \frac{-B(A+2C) + \sqrt\Delta}{2C(A+C)}.
\end{equation}
Otherwise, $u_k(\alpha_k)$ is increasing with $\alpha_k$, the maximum point is $\alpha_{kmax}$.

5) $A+C<0$. We have $0<\alpha_k<\alpha_{k1}<\alpha_{k2}$, $f(\alpha_k)>0$ and $u_k'(\alpha_k)$ is monotonically increasing with $\alpha_k$. Thus, there does not exist an maximum point within the feasible region. By similar analysis, we derive the optimal price $\alpha_k$ is on the boundary of the feasible region.

Based on the discussion above, the optimal $\alpha_k$ can be uniquely decided. The strategies of the leader and the follower construct a Stackelberg equilibrium defined below.
\begin{definition}
A pair of strategies $(\alpha_k,p_i)$ is a Stackelberg equilibrium if no unilateral deviation in strategy by the leader or the follower is profitable, i.e.,
\begin{equation}
  u_i(\alpha_k,p_i)\geq u_i(\alpha_k,p_i'),
\end{equation}
\begin{equation}
  u_k(\alpha_k,p_i(\alpha_k))\geq u_k(\alpha_k',p_i(\alpha_k')).
\end{equation}
\end{definition}

The equilibrium is a stable outcome of the Stackelberg game where the leader and the follower compete through self-optimization and reach a point where no player wishes to deviate. The analysis of the leader and the follower above shows the existence and uniqueness of the Stackelberg equilibrium.

\subsection{Joint Scheduling and Resource Allocation}
The scheduling process is conducted at each TTI. The D2D UEs form a priority queue for each channel. During each TTI, the eNB selects $K$ D2D UEs with the highest priority for each channel sequentially and other D2D UEs have to wait.

\begin{figure}[!t]
\centering
\small
\algsetup{indent=2em}
\hrule
Joint D2D Scheduling and Resource Allocation Algorithm
\vskip 0.2em
\hrule
\vskip 0.2em
\begin{algorithmic}[1]
\STATE Given CSI, TTI $t$, the scale factor $\beta$, the fairness coefficient $\delta$, and the additional cost $c_i(t), \forall i$.
\STATE Initialize $x_{ik}=0, \forall i,k$.
\STATE Search the optimal $\alpha_{ik}^*,\forall i,k$ in
\[ \left\{ \frac{B}{p_{max}g_{ie}+N_0-C}, \frac{B}{p_{min}g_{ie}+N_0-C},\right\} \]
and
\[ \left \{ \frac{B}{N_0\beta}-\frac{B}{A}, \frac{B}{A}-\frac{B}{(A+N_0)\beta}, \frac{-B(A+2C) + \sqrt\Delta}{2C(A+C)} \right\}.
\]
\STATE Calculate the optimal power \[p_{ik}^*=\frac{1}{\alpha_{ik}^* g_{ie}\ln 2} - \frac{p_k g_{ki} + N_0}{g_{ii}}, \forall i,k. \]
\STATE Calculate priorities
\[ P_{ik} = \log_2\left(1+\frac{p_{ik}^* g_{ii}}{p_k g_{ki}+N_0}\right)-\alpha_{ik}^* g_{ie} p_{ik}^* - c_i(t), \forall i,k. \]
\STATE Sort $P_{ik}$ in descending order to form a priority queue.
\WHILE{$\sum_i x_{ik}=0, \exists k$}
    \STATE Select the head of the queue. The pair is$(i^*,k^*)$.
    \IF{$\sum_i x_{ik^*}=0$ \AND $\sum_k x_{i^*k}=0$}
        \STATE Schedule the pair $(i^*,k^*)$.
        \STATE Set $x_{i^*k^*}=1$ and $c_{i^*}(t+1)=c_{i^*}(t)+\delta u_{i^*k^*}$.
    \ENDIF
    \STATE Delete the head of the queue.
\ENDWHILE
\end{algorithmic}
\hrule
Algorthm~1.~~~Joint scheduling and resource allocation algorithm.
\end{figure}

In our Stackelberg game framework, the priority is based on the utilities of the followers, which express the satisfaction of the followers. In the design of scheduling scheme, fairness is considered as an important goal. The scheme should take the outcome in the previous TTIs into account. This  can be achieved by adjusting prices for using the channel. The follower has to pay an additional fee for using the channel at TTI $t$ if it has been selected in previous TTIs, which will lead to a decrease in the priority. The additional fee is decided by the cumulative utility of follower. The priority for follower $i$ at TTI $t$ can be defined as
\begin{equation}
  P_{ik}(t) = u_i(\alpha_k^*(t),p_i^*(t)) - c_i(t),
\end{equation}
where $\alpha_k^*(t)$ and $p_i^*(t)$ are the optimal strategy pair under the Stackelberg equilibrium at TTI $t$. $c_i(t)$ is the additional cost, and can be defined as
\begin{equation}
  c_i(t) = \sum_{\tau=0}^{t-1} \sum_{k=1}^K \delta x_{ik}(\tau)  u_i(\alpha_k^*(\tau),p_i^*(\tau)),
\end{equation}
where $\delta>0$ is the fairness coefficient. For a larger $\delta$, the cumulative utility has a larger influence on the priority. If $\delta=0$, the scheduling scheme does not take fairness into account.

Based on the above discussion, during each TTI, every cellular UE and D2D UE form a leader-follower pair and play the Stackelberg game. The optimal price and power can be decided for each pair. The priority for each pair can be calculated and they form a priority queue. Then, the eNB schedules the D2D pairs sequentially according to their order in the queue. If there is a tie, that one channel has been allocated to another D2D pair, or the D2D pair has been scheduled to another channel, the pair is skipped. When each channel is allocated to one D2D pair, the eNB record the outcome and the scheduling is over. The algorithm is summarized in Algorithm~1.

The algorithm has a low complexity, as the optimal strategy for each leader-follower pair is searched in a set with a constant number of elements. To form the priority queue with length $K\times D$, the complexity is $O(KD)$.

To evaluate the performance of the proposed algorithm, we perform several simulations.
We consider a single circular cell environment. The cellular UEs and D2D pairs are uniformly distributed in the cell. The two D2D UEs in a D2D pair are close enough to satisfy the maximum distance constraint of D2D communication. The received signal power is $P_i=P_j d_{ij}^{-2} |h_{ij}|^2$, where $P_i$ and $P_j$ are received power and transmit power, respectively. $d_{ij}$ is the distance between the transmitter and the receiver. $h_{ij}$ represents the complex Gaussian channel coefficient that satisfies $h_{ij}\sim\mathcal{CN}(0,1)$. The scheduling takes place every TTI.
Simulation parameters are summarized in TABLE \ref{table:simulation}.

\begin{table}[!t]
\renewcommand{\arraystretch}{1.3}
\caption{Simulation Parameters and Values}
\label{table:simulation}
\centering
\begin{tabular}{p{1.6in}p{1.5in}}
\hline
Parameter & Values\\
\hline
Cell layout & 1 isolated, circular\\
Cell radius & 500m\\
Number of cellular UEs & 5\\
Number of D2D pairs & 10\\
Max D2D communication distance & 50m\\
Cellular UE Tx power & 23dBm\\
D2D UE Tx power & 0dBm -- 23dBm\\
Thermal noise power density & -174dBm/Hz\\
Bandwidth & 180kHz\\
Transmission time interval & 1ms\\
\hline
\end{tabular}
\end{table}


\begin{figure}[!t]
\centering
\includegraphics[height=3.5in]{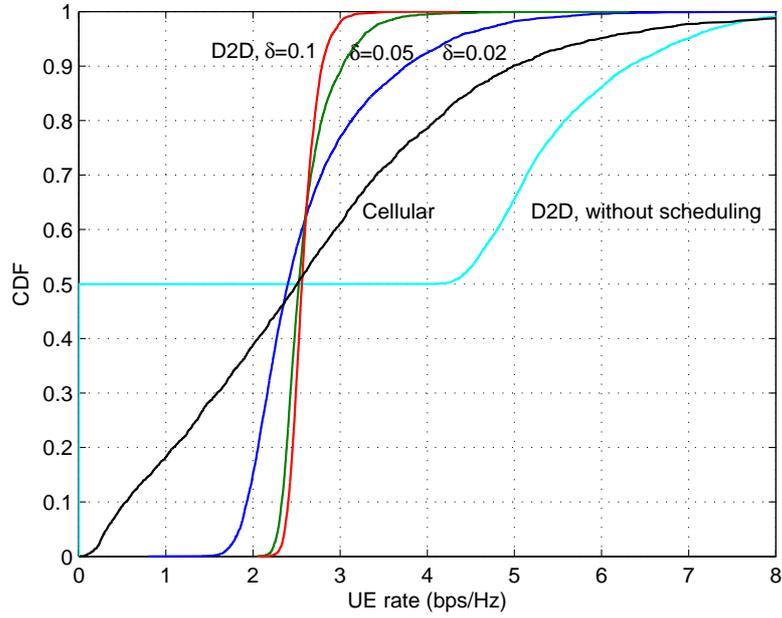}
\caption{UE rate distribution under different $\delta$.}
\label{fig:cdf_rate_delta}
\end{figure}

In Fig.~\ref{fig:cdf_rate_delta}, we study the effect of the fairness coefficient $\delta$. We plot cumulative distribution function (CDF) of UE rate. $\delta$ has little effect on the performance of the cellular UEs. For a small $\delta$, D2D UE rate is distributed in a large range, and has a tendency to converge with a larger $\delta$. Thus, scheduling with a larger $\delta$ achieves better fairness. If we set $\delta$ too large, the scheduling algorithm behaves like the Round Robin scheduling, in which the previous utility is the deciding factor and the utility of the current TTI has little influence. If D2D scheduling is not considered, there will be only $C$ D2D pairs that can get access to the network, resulting in $1-C/D$ proportion of D2D UEs cannot achieve any data transmission.

\begin{figure}[!t]
\centering
\includegraphics[height=3.5in]{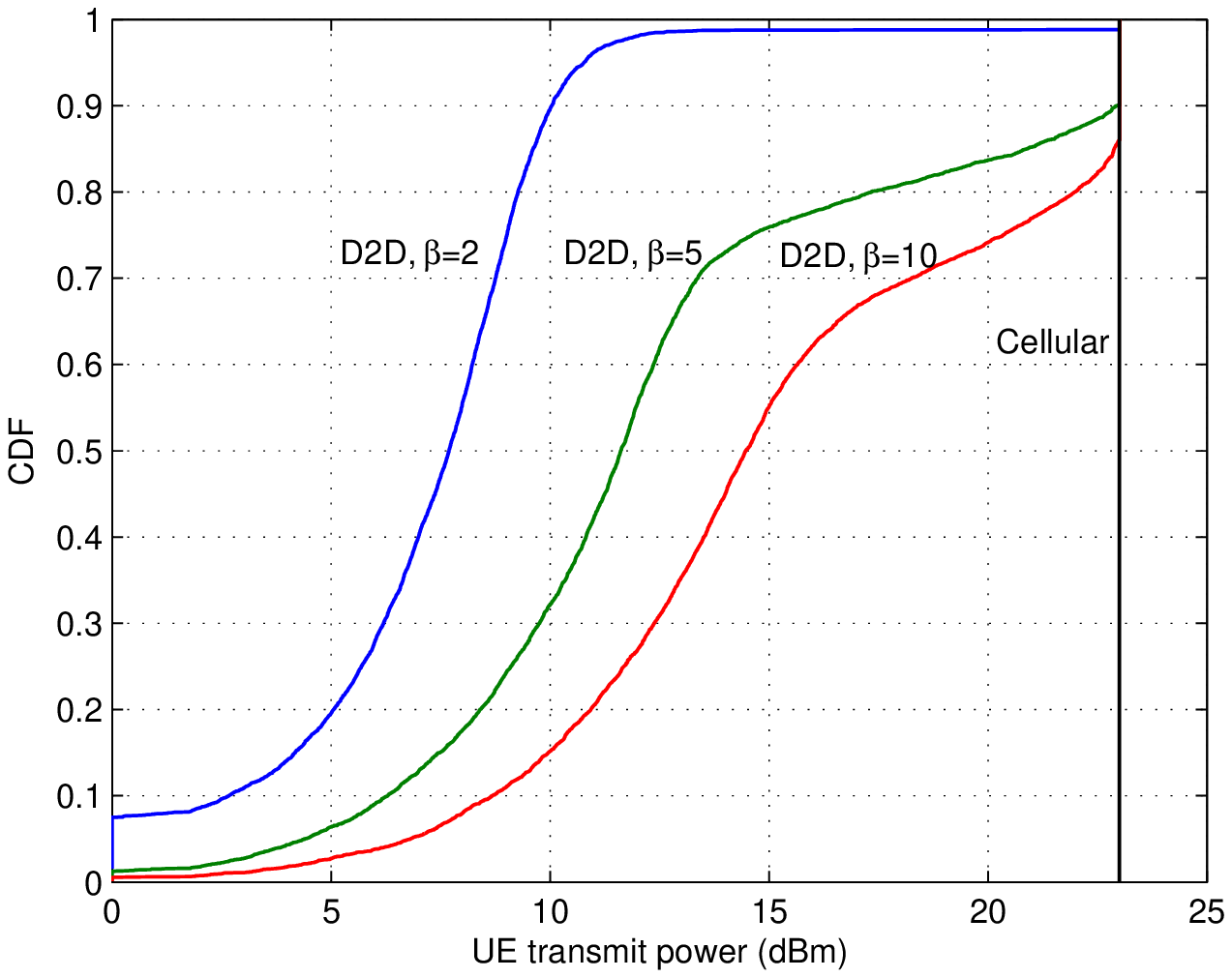}
\caption{UE power distribution under different $\beta$.}
\label{fig:cdf_power} \vskip -0.5em
\end{figure}

\begin{figure}[!t]
\centering
\includegraphics[height=3.5in]{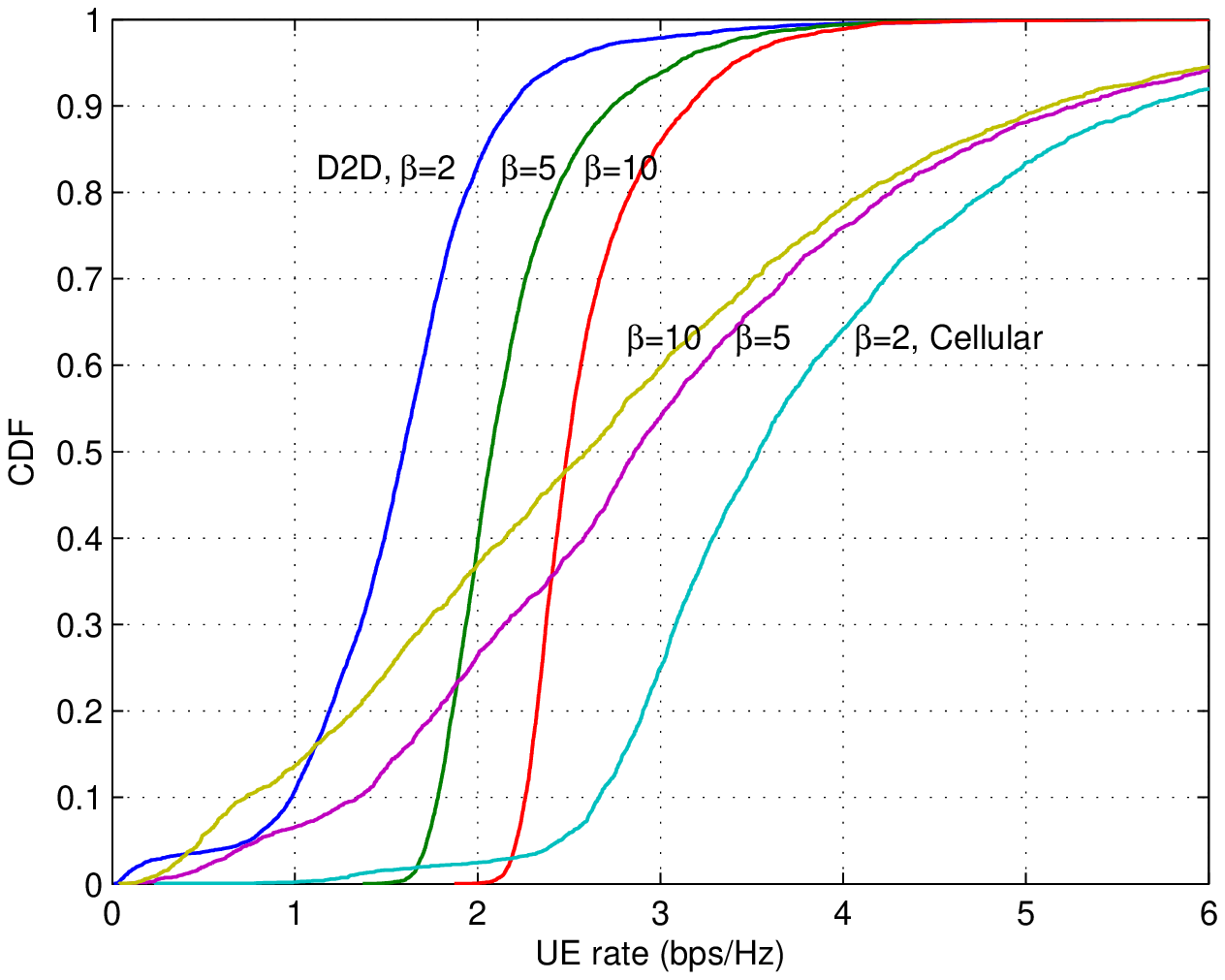}
\caption{UE rate distribution under different $\beta$.}
\label{fig:cdf_rate_beta}
\end{figure}

\begin{figure}[!t]
\centering
\includegraphics[height=3.4in]{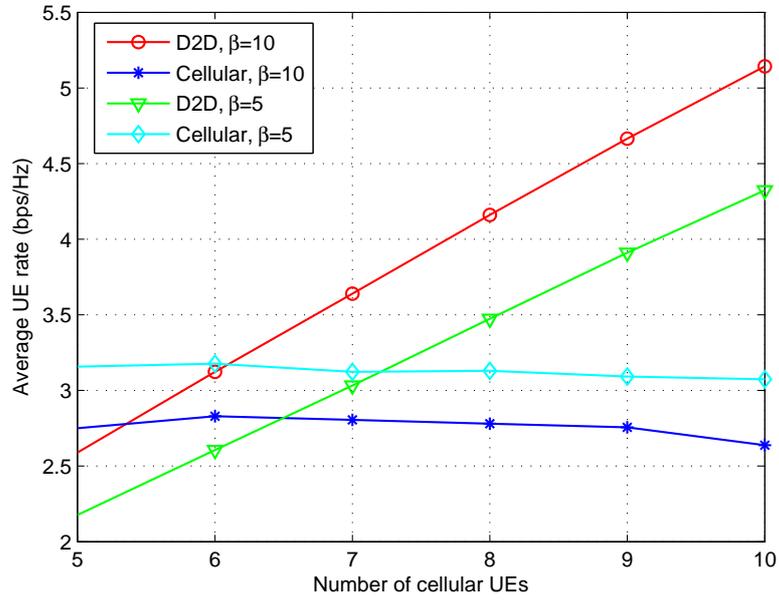}
\caption{Cellular and D2D average UE rate with different number of
cellular UEs.} \label{fig:rate}
\end{figure}

\begin{figure}[!t]
\centering
\includegraphics[height=3.3in]{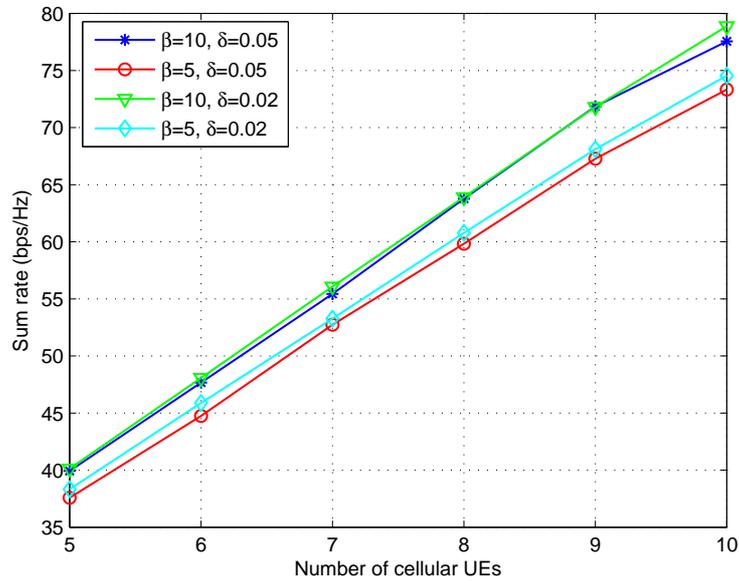}
\caption{System sum rate with different numbers of cellular UEs.}
\label{fig:rate_number}
\end{figure}

In Fig.~\ref{fig:cdf_power} and Fig.~\ref{fig:cdf_rate_beta}, we plot CDF of UE transmit power and UE rate under different scale factor $\beta$, respectively. $\beta$ is the ratio of the leader's gain and the follower's payment. For a larger $\beta$, the payment for the follower is relatively lower, and thus, the follower will choose larger transmit power. This is illustrated in Fig.~\ref{fig:cdf_power}. With larger transmit power of D2D UEs, the rate of D2D UEs is larger. In the meanwhile, D2D UEs cause more interference to cellular UEs, causing a decrease in the rate of cellular UEs. This effect is illustrated in Fig.~\ref{fig:cdf_rate_beta}. We can also observe that D2D UEs use a much smaller transmit power than the cellular UEs.

In Fig.~\ref{fig:rate} and Fig.~\ref{fig:rate_number}, we plot average rate of UEs and system sum rate with different numbers of cellular UEs, respectively. Fig.~\ref{fig:rate} shows that when the number of cellular UEs increases, the rate performance of D2D UEs is improved. This is due to the fact that D2D UEs has more resources to use. We observe that when $C=10$, the rate performance of D2D UEs is much better than that of cellular UEs. 
The effect of scale factor $\beta$ and fairness coefficient $\delta$ is shown clearly in the two figures. It can be seen that the interference is properly managed and the cellular UEs achieve a reasonable rate performance. Besides, as there are $C$ resources and $D$ D2D UEs, the average transmission time of D2D UEs is $D/C$ of that of cellular UEs. In Fig.~\ref{fig:rate}, we observe that D2D UEs achieve a similar or higher rate with cellular UEs. The D2D communication obviously has higher efficiency.

In this section, we developed a Stackelberg game framework for joint power control, channel allocation and scheduling of D2D communication. We analyzed the optimal strategy for the constructed game, and proposed an algorithm to allocate resources and schedule D2D UEs. Throughput, interference management and fairness of the system were considered. Simulation results show that the proposed algorithm can achieve a good throughput performance for both the cellular and the D2D UEs. The D2D UEs can be fairly served. The scale factor $\beta$ and fairness coefficient $\delta$ have an important effect on the performance of the algorithm.
It is also shown that D2D communication can improve the throughput of the system.

%



\section{Energy Efficient Improvement}\label{sec:energy}
Device-to-device (D2D) communication as an underlay to cellular
networks brings significant benefits to users' throughput and
battery lifetime. The allocation of power and channel resources to
D2D communication needs elaborate coordination, as D2D user
equipments (UEs) cause interference to other UEs. In this section, we
propose a novel resource allocation scheme to improve the
performance of D2D communication. Battery lifetime is explicitly
considered as our optimization goal. We first formulate the
allocation problem as a non-cooperative resource allocation game in
which D2D UEs are viewed as players competing for channel resources.
Then, we add pricing to the game in order to improve the efficacy,
and propose an efficient auction algorithm. We also perform
simulations to prove efficacy of the proposed algorithm.

\subsection{Model Assumptions and Battery lifetime}
We also consider a single cell environment with multiple UEs and one eNB
located at the center of the cell. Both the eNB and the UEs are
equipped with a single omni-directional antenna. The system consists
of two types of UEs, cellular UEs and D2D UEs. The D2D UEs are in
pairs, each including one transmitter and one receiver. The number
of cellular UEs and D2D UEs is $C$ and $D(D<C)$, respectively.
There are $C$ orthogonal channels, which is occupied by the corresponding cellular UEs. The channels allocated to cellular UEs are assumed to be fixed. A D2D pair can reuse RBs of one or multiple cellular users, and multiple D2D pairs can share the same RBs.

\begin{figure}[!t]
\centering
\includegraphics[height=3.7in]{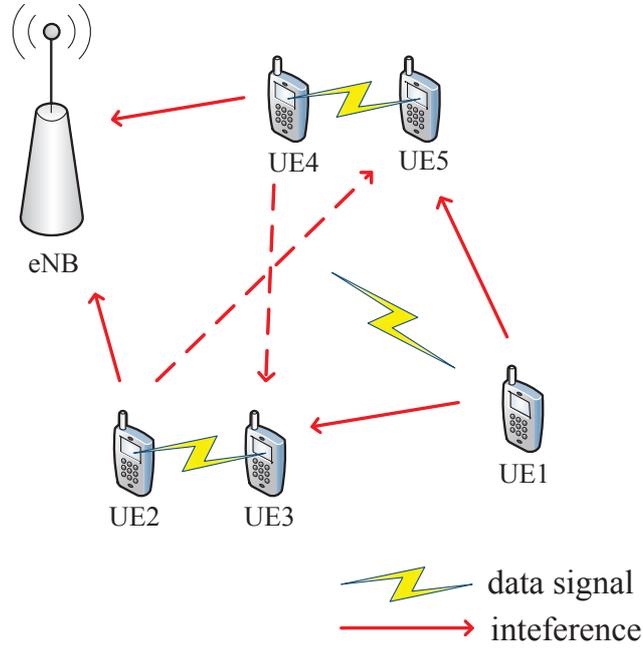}
\caption{System model of D2D underlay communication with uplink resource sharing. UE$_1$ is a cellular user whereas UE$_2$ and UE$_3$, UE$_4$ and UE$_5$ are in D2D communication.}
\label{fig:system_model_ICC}
\end{figure}

D2D communications during the downlink is infeasible
\cite{Wang2013}, and thus, we focus on the uplink period. A scenario
of uplink resource sharing is illustrated in \figurename{
}\ref{fig:system_model_ICC} where one cellular UE (UE$_1$) and two
D2D pairs (UE$_2$ and UE$_3$, UE$_4$ and UE$_5$) are sharing the
same radio resource. UE$_2$ and UE$_4$ are transmitters while UE$_3$
and UE$_5$ are receivers. The two D2D UEs in a pair are close enough
to satisfy the maximum distance constraint of D2D communication, in
order to guarantee the quality of D2D services. The cellular UE
(UE$_1$) transmits data to the eNB, while the eNB suffers
interference from D2D transmitters (UE$_2$ and UE$_4$). Also, the
D2D pairs are in communication while the receivers (UE$_3$, UE$_5$)
are exposed to interference from the cellular user (UE$_1$) and the
other D2D transmitters (UE$_4$, UE$_2$, respectively).

The SINR at the $i$-th D2D receiver's on channel $c$ can be expressed as
\begin{equation}
  \gamma_i^c=\dfrac{p_i^c g_{ii}}{p^c_0 g_{ci}+\sum_{j\neq i} p_j^c g_{ji}+N_0},
\end{equation}
where $p^c_0$ and $p_j^c$ is the transmit power of the $c$-th cellular UE and the $j$-th D2D transmitter on channel $c$, respectively. $g_{ci}$ denotes the channel gain between the $c$-th cellular UE and $i$-th D2D receiver. $g_{ji}$ denotes the channel gain between the $j$-th D2D transmitter and the $i$-th D2D receiver. $N_0$ is the noise power. Note that if $p_i^c=0,\exists i,c$, it means D2D pair $i$ does not reuse the resources of the $c$-th channel. The channel rate of the D2D pair $i$ on the $c$-th channel is
\begin{equation}
  r_i^c=\log_2(1+\gamma_i^c).
\end{equation}

In this section, the battery lifetime is considered as our optimization objective. We focus on the scheduling of D2D communication while the cellular network works in a standard way. The transmit power of cellular users is assumed to be fixed. We want to extend the battery lifetime while achieving a reasonable target rate for every D2D pair.

The energy consumption of D2D users includes two parts, the
transmission energy and circuit energy. The circuit energy
consumption cannot be ignored since it has an important effect on
the battery lifetime. It is the energy consumed by all the circuit
blocks along the signal path \cite{Cui2004}. We define $P_c$ as the
total power consumption of these circuit blocks. Without loss of
generality, we assume all D2D UEs have the same constant circuit
power consumption. Thus, we focus on D2D transmitters. Each D2D
transmitter can distribute its transmit power into $C$ channels. The
total power consumption of the $i$-th D2D transmitter is
$P_i=\sum_{c=1}^C p_i^c+P_c$. According to Peukert's law
\cite{Rao2003}, the battery lifetime $L$ can be approximated by
\begin{equation}
  L=\frac{Q}{I^b},
\end{equation}
where $Q$ is the battery capacity. $I$ is the discharge current. $b$ is a constant around 1.3. We denote the average power consumption of the $i$-th D2D transmitter by $\mathbb{E}[P_i]$.
With an operating voltage $V_0$, the battery lifetime $L_i$ of D2D transmitter $i$ is
\begin{equation}
  L_i=\frac{Q}{(\mathbb{E}[P_i]/V_0)^b}.
\end{equation}

Each D2D pair requires rate $R$.
Our objective is to maximize the total battery lifetime under the rate constraints, which can be expressed as
\begin{equation}
\begin{split}
  \max \sum_{i=1}^D L_i, \quad
  \mathrm{s.t.} \sum_{c=1}^C r_i^c\geq R, \forall i;
   \; p_i^c\geq 0, \forall i, c.
\end{split}
\label{eq:problem}
\end{equation}
Since power is an increasing function of rate, lifetime is maximized when $\sum_{c=1}^C r_i^c = R$. The problem in (\ref{eq:problem}) is complicated to solve and hard for distributed implementation. Therefore, we develop an alternative game-theoretic approach.

\subsection{Resource Allocation Game}\label{sec:game}


Consider D2D transmitters as players. The players are self-interested and each player wants to maximize its own battery lifetime.
In order to achieve this, the best strategy for each player is to minimize its own transmit power at any given time, regardless of other players. Define power vector $\mathbf{p}_i=(p_i^1,p_i^2,...,p_i^C)$ as the transmit power of D2D transmitter $i$ on each channel, which is also seen as player $i$'s strategy. The utility function $u_i$ can be defined as the negative of the total transmit power of the $i$-th D2D transmitter, i.e.,
\begin{equation}
  u_i(\mathbf{p}_i,\mathbf{p}_{-i})=-\sum_{c=1}^C p_i^c,
\end{equation}
where $\mathbf{p}_{-i}$ is the strategy of other players. By taking proper strategy $\mathbf{p}_{i}$, each user wants to maximize its utility.
We begin by characterizing the best response for any player $i$. Define $x^+=\max(x,0), \mathbf{x}^+=(x_1^+, x_2^+,...,x_n^+), \mathbf{q}_i=g_{ii}^{-1}(\mathbf{p}_0 g_{ci}+\sum_{j\neq i}\mathbf{p}_j g_{ji}+N_0)$.
We assume $q_i^1 \leq q_i^2 \leq \cdots \leq q_i^C$. It is obvious that $q_i^c\geq 0, \forall c$.

\begin{proposition}
Given other players' strategies, the best response of player $i$, i.e., the strategy that maximizes its utility is
\begin{equation}
  \mathbf{p}_{i}^*=\left[\left(2^{R} \prod_{c=1}^k q_i^c \right)^{1/k} -\mathbf{q}_i \right]^+,
\label{eq:pi}
\end{equation}
where $k=\arg\min_{k} \sum_{c=1}^k p_i^c, k\in\{1,2,...,C\}.$
\label{prop:strategy}
\end{proposition}

\begin{proof}
Consider the Lagrangian
\begin{equation}
  \mathcal{L}(\lambda,\mathbf{p}_i)=-\sum_{c=1}^C p_i^c + \lambda\left(\sum_{c=1}^C r_i^c-R\right),
\end{equation}
where $\lambda$ is the Lagrange multiplier.
Solving the Kuhn-Tucker condition $\frac{\partial \mathcal{L}}{\partial p_i^c}=0, \forall c$ and
considering $p_i^c\geq 0,\forall c$, we obtain
\begin{equation}
  {p}_{i}^c=\left(\frac{\lambda}{\ln 2}-q_i^c\right)^+, \forall c.
\label{eq:pi2}
\end{equation}
From the Kuhn-Tucker condition, we know $\lambda \geq 0$. Let us assume $0\leq q_i^0\leq \cdots \leq q_i^k \leq \lambda \leq q_i^{k+1} \leq\cdots\leq q_i^C,1\leq k\leq C$. Considering the rate constraint, we derive
\begin{equation}
  \lambda=\ln 2 \left(2^{R} \prod_{c=1}^k q_i^c \right)^{1/k}.
\label{eq:lambda}
\end{equation}
Substituting (\ref{eq:lambda}) into (\ref{eq:pi2}), we get (\ref{eq:pi}). There are $C$ solutions for $k$, and thus, we get $C$ strategies. To maximize the utility function, $k$ is searched in $\{1,2,...,C\}$.
\end{proof}

From the proof, we know that given other players' strategies, the best response of player $i$ is to allocate its transmit power into the $k$ channels with the best channel qualities. Here, the channel quality refers to $q_i^c$, since the SINR of player $i$ on channel $c$ is given by $\gamma_i^c=p_i^c/q_i^c$, and for lower $q_i^c$, the channel quality is better. $k$ is determined by the channel qualities.

\begin{definition}
A set of strategies $\mathbf{p}$ for each player is a Nash equilibrium if no unilateral deviation in strategy by any single player is profitable for that player, i.e.,
\begin{equation}
  u_i(\mathbf{p}_i,\mathbf{p}_{-i})\geq u_i(\mathbf{p}_i',\mathbf{p}_{-i}),\forall i.
\end{equation}
\end{definition}

The Nash equilibrium offers a stable outcome of a game where multiple players with conflicting interests compete through self-optimization and reach a point where no player wishes to deviate. 

\begin{proposition}
A Nash equilibrium exists in the proposed game.
\end{proposition}
\begin{proof}
A Nash equilibrium exists \cite{Osborne1994}, if $\forall i$
\begin{enumerate}
  \item the set of strategies is a nonempty compact convex subset of a Euclidean space;
  \item the utility function is continuous and quasi-concave.
\end{enumerate}

Since rate is a $\log$-increasing function of transmit power $p$, to achieve a rate target $R$, player $i$ has an upper bound of its transmit power $p_i^c \leq p_{max}, \forall c$. The set of player $i$'s strategies is $\mathcal{P}_i=\{\mathbf{p}_i|0 \leq p_i^c \leq p_{max}, \forall c\}$, which is nonempty compact convex subset of Euclidean space $\mathbb{R}^C$. For any two strategies $\mathbf{p}_i,\mathbf{p}_i' \in \mathcal{P}_i$, 
we have $\theta \mathbf{p}_i + (1-\theta) \mathbf{p}_i' \in \mathcal{P}_i, \forall \theta\in[0,1]$. Thus, $\mathcal{P}_i$ is a convex set.
$u_i$ is obviously continuous. Besides, we have $u_i(\theta \mathbf{p}_i+ (1-\theta) \mathbf{p}_i')=-\sum_{c=1}^C[\theta p_i^c+(1-\theta) p_i'^{c}]
\geq \min \big(u_i(\mathbf{p}_i), u_i(\mathbf{p}_i') \big), \forall \mathbf{p}_i, \mathbf{p}_i' \in \mathcal{P}_i, \forall \theta\in [0,1]$. Thus, $u_i$ is quasi-concave.
\end{proof}

The proposition establishes the existence of a Nash equilibrium of the game, which guarantees the feasibility of the resource allocation game. We are also concerned about the efficiency of the game.

\begin{definition}
A strategy $\mathbf{\hat p}$ is Pareto optimal (efficient) if there exists no other strategy $\mathbf{p}$ such that $u_i(\mathbf{\hat p}) \geq u_i(\mathbf{p})$ for all $i$ and $u_i(\mathbf{\hat p})>u_i(\mathbf{p})$ for some $i$.
\end{definition}
\begin{proposition}
Pareto optimal $\mathbf{\hat p}$ of the resource allocation game must be a Nash equilibrium, i.e., $\mathbf{\hat p}=\mathbf{p}^*$.
\end{proposition}
\begin{proof}
We prove the proposition by contradiction.
Assume $\mathbf{\hat p}_i \neq \mathbf{p}_i^*, \exists i$. Since the utility function $u_i(\mathbf{p}_i,\mathbf{p}_{-i})=-\sum_c p_i^c$ is maximized when $\mathbf{p}_i=\mathbf{p}_i^*$ given $\mathbf{p}_{-i}$, there exists $\ell$ such that $\hat{p}_i^\ell>{p}_i^{*\ell}$. Thus, if player $i$ sets ${p}_i'^{\ell}=p_i^{*\ell}$ and keeps other components unchanged, player $i$ will get larger utility $u'_i>\hat{u}_i$. Moreover, For all $j\neq i$, we have $q_j'^{\ell}<\hat{q}_j^{\ell}$ and $r_j'^{\ell} = \log_2(1+p_j'^\ell/q_j'^\ell) > \hat{r}_j^\ell$. Thus, $\sum_c r_j'^{c}>R$, which indicates the other players' rate constraints are not violated. We still have $u'_j=\hat{u}_j,\forall j\neq i$. This is a contradiction with that $\mathbf{\hat p}$ is Pareto optimal. Therefore, we have $\mathbf{\hat p}_i = \mathbf{p}_i^*$.
\end{proof}

The game constructed above is a general case. In practical scenarios, there may be resource sharing constraints of D2D communication. We address different resource sharing modes based on the following criteria:
\begin{enumerate}
  \item Multiple D2D pairs sharing the same channel;
  \item one D2D pair reusing multiple channels.
\end{enumerate}

For the first criterion, if one channel is only allowed to be reused by no more than one D2D pair, D2D UEs do not receive interference from other D2D UEs. Player $i$'s strategies will be constrained to allocating power to channels that satisfies $p_j^c=0, \forall j\neq i$. The channel qualities are better due to less interference. However, it has the disadvantage of low spectrum
efficiency because of the low frequency reuse.

For the second criterion, if one D2D pair is only allowed to reuse one of the $C$ channels, player $i$ will always choose to reuse the channel of the best quality, i.e., with the lowest $q_i^c$, as analyzed above. Thus, the best response for player $i$ is $p_i^c=q_i^c(2^{R}-1),  c=\arg\min q_i^c$ and $p_i^c=0$ for other $c$.

\subsection{Resource Auction Algorithm}\label{sec:auction}
In the above game, each player competes to maximize its own utility by adjusting transmit power on each channel. However, it ignores the cost it imposes on the cellular UEs and other D2D UEs by causing interference to them. In other words, the behavior of the players has an externality for the system. Therefore, we consider pricing for the channel resources, which is an effective way to deal with externality.

Pricing for the channel resources encourages D2D pairs to reuse the channels more efficiently. The cost is larger if a D2D pair reuse multiple channel resources or multiple D2D pairs sharing the same resources. Here, we consider linear pricing, which is widely used due to its simplicity and efficiency. We redefine the utility function of player $i$ for channel $c$ as
\begin{equation}
  u_i^c = -p_i^c - \beta(m_c+n_i),
\label{eq:utility}
\end{equation}
where $\beta$ is the unit price, $m_c$ denotes the number of players occupying channel $c$ and $n_i$ denotes the number of channels player $i$ reuses.
With pricing, the resource allocation problem can be efficiently solved by auction. Auction is effective in that the bidding and pricing is a guidance to lead the player to consider the real payoffs.

\begin{figure}[!t]
\centering
\small
\algsetup{indent=2em}
\hrule
\vskip 0.4em
Algorithm 1. Resource Auction Algorithm
\vskip 0.3em
\hrule
\vskip 0.4em
\begin{algorithmic}[1]
\STATE Set $\beta$, $R$; $\mathbf{P}_{i,j}\leftarrow 0, \forall i,j$; $m_c\leftarrow 0,\forall c$; $n_i \leftarrow 0, \forall i$;
\WHILE{$n_i=0, \exists i$}
    \FOR{$i=1$ \TO $D$}
    \IF{$n_i=1$}
        \STATE 
        $p_i^{c} \leftarrow q_i^c(2^R-1), \forall c$;
    \ELSE
        \STATE
        Evaluate all channels as its $(n_i+1)$-th channel;
    \ENDIF
    \STATE $u_i^c \leftarrow p_i^c + \beta(m_c+n_i), \forall c$;
    \IF{$\mathbf{P}_{i,c} \neq 0$}
        \STATE $u_i^c\leftarrow\infty$;
    \ENDIF
    \ENDFOR
    \STATE $(\tilde{i},\tilde{c}) \leftarrow \arg\max u_i^c$;
    \IF{$m_{\tilde{c}} \neq 0$ \OR $n_{\tilde{i}} \neq 0$}
        \STATE Adjust the power according to (\ref{eq:multiple});
    \ENDIF
    \STATE $\mathbf{P}_{\tilde{i},\tilde{c}} \leftarrow p_{\tilde{i}}^{\tilde{c}}$; $\mathbf{P}_{i',j'} \leftarrow p_{i'}^{j'}$ for adjusted power;
    \STATE $m_{\tilde{c}} \leftarrow m_{\tilde{c}} + 1, n_{\tilde{i}} \leftarrow n_{\tilde{i}} + 1$;
\ENDWHILE
\end{algorithmic}
\hrule
\end{figure}

Based on (\ref{eq:utility}), we design the mechanism of resource allocation auction which iteratively decides the transmit power. The channels are auctioned one by one. During the auction, players need to bid for all channels to compete for resources.
In the first round, each player calculates its best transmit power and corresponding utility. Since this is the first resource for all players, they will allocate enough power to reach the rate constraint, i.e., $p_i^c=q_i^c(2^R-1),\forall i,c$.
The player and channel with the highest utility $(i,c)=\arg\max u_i^c$ wins the auction and channel ${c}$ is sold to player ${i}$. Then, the auction moves to the next round. In this round, player $i$ can compete for a second resource. It can distribute its rate into two channels by evaluating all channels as its second channel. The transmit power is calculated according to (\ref{eq:pi}). $c$ can also be auctioned once more. But the cost of player $i$ for any resource and the cost of any player for channel $c$ is increased. Besides, the utility of player $i$ for $c$ is set to infinity, to prevent the resource from being sold to the same player twice. If $c$ is sold to another player $j$, player $i$ needs to readjust its transmit power. In general cases, if multiple players share the same channel $c$, and a new player comes in, the transmit power is adjusted in order to keep the former players' rate on channel $c$ unchanged, which can be obtained by solving the linear equations
\begin{equation}
  \frac{g_{ii}}{2^{r_i^c}-1}p_i^c - \sum_{j\neq i}g_{ji} p_j^c - p_0^c g_{ci} - N_0 = 0, \forall i \text{ sharing } c.
  \label{eq:multiple}
\end{equation}
$r_i^c$ is the original rate for the former players, or the target rate for the newcomer. Similarly, if player $i$ obtains another channel, all players adjust the power on $c$.
The auction repeats the above steps until all the players get at least one channel. The algorithm is summarized in Algorithm 1.

To evaluate the performance of the proposed algorithm, we perform several simulations.
We consider a single circular cell environment. The cellular UEs and D2D pairs are uniformly distributed in the cell. The maximum distance constraint of D2D communication is satisfied.
The received signal power is $P_i=P_j d_{ij}^{-2} |h_{ij}|^2$, where $P_i$ and $P_j$ are received power and transmit power, respectively. $d_{ij}$ is the distance between the transmitter and the receiver. $h_{ij}$ represents the complex Gaussian channel coefficient that satisfies $h_{ij}\sim\mathcal{CN}(0,1)$.
Simulation parameters are summarized in TABLE \ref{table:simulation}.

\begin{table}[!t]
\renewcommand{\arraystretch}{1.3}
\caption{Simulation Parameters and Values}
\label{table:simulation}
\centering
\begin{tabular}{p{1.6in}p{1.5in}}
\hline
Parameter & Values\\
\hline
Cell radius & 350m\\
Number of cellular UEs & 8\\
Number of D2D pairs & 3\\
Max D2D communication distance & 30m\\
Cellular UE Tx power & 250mW (24dBm)\\
Thermal Noise power & 1e-7W (-40dBm)\\
Circuit power consumption & 100mW (20dBm)\\
Battery capacity & 800mA$\cdot$h\\
Battery operating voltage & 4V\\
\hline
\end{tabular}
\end{table}

We plot average D2D battery lifetime and cellular rate under
different conditions. In each figure, we compare our algorithm with
random allocation, which allocates resources to D2D pairs randomly,
and centralized allocation, which maximizes the overall lifetime
rather than the individuals. The figures show that our algorithm
performs close to the centralized scheme and is much superior to
random allocation.

\begin{figure}[!t]
\centering
\includegraphics[height=3.5in]{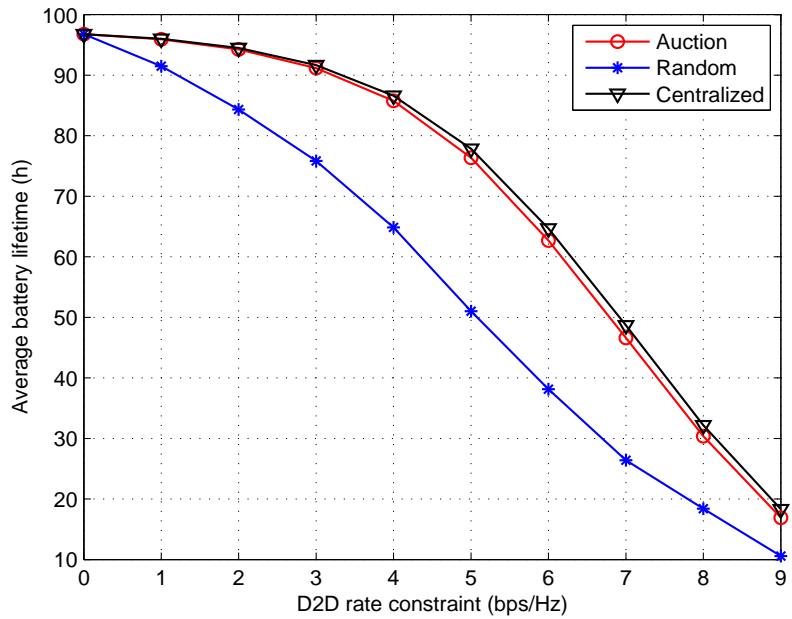}
\caption{Average D2D battery lifetime with different rate
constraints.} \label{fig:lifetime_r}
\end{figure}

\begin{figure}[!t]
\centering
\includegraphics[height=3.5in]{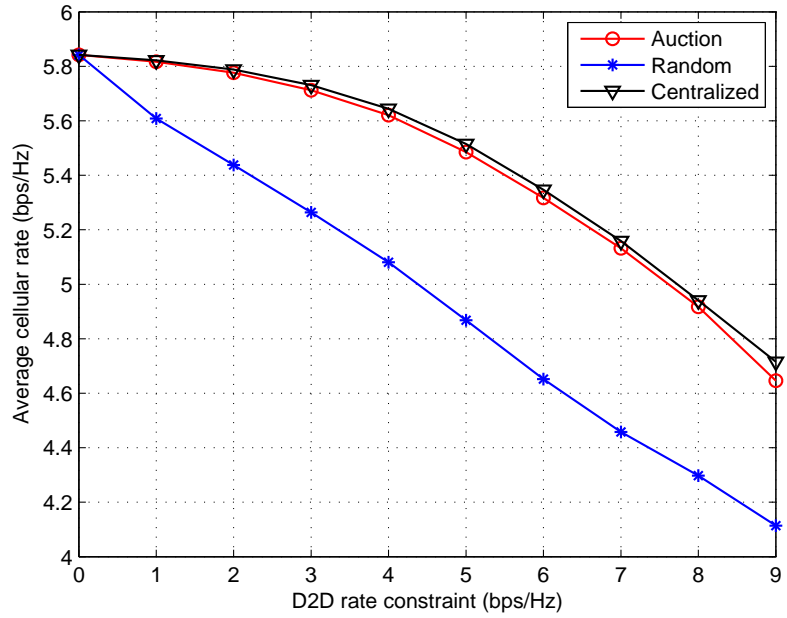}
\caption{Average cellular rate with different D2D rate constraints.}
\label{fig:rate_r}
\end{figure}

\begin{figure}[!t]
\centering
\includegraphics[height=3.4in]{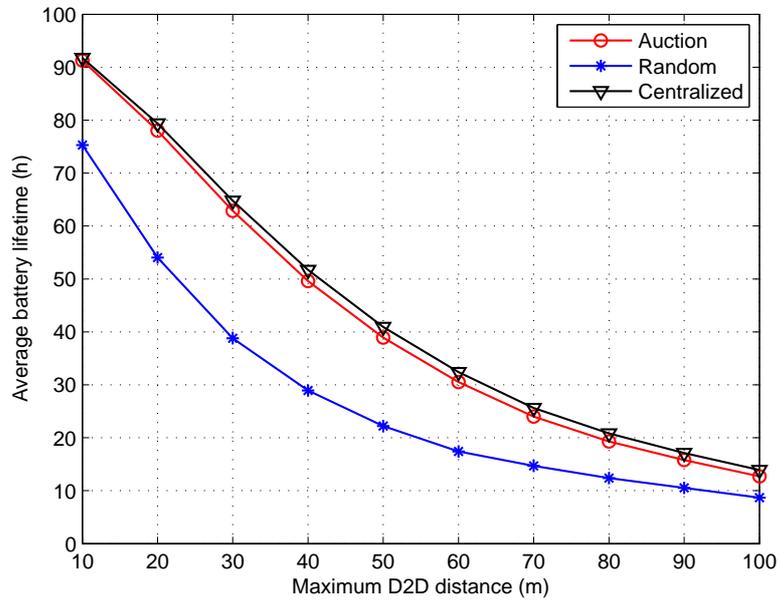}
\caption{Average D2D battery lifetime with different maximum D2D
distances.} \label{fig:lifetime_d}
\end{figure}

\begin{figure}[!t]
\centering
\includegraphics[height=3.5in]{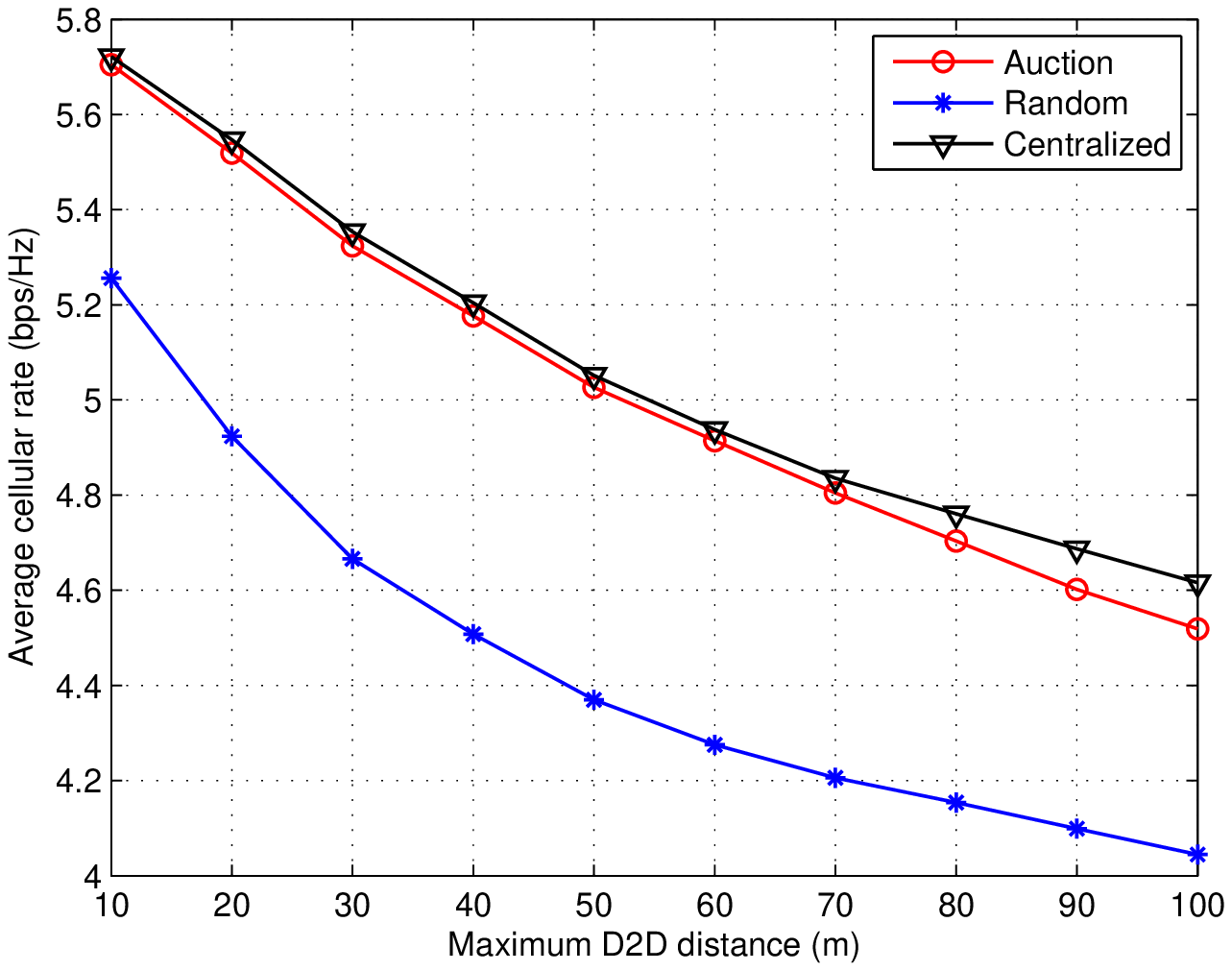}
\caption{Average cellular rate with different maximum D2D
distances.} \label{fig:rate_d}
\end{figure}

In Fig. \ref{fig:lifetime_r} and Fig. \ref{fig:rate_r}, we plot the effect of rate constraints. In Fig. \ref{fig:lifetime_r}, average D2D battery lifetime goes down with higher rate constraints. This also causes more interference to the cellular network, resulting in the degrade in cellular rate, which is shown in Fig. \ref{fig:rate_r}. It is also important to note that without D2D communication, the average rate of cellular UE is about 5.8bps/Hz. The average battery lifetime of cellular UEs under such a rate is $L_c=Q/((p_0+P_c)/V_0)^b=0.8/(0.35/4)^{1.3} \approx 9$h. From Fig. \ref{fig:lifetime_r}, we can see that to reach a rate constraint of 6bps/Hz, the average battery lifetime of D2D UEs is approximately 63h, which is about 7 times of the battery lifetime of cellular UEs. D2D communication fully utilizes the proximity between D2D UEs, and thus greatly extends battery lifetime.



Fig. \ref{fig:lifetime_d} and Fig. \ref{fig:rate_d} show the
influence of maximum D2D communication distances. In D2D underlaying
cellular networks, the maximum communication distance between D2D
UEs is a critical parameter. It is a criterion for the eNB to decide
whether to set up a direct link between the two users or communicate
in ordinary cellular mode. Another important aspect is that D2D
needs more transmit power for longer distance, resulting in shorter
battery lifetime as well as causing more interference to the
cellular network. From the figures, we can see that the distance has
an important influence on the battery lifetime and cellular rate.
Battery lifetime and cellular rate goes down quickly with a larger
D2D distance.

In this section, we investigated resource allocation for
device-to-device communication underlaying cellular networks, in
order to extend UE battery lifetime. The proposed resource
allocation game was analyzed to have a Nash equilibrium that is
Pareto efficient. We added pricing to the game to deal with
externality, and proposed an auction-based resource allocation
algorithm. The simulation results indicate that D2D communication
can greatly extend UE battery lifetime compared with traditional
cellular communication. The results also show that the proposed
algorithm performs close to the centralized scheme, and much better
than the random allocation.

\chapter{Summary}\label{summary}
In this book, Device-to-Device communication underlaying cellular networks has been studied. Some physical-layer techniques and cross-layer optimization methods on resource management and interference avoidance have been proposed and discussed. WINNER II channel models \cite{WINNER} has been applied to be the signal and interference model and simulation results show that the performance of D2D link is closely related to the distance between D2D transmitter and receiver and that between interference source and the receiver. Besides, by power control, D2D SINR has degraded, which will naturally contribute to low interference to cellular communication. A simple mode selection method of D2D communication has been introduced. Based on path-loss (PL) mode selection criterion, D2D gives better performance than traditional cellular system. When D2D pair is farther away from the BS, a better results can be obtained. Game theory, which offers a wide variety of analytical tools to study the complex interactions of players and predict their choices, can be used for power and radio resource management in D2D communication.

A distributed threshold-based power control scheme has been proposed to guarantee the feasibility of D2D connection, and at the same time limit cellular SINR degradation. Power is calculated by D2D transmitter itself, which makes the operation flexible and convenient, improving the system efficiency. Furthermore, a joint beamforming and power control scheme that aims to maximize the system sum rate while guarantees the performance of both cellular and D2D connections has been given. The BS carries out beamforming to avoid D2D from excessive interference, and D2D transmit power is calculated by the BS based on maximizing the system sum rate. Also, the BS decides whether the calculated D2D transmit power available according to SINR threshold of cellular and D2D links.

In the following chapters, game theory was applied to solve resource management problem and cross-layer optimization problem. A reverse iterative combinatorial auction has been formulated as a mechanism to allocate the spectrum resources for D2D communications with multiple user pairs sharing the same channel. In addition, a game theoretic approach has been developed to implement joint scheduling, power control and channel allocation for D2D communication. Finally, joint power and spectrum resource allocation method has been studied under consideration of battery lifetime, which is an important application of D2D communication on increasing UE's energy efficiency. The simulation results show that all these methods have beneficial effects on improving the system performance.

In fact, there still exist numerous challenging problems such as simplification of resource allocation algorithms, multi-cell joint optimization of system performance, multi-hop transmission in D2D communications and other optimization problems, all of which wait to be investigated in the future work. The solutions will bring large improvements on the performance of D2D underlaying systems, which, indeed plays a good role in the next generation wireless communication network.

%


%

\chapter*{List of Acronyms}
\begin{acronym}
\acro{3GPP}{Third Generation Partnership Project}
\acro{LTE}{Long Term Evolution}
\acro{LTE-A}{Long Term Evolution-Advanced}
\acro{BS}{Base Station}
\acro{MS}{Mobile Station}
\acro{D2D}{Device-to-Device}
\acro{UEs}{User Equipments}
\acro{CSI}{Channel State Information}
\acro{UL}{Uplink}
\acro{DL}{Downlink}
\acro{OFPC}{Open Loop Fraction Power Control}
\acro{PL}{Path Loss}
\acro{CAs}{Combinatorial Auctions}
\acro{WDP}{Winner Determination Problem}
\acro{I-CAs}{Iterative Combinatorial Auctions}
\acro{HARQ}{Hybrid Automatic Repeat Request}
\acro{eNB}{Evolved Node Base Station}
\acro{SINR}{Signal to Interference plus Noise Radio}
\acro{AWGN}{Additive White Gaussian Noise}
\acro{PC}{Power Control}
\acro{SLNR}{Signal to Leakage plus Noise Radio}
\acro{BF}{Beamforming}
\acro{RBs}{Resource Blocks}
\acro{CAP}{Combinatorial Allocation Problem}
\acro{Max C/I}{Maximum Carrier to Interference}
\acro{PF}{Proportional Fair}
\acro{TTI}{Transmission Time Interval}
\acro{CDF}{Cumulative Distribution Function}

\end{acronym}
\end{document}